\documentclass{lmcs} %%% last changed 2014-08-20

\keywords{Differential Logical Relations, Quantales, Quasi-Metrics, Partial Metrics}

\usepackage{hyperref}
\usepackage{enumitem}
\usepackage{amsmath}
\usepackage{amssymb}
\usepackage{mathabx}
\usepackage{bbm}
\usepackage{cmll}
\usepackage{mathtools}
\usepackage{prftree}
\usepackage{stmaryrd}
\SetSymbolFont{stmry}{bold}{U}{stmry}{m}{n}
\usepackage{tikz}

\newcommand{\real}{\mathsf{Real}}
\newcommand{\fst}{\mathsf{fst}}
\newcommand{\snd}{\mathsf{snd}}
\newcommand{\lam}[3]{\lambda #1\,{:}\,#2.\,#3}
\newcommand{\pair}[2]{\langle #1,#2 \rangle}
\newcommand{\type}{\mathbf{Type}}
\newcommand{\LL}{\Lambda_{\real}}
\newcommand{\prt}[1]{#1^{\bullet}}
\newcommand{\term}[1]{\mathbf{Term}_{#1}}
\newcommand{\psem}[1]{\llparenthesis\hspace{0.5pt} #1\hspace{0.5pt}\rrparenthesis}
\newcommand{\sem}[1]{\llbracket #1 \rrbracket}
\newcommand{\semp}[2]{\llbracket #1 \rrbracket_{#2}}
\newcommand{\fsem}[1]{\scalebox{0.6}[1]{\{} \hspace{-1.9pt} | #1 | \hspace{-1.9pt} \scalebox{0.6}[1]{\}}}
\newcommand{\fsemp}[3]{\scalebox{0.6}[1]{\{} \hspace{-1.9pt} | #1 | \hspace{-1.9pt} \scalebox{0.6}[1]{\}}_{#2,\hspace{0.75pt}#3}}

\newcommand{\qqm}{\mathbf{Qqm}}
\newcommand{\add}{\mathbin{::}}
\newcommand{\qt}[1]{\mathcal{Q}_{#1}}
\newcommand{\uset}[1]{|#1|}
\newcommand{\dfn}[1]{\lVert #1 \rVert}

\newcommand{\itl}[1]{#1}
\newcommand{\lawsqsubseteq}{\geq}
\newcommand{\lawsqsupseteq}{\leq}
\newcommand{\lawsqsubsetneq}{\gneq}
\newcommand{\lawsqsupsetneq}{\lneq}

\newcommand{\lawsqcap}[2]{#1 \vee #2}

\newcommand{\lawsqcapparen}[2]{\left(#1 \vee #2\right)}

\newcommand{\lawbigsqcap}{\sup}
\newcommand{\lawotimes}{+}

\begin{document}

% If the title is longer than 55 characters, then
% specify a shorter running title as the optional
% argument to \title. The running title should be
% roughyl at most 55 characters:

\title[On The Metric Nature of (Differential)
Logical Relations] {On The Metric Nature
  \texorpdfstring{\\}{} of (Differential) Logical
  Relations}

% \titlecomment{{\lsuper*}OPTIONAL comment
%   concerning the title, \eg, if a variant or an
%   extended abstract of the paper has appeared
%   elsewhere.} \thanks{thanks, optional.} %optional

% affiliations are numbered automatically with a,
% b, c (see below) use the optional argument to
% indicate the affiliation(s) of each author omit
% the argument if there is only one author, or
% only one affiliation
\author[U.~{Dal Lago}]{Ugo {Dal
    Lago}\lmcsorcid{0000-0001-9200-070X}}[a,b]
\author[N.~Hoshino]{Naohiko
  Hoshino\lmcsorcid{0000-0003-2647-0310}}[c]
\author[P.~Pistone]{Paolo
  Pistone\lmcsorcid{0000-0003-4250-9051}}[d]

% affiliation 1 (automatically numbered a)
\address{University of Bologna, Italy} %optional
% write emails for all authors having that
% affiliation
\email{ugo.dallago@unibo.it} %optional

% affiliation 2 (automatically numbered b)
\address{INRIA Sophia Antipolis, France} %optional
% \email{name2@email2} %optional

% affiliation 3 (automatically numbered c)
\address{Sojo University, Japan} %optional
\email{nhoshino@cis.sojo-u.ac.jp} %optional

% affiliation 3 (automatically numbered d)
\address{Universit\'e Claude Bernard Lyon 1,
  France} %optional
\email{paolo.pistone@ens-lyon.fr} %optional

%% etc.

%% required for running head on odd and even
%% pages, use suitable abbreviations in case of
%% long titles and many authors:

%%%%%%%%%%%%%%%%%%%%%%%%%%%%%%%%%%%%%%%%%%%%%%%%%%%%%%%%%%%%%%%%%%%%

%% the abstract has to PRECEDE the
%% command \maketitle: be sure not to issue
%% the \maketitle command twice!

\begin{abstract}
  Differential logical relations is a method to measure distances between higher-order programs where distances between functional programs are themselves \emph{functions}, relating errors in inputs with errors in outputs. This way, differential logical relations provide a more fine-grained and contextual information of program distances. This paper aims to clarify the metric nature of differential logical relations. We introduce the notion of quasi-quasi-metrics and observe that the cartesian closed category of quasi-quasi-metric spaces reflects the construction of differential logical relations in the literature. The cartesian closed structure induces a fundamental lemma, which can be seen as a compositional reasoning principle for program distances. Furthermore, we investigate the quasi-quasi-metric spaces arising from the interpretation of types, and we prove that they satisfy variants of the strong transitivity condition and indistancy condition, as well as a weak form of the symmetry condition. In the last part of this paper, we introduce a notion of differential prelogical relations arising as a quantitative counterpart of the framework of prelogical relations. Roughly speaking, differential prelogical relations are quasi-quasi-metrics on the collection of programs. The poset of differential prelogical relations has the finest differential prelogical relation presented as a formal quantitative equational theory, while the poset does not necessarily have a coarsest differential prelogical relation. The absence of a coarsest differential prelogical relation contrasts with the situations of typed lambda calculi, where the contextual equivalences serve as the coarsest program equivalences.
\end{abstract}

\maketitle

% BEGIN introduction.tex

% !TEX root = main.tex

\section{Introduction}

Program equivalence is a crucial concept in program semantics, and ensures that different implementations of a program produce \emph{exactly} the same results under the same conditions, i.e., in any environment. This concept is fundamental in program verification, code optimization, and for enabling reliable refactoring: by proving that two programs are equivalent, developers and compiler designers can confidently replace one program with the other, knowing that the behavior and outcomes will remain consistent. In this respect, guaranteeing that the underlying notion of program equality is a \emph{congruence} is of paramount importance.

In the research communities mentioned above, however, it is known that comparing programs through a notion of equivalence without providing the possibility of measuring the \emph{distance} between non-equivalent programs makes it impossible to validate many interesting and useful program transformations \cite{Mittal2016}. All this has generated interest around the concepts of program metrics and more generally around the study of techniques through which to quantitatively compare non-equivalent programs, so as, e.g., to validate those program transformations which do not introduce too much of an error \cite{Reed2010, Plotk}.

What corresponds, in a quantitative context, to the concept of congruence? Once differences are measured by some (pseudo-)metric, a natural answer to this question is to require that any language construct does not increase distances, that is, that they are \emph{non-expansive}. Along with this, the standard properties of (pseudo-)metrics, like the triangle inequality $ d(x, z) + d(z, y) \geq d(x, y)$, provide general principles that are very useful in metric reasoning, replacing standard qualitative principles (e.g., in this case, transitivity $\mathrm{eq}(x,  z) \land \mathrm{eq}(z, y) \vdash \mathrm{eq}(x, y)$).

Still, as already observed on many occasions \cite{DGY19,DLG21}, the restriction to non-expansive language constructs with respect to some purely numerical metric turns out to be too severe in practice. On the one hand, the literature focusing on higher-order languages has mostly restricted its attention to linear or graded languages \cite{Reed2010, Gaboardi2017}, due to well-known difficulties in constructing metric models for full ``simply-typed'' languages \cite{Honsell2022}. On the other hand, even if one restricts to a linear language, the usual metrics defined over functional types are hardly useful in practice, as they assign distances to functions via a comparison of their values in the worst case. For instance, as shown in \cite{DGY19}, the distance between the identity function on $\mathbb{R}$ and the sine function is typically $\infty$ since their values grow arbitrarily far from each other in the worst case, even though their behaviour is very close near $0$.

The \emph{differential logical relations} \cite{DGY19,DLG21,PistoneLICS,DG22} have been introduced as a solution to the aforementioned problems. In this setting, which natively works for unrestricted higher-order languages, the distance between two programs is not necessarily a single number. For instance, two programs of functional type are far apart according to a function itself, which measures how the \emph{error} in the output depends on the \emph{error} in the input, but also on the \emph{value} of the input itself. This way, the notion of distance becomes sufficiently expressive of local differences, at the same time guaranteeing the possibility of compositional reasoning. This paradigm also scales to languages with duplication, recursion \cite{DLG21}, and works even in the presence of effects \cite{DG22}. A similar idea targeting a graded type theory can be found in a recent work \cite{SB26}.

In the literature on program metrics, it has become common to consider \emph{quantale}-valued metrics \cite{Hofmann2014, Stubbe2014}. This means that, as for the differential logical relations, the distance between two programs belongs to some suitable algebra of ``quantities'', i.e., a quantale. This has led to the study of different classes of quantale-valued metrics, each characterized by a particular formulation of the reflexivity condition and the triangle inequality. Among these, \emph{quasi-metrics} \cite{Goubault-Larrecq_2013} and \emph{partial metrics} \cite{matthews, Stubbe2018} have been explored for the study of domains, even for higher-order languages \cite{Geoffroy2020, maestracci2025}. While the first obeys the usual triangle inequality, or the transitivity condition, the second obeys the \emph{strong} transitivity condition $d(x, y) + d(y, z) \geq d(x, z) + d(y , y)$. The latter also takes into account the replacement of the standard reflexivity condition $d(x, x) = 0$ by the \emph{left quasi-reflexivity} condition $d(x, x) \leq d(x, y)$ and the \emph{right quasi-reflexivity} condition $d(y, y) \leq d(x, y)$, implying that a point need not be at distance zero from itself.

A natural question is thus: do the distances between programs that are obtained via differential logical relations constitute some form of quantale-valued metric? In particular, what forms of reflexivity and transitivity do these relations support? The original paper \cite{DGY19} defined symmetric differential logical relations and gave a very weak form of triangle inequality. Subsequent works, relating to the more natural asymmetric case, have either ignored the metric question \cite{DLG21,DG22} or shown that the distances produced must violate \emph{both} the reflexivity of quasi-metric and the strong transitivity of partial metrics \cite{Geoffroy2020,PistoneLICS}.

This paper provides an answer to the above question, building a bridge between current methods for higher-order program distances and the well-established literature on quantale-valued metrics. More specifically, we introduce the notion of \emph{quasi-quasi-metrics}, obtained by weakening the notion of partial metrics. The first ``quasi'' in the name ``quasi-quasi-metric'' comes from left \emph{quasi-}reflexivity, and the second ``quasi'' means the absence of symmetry. The category $\qqm$ of quasi-quasi-metric spaces is cartesian closed, and the cartesian closed structure reflects the construction of differential logical relations. This implies that the construction of differential logical relations preserves the structure of quasi-quasi-metric spaces. We introduce a simply typed lambda calculus $\LL$, and we observe that a fundamental lemma for $\LL$ follows from the cartesian closed structure of $\qqm$. Moreover, we show that quasi-quasi-metric spaces arising from the interpretation of types satisfy, in fact, a variant of the indistancy condition, a variant of the strong transitivity condition and a weak form of the symmetry condition.
% This result does not contradict the ``no-go'' results in \cite{Geoffroy2020,PistoneLICS} because our proof essentially depends on quasi-reflexivity as well as \emph{monotonicity} of morphisms between quasi-quasi-metric spaces: failure of the strong transitivity is deduced without taking quasi-reflexivity into account in \cite{Geoffroy2020}, and monotonicity is not assumed in \cite{PistoneLICS}.

The category of quasi-quasi-metric spaces is not the only way to compositionally measure distances between programs. In the last part of this paper, we introduce a notion of differential prelogical relation as a type-indexed family of quasi-quasi-metrics satisfying the fundamental lemma. This approach extends the framework of prelogical relations \cite{HD02} to the differential setting. Intuitively speaking, differential prelogical relations are ``quasi-quasi-metrics'' on the simply typed lambda calculus $\LL$. We equip the set of differential prelogical relations with a naturally defined partial order, and we investigate the existence of a least and a greatest differential prelogical relation. This investigation is motivated by the fact that the set of program equivalences for a given language does indeed form a complete lattice, with the coarsest one given by the contextual equivalence, and the finest one derived from a syntactic equational theory. A similar pattern emerges here: the poset of differential prelogical relations has arbitrary non-empty meets and the least element presented as a quantitative equational theory \cite{Plotk, Dahlqvist2023}. Conversely, the poset does not necessarily have a greatest differential prelogical relation, suggesting that there is no natural notion of ``contextual quasi-quasi-metric'' on $\LL$. This situation contrasts with that of the lattices of program equivalences, and of the lattices of program metrics for linear programming languages \cite{dallagoFSCD23}.

\paragraph{Outline of this Paper}

\begin{itemize}
\item
  In Section~\ref{sec:from-logic-relat}, we recall logical relations and how differential logical relations in the literature emerge as a quantitative counterpart of logical relations. Furthermore, we investigate how (quasi-)reflexivity and transitivity, including their quantitative counterparts, are preserved by the construction of logical relations and differential logical relations.
\item
  In Section~\ref{sec:quasi-quasi-metric}, based on the observations given in Section~\ref{sec:from-logic-relat}, we introduce a new class of quantale-valued metrics, called quasi-quasi-metrics, and examine the relationship between quasi-quasi-metrics and well-established notions of quasi-metrics and partial quasi-metrics \cite{KUNZI2006}.
  % There are various ways to generalize strong transitivity to the quantitative setting, and among them, we focus on a variant of strong transitivity that we call left strong transitivity.
\item
  In Section~\ref{sec:category-quasi-quasi}, we give the cartesian closed category $\qqm$ of quasi-quasi-metric spaces, observing that the cartesian closed structure serves as the categorical analogue of differential logical relations. In particular, $\qqm$ is a model of the simply typed lambda calculus $\LL$ (Section~\ref{sec:target-language-its}), and the fundamental lemma (Theorem~\ref{thm:fundamental-lemma}) follows from the interpretation of $\LL$ in the cartesian closed category $\qqm$. Moreover, we show that quasi-quasi-metrics arising from the interpretation of $\LL$ satisfy self-indistancy, the left strong transitivity condition and a weak form of the symmetry condition.
\item
  Finally, in Section~\ref{sec:towards-lattice-mmms}, we introduce a notion of differential prelogical relations. The set of differential prelogical relations forms a poset with arbitrary non-empty meets, and in particular, the poset has the least differential prelogical relation. The least differential prelogical relation, which is the ``finest metric'' on $\LL$, coincides with the one derived from a quantitative equational theory on $\LL$. Furthermore, we observe that the poset does not necessarily have a greatest differential prelogical relation. The absence of a greatest differential prelogical relation implies that there is no ``contextual metric'' on $\LL$.
\end{itemize}

\section{From Logical Relations to Differential Logical Relations}
\label{sec:from-logic-relat}

In this section, we informally recall how differential logical relations in the literature emerge as a quantitative counterpart of standard logical relations, at the same time highlighting how the construction of (differential) logical relations preserves the (quantitative counterparts of) properties such as reflexivity, quasi-reflexivity, and transitivity. % , symmetry and transitivity.

\subsection{Logical Relations}
\label{sec:logical-relations}

The theory of logical relations is well-known and has various applications to establish \emph{qualitative} properties of type systems, like e.g.~termination \cite{Girard1989}, bisimulation \cite{Sangiorgi2007}, adequacy \cite{Plotkin77} or parametricity \cite{Plotkin1993, Hermida2014}. The idea is to start from some basic binary relations $\rho_{\itl{o}}$ between the set of values of ground types $\itl{o}$, and \emph{lift} the relations $\rho_{o}$ to a family of binary relations $\rho_{A}$ between the set of values of type $A$ by induction on $\itl{A}$, where $\itl{A}$ varies over all simple types. Indeed, one may consider logical relations for recursive \cite{Dreyer2009}, polymorphic \cite{Reynolds1983, Plotkin1993} or monadic \cite{GL2002} types as well, but we here limit our discussion to simple types. (Here and in what follows in this section, we abuse notation by identifying terms with their values and types with their sets of values.)
% (Here and in what follows in this section, we confuse types and lambda terms with their set theoretic interpretations.)

% \begin{align}
%   (\langle v, v' \rangle, \langle u, u' \rangle) \in
%   \rho_{\itl{A}\times \itl{B}}
%   &\  \iff  \
%     (v, u)
%     \in \rho_{\itl{A}}
%     \ \text{and} \
%     (v', u')
%     \in \rho_{\itl{B}},
%     \tag{$\times$}\label{eq:land} \\
%   (\lam{x}{A}{t}, \lam{x}{A}{s})
%   \in \rho_{\itl{A}\Rightarrow \itl{B}}
%   &\  \iff  \
%     \text{for every }
%     (\itl{v}, \itl{u}) \in \rho_{\itl{A}}, \,
%     \text{we have }
%     (\itl{t}[\itl{v}/x], \itl{s}[\itl{u}/x])
%     \in \rho_{\itl{B}}.
%     \tag{${\Rightarrow}$}\label{eq:to}
% \end{align}
% \begin{align}
%   (v, u) \in
%   \rho_{\itl{A}\times \itl{B}}
%   &\  \iff  \
%     (\mathrm{fst}(v), \mathrm{fst}(u))
%     \in \rho_{\itl{A}}
%     \ \text{and} \
%     (\mathrm{snd}(v), \mathrm{snd}(u))
%     \in \rho_{\itl{B}},
%     \tag{$\times$}\label{eq:land} \\
%   (v, u)
%   \in \rho_{\itl{A}\Rightarrow \itl{B}}
%   &\  \iff  \
%     \text{for every }
%     (\itl{p}, \itl{q}) \in \rho_{\itl{A}}, \,
%     \text{we have }
%     (\itl{v}\, \itl{p}, \itl{u}\, \itl{q})
%     \in \rho_{\itl{B}}.
%     \tag{${\Rightarrow}$}\label{eq:to}
% \end{align}
The lifting is inductively given as follows:
\begin{align}
  (t, s) \in
  \rho_{\itl{A}\times \itl{B}}
  &\  \iff  \
    (\mathrm{fst}(t), \mathrm{fst}(s))
    \in \rho_{\itl{A}}
    \ \text{and} \
    (\mathrm{snd}(t), \mathrm{snd}(s))
    \in \rho_{\itl{B}},
    \tag{$\times$}\label{eq:land} \\
  (t, s)
  \in \rho_{\itl{A}\Rightarrow \itl{B}}
  &\  \iff  \
    \text{for every }
    (\itl{p}, \itl{q}) \in \rho_{\itl{A}}, \,
    \text{we have }
    (\itl{t}\, \itl{p}, \itl{s}\, \itl{q})
    \in \rho_{\itl{B}}.
    \tag{${\Rightarrow}$}\label{eq:to}
\end{align}
% where $\rho^{\Downarrow}_{A}$ is the set of term pairs $(t, s)$ such that $t$ and $s$ evaluate to values $v$ and $u$, respectively, for some $(v, u) \in \rho_{A}$.
In order to derive some properties of type systems from logical relations, one typically wishes to establish the so-called \emph{fundamental lemma}, stating that every well-typed term $x : \itl{A} \vdash \itl{t} : \itl{B}$ \emph{preserves logical relations}. This means that the family of logical relations $\rho_{\itl{A}}$ defined as above satisfies the following implication
\begin{equation*}
  (\itl{p}, \itl{q}) \in \rho_{\itl{A}}
  \implies
  (\itl{t}[\itl{p}/\itl{x}] ,
  \itl{t}[\itl{q}/\itl{x}])
  \in \rho_{\itl{B}}%^{\Downarrow}
  \tag{Fundamental Lemma}
\end{equation*}
for all $p, q \in A$. Notice that the fundamental lemma is equivalent to the instance of reflexivity: the above implication holds if and only if $(\lam{\itl{x}}{\itl{A}}{\itl{t}}, \lam{\itl{x}}{\itl{A}}{\itl{t}}) \in \rho_{\itl{A} \Rightarrow \itl{B}}$ holds for arbitrary term $\itl{x} : \itl{A} \vdash \itl{t} : \itl{B}$.

Of particular interest are the \emph{equivalence} relations, that is, the reflexive, transitive, and symmetric relations, and the \emph{preorders}, that is, the reflexive and transitive relations. As is well-known, both are fundamental structures for reasoning about program properties. We here consider preorders, as we are interested in asymmetric (differential) logical relations in this paper (see Remark \ref{rem:symmetry}). A fundamental observation is that the construction of logical relations preserves preorders: if $\rho_{\itl{A}}$ and $\rho_{\itl{B}}$ are reflexive and transitive, then we can show that $\rho_{\itl{A}\times \itl{B}}$ and $\rho_{\itl{A}\Rightarrow \itl{B}}$ are reflexive and transitive as well whenever the fundamental lemma holds. The case of the function space crucially relies on the fact that all terms of type $\itl{A} \Rightarrow \itl{B}$ must preserve relations (or, in other words, that the fundamental lemma holds): as we observed above, reflexivity $(t, t) \in \rho_{\itl{A} \Rightarrow \itl{B}}$ is equivalent to the fact that the function $t$ is relation-preserving; transitivity, instead, follows from reflexivity of $\rho_{\itl{A}}$ and transitivity of $\rho_{\itl{B}}$. More specifically, given $(t, s) \in \rho_{A \Rightarrow B}$ and $(s, u) \in \rho_{A \Rightarrow B}$, we can deduce $(t, u) \in \rho_{A \Rightarrow B}$ as follows: Suppose that $(p, q) \in \rho_{A}$. By reflexivity of $\rho_{A}$, we have $(p, p) \in \rho_{A}$. Then, by the definition of $\rho_{A \Rightarrow B}$, we obtain $(t \, p, s \, p) \in \rho_{B}$ and $(s \, p, u \, q) \in \rho_{B}$. Since $\rho_{B}$ is transitive, we have $(t \, p, u \, q) \in \rho_{B}$. Thus, $\rho_{A \Rightarrow B}$ is transitive. In the next section, we explore quantitative counterparts of logical relations and preorders.

% The case of the function space crucially relies on the fact that all terms of type $\itl{A} \Rightarrow \itl{B}$ must preserve relations (or, in other words, that the fundamental lemma holds): as we observed above, reflexivity $(\lam{\itl{x}}{\itl{A}}{\itl{t}},\lam{\itl{x}}{\itl{A}}{\itl{t}}) \in \rho_{\itl{A} \Rightarrow \itl{B}}$ is equivalent to the fact that the function $\lam{\itl{x}}{\itl{A}}{\itl{t}}$ is relation-preserving; transitivity, instead, follows from relation-preservation, reflexivity of $\rho_{\itl{A}}$ and transitivity of $\rho_{\itl{B}}$. In the next section, we explore quantitative counterparts of logical relations and preorders.

\begin{rem}\label{rem:symmetry}
  While differential logical relations were symmetric in the original definition \cite{DGY19}, symmetry was abandoned in all subsequent works. The first reason is that several interesting notions of program difference, like e.g.~those arising from \emph{incremental computing} \cite{DLG21, Giarrusso2014, AP2019}, are not symmetric. A second reason is that the cartesian closure is problematic in the presence of both (quasi-)reflexivity and symmetry \cite{PistoneLICS}.
\end{rem}

\subsection{Differential Logical Relations}
\label{sec:diff-logic-relat-intro}

We now discuss what happens when extending logical relations to a quantitative setting. Rather than considering binary relations $r \subseteq X \times X$ expressing that a certain property holds for two elements $x$ and $y$ in $X$, we consider \emph{ternary} relations $r \subseteq X \times Q \times X$. A triple $(x, a, y) \in r$ indicates that we have a discrepancy between $x$ and $y$ to \emph{a certain extent}, quantified by $a \in Q$. Here $Q$ is a \emph{quantale}, an ordered monoid that captures several properties of quantities as expressed by e.g.~non-negative real numbers. A leading example is a ternary relation $r \subseteq \mathbb{R} \times [0,+\infty] \times \mathbb{R}$ given by
\begin{equation*}
  (x, a, y) \in r
  \iff
  |x - y| \leq a.
\end{equation*}
In other words, $(x, a, y) \in r$ holds precisely when $a$ bounds the Euclidean distance between $x$ and $y$.

Given quantales $\mathcal{Q}_{o}$ and ternary relations $\rho_{\itl{o}} \subseteq \itl{o} \times \mathcal{Q}_{\itl{o}} \times \itl{o}$ for ground types $o$, just like the construction of standard logical relations, we can define a family of ternary relations $\rho_{A} \subseteq A \times \mathcal{Q}_{A} \times A$ for each type $A$ by \emph{lifting} the ternary relations $\rho_{o}$. For the lifting, we first inductively define quantales $\mathcal{Q}_{A \times B}$ and $\mathcal{Q}_{\itl{A} \Rightarrow \itl{B}}$ as follows:
\begin{align*}
  \mathcal{Q}_{\itl{A} \times \itl{B}}
  &= \mathcal{Q}_{\itl{A}}
    \times \mathcal{Q}_{\itl{B}},
  \\
  \mathcal{Q}_{\itl{A} \Rightarrow \itl{B}}
  &= (\itl{A} \times \mathcal{Q}_{\itl{A}})
    \Rightarrow \mathcal{Q}_{\itl{B}}.
\end{align*}
% where $(\itl{A}, \mathcal{Q}_{\itl{A}}) \rightarrowtriangle \mathcal{Q}_{\itl{B}}$ is the set of functions $f$ from $\itl{A} \times \mathcal{Q}_{\itl{A}}$ to $\mathcal{Q}_{\itl{A}}$ such that $f(x, -)$ is monotone for all $x \in \itl{A}$.
We equip $\mathcal{Q}_{\itl{A} \times \itl{B}}$ with the componentwise order and $\mathcal{Q}_{\itl{A} \Rightarrow \itl{B}}$ with the pointwise order. The definition of $\mathcal{Q}_{\itl{A} \times \itl{B}}$ means that we measure discrepancies between elements of $\itl{A} \times \itl{B}$ componentwise. Here, we cannot adopt the Euclidean distance or the Manhattan distance because $\mathcal{Q}_{\itl{A}}$ and $\mathcal{Q}_{\itl{B}}$ may be quantales consisting of higher-order functions as seen in the case of $\mathcal{Q}_{\itl{A} \Rightarrow \itl{B}}$. As for the function types $A \Rightarrow B$, each $d \in \mathcal{Q}_{\itl{A} \Rightarrow \itl{B}}$ represents a discrepancy between functions $t$ and $s$ from $\itl{A}$ to $\itl{B}$. The idea is as follows: for all $\itl{p}, \itl{q} \in \itl{A}$ and $a \in \mathcal{Q}_{\itl{A}}$, if $\itl{q}$ is in the disk with center $\itl{p}$ and radius $a$, then $\itl{s} \, \itl{q}$ is in the disk with center $\itl{t} \, \itl{p}$ and radius $d(\itl{p}, a)$. In other words, $d$ characterizes similarity between each point of $\itl{t}$ and local behavior of $\itl{s}$. Figure~\ref{fig:discrepancy} illustrates how $d$ bounds the distance between $\itl{t}$ and $\itl{s}$ for the case $\itl{A} = \itl{B} = \mathbb{R}$.
% The idea of how $d$ represents a discrepancy between $t, s \colon A \to B$ is as follows: for any $\itl{p}, \itl{q} \in \itl{A}$ and any $a \in \mathcal{Q}_{\itl{A}}$, if $\itl{q}$ is in the disk with center $\itl{p}$ and radius $a$, then both $\itl{t} \, \itl{q}$ and $\itl{s} \, \itl{q}$ are in the disk with center $\itl{t} \, \itl{p}$ and raduis $d(\itl{p}, a)$. In other words, $d$ characterizes local similarity between $\itl{t}$ and $\itl{s}$. In Figure~\ref{fig:discrepancy}, we illustrate how $d$ bounds the distance between $\itl{t}$ and $\itl{s}$ for the case $\itl{A} = \itl{B} = \mathbb{R}$. The monotonicity condition required to $d \in \mathcal{Q}_{\itl{A} \Rightarrow \itl{B}}$ is a novelty of this paper, and the condition is essential to derive results in Section~\ref{sec:relating-qq-metric}.

Formalizing this idea, we can lift ternary relations $\rho_{\itl{o}} \subseteq \itl{o} \times \mathcal{Q}_{\itl{o}} \times \itl{o}$ to ternary relations $\rho_{\itl{A}} \subseteq \itl{A} \times \mathcal{Q}_{\itl{A}} \times \itl{A}$ as follows:
\begin{align*}
  ((\itl{t},\itl{s}),
  (a, b),
  (\itl{p},\itl{q}))
  \in \rho_{\itl{A} \times \itl{B}}
  &\iff
    (\itl{t}, a, \itl{p})
    \in \rho_{\itl{A}} \ \text{and} \
    (\itl{s}, b, \itl{q})
    \in \rho_{\itl{B}}, \\
  (\itl{t}, d, \itl{s})
  \in \rho_{\itl{A} \Rightarrow  \itl{B}}
  &\iff
    \text{for every }
    (\itl{p}, a, \itl{q})
    \in \rho_{\itl{A}}, \,
    \text{we have }
  % \\
  % &\mathrel{\phantom{\iff}}
    (\itl{t} \, \itl{p},
    d(\itl{p}, a),
    \itl{s} \, \itl{q})
    \in \rho_{\itl{B}}.
    % \text{ and }
    % (\itl{t} \, \itl{p},
    % d(\itl{t}, a),
    % \itl{t} \, \itl{q})
    % \in \rho_{\itl{B}}.
\end{align*}
\begin{figure}
  \centering
  \begin{tikzpicture}
    \draw[->] (-5,-2) -- ++ (9,0);
    \draw[->] (-4,-2.5) -- ++ (0,4.5);
    \draw[fill=black] (0,0) circle [radius = 0.05] node (a) {};
    \draw[dotted] (a) -- ++(0, -2)
    node[below] {$\itl{p}$};
    \draw[|-|] (0.1,-1.25) --
    node[below] {$a$} ++(1.4,0);
    \draw[|-|] (-0.1,-1.25) --
    node[below] {$a$} ++(-1.4,0);
    \draw[dotted] (a) -- ++(-4, 0)
    node[left] {$\itl{t} \, \itl{p}$};
    \draw[|-|] (-2, 0.1) --
    node[left] {$d(\itl{p}, a)$} ++(0, 0.9);
    \draw[|-|] (-2, -0.1) --
    node[left] {$d(\itl{p}, a)$} ++(0, -0.9);
    \node[draw, minimum width=3cm, minimum height=2cm]
    at (0, 0)  {};
    \draw (-2, 1.25) to[out=330, in=180] (-1.5,0.9);
    \draw[very thick] (-1.5,0.9) to[out=0, in=180] (0,-0.5)
    to[out=0, in=180] (1,0)
    to[out=0, in=180] (1.5,0.25);
    \draw (1.5,0.25) to[out=0, in=170] (3, -1.5);
    \node at (2.5,-0.75) {$\itl{s}$};
  \end{tikzpicture}
  \caption{The Idea of Differential Logical Relations at Function Types}
  \label{fig:discrepancy}
\end{figure}
% Observe that the quantale $\mathcal{Q}_{\itl{A}\Rightarrow \itl{B}}$ for the function type is itself a set of functions relating pairs of a term of type $\itl{A}$ and a quantity in $\mathcal{Q}_{\itl{A}}$ with quantities in $\mathcal{Q}_{\itl{B}}$.
As is shown in \cite{DGY19}, a fundamental lemma holds under the following form: for every term $x:\itl{A} \vdash \itl{t} : \itl{B}$, there exists $\prt{\itl{t}} \in \mathcal{Q}_{\itl{A} \Rightarrow \itl{B}}$ such that for all $\itl{p}, \itl{q} \in \itl{A}$ and $a \in \mathcal{Q}_{\itl{A}}$,
\begin{equation}\label{eq:log_rel_dlr}
  (\itl{p}, a, \itl{q})
  \in \rho_{\itl{A}}
  \implies
  (\itl{t} [\itl{p}/x],
  \prt{\itl{t}}(\itl{p}, a),
  \itl{t} [\itl{q}/x])
  \in \rho_{\itl{B}}
  \tag{Fundamental Lemma}.
\end{equation}
The function $\prt{\itl{t}}$ behaves like some sort of \emph{derivative} of $\itl{t}$: it relates input errors with output errors. This connection is investigated in more detail in \cite{DLG21,PistoneLICS}.

So far, everything seems to work just as in the standard qualitative case. There is, however, an important point on which differential logical relations differ from standard logical relations. This difference is well visible when we consider the quantitative counterpart of the notions of reflexivity and transitivity. Given a ternary relation $r \subseteq X \times Q \times X$ between a set $X$ and a quantale $Q$, we can translate reflexivity and transitivity into the following forms:
\begin{align*}
  r \text{ is reflexive}
  &\iff
    \text{for every $\itl{x} \in \itl{X}$, we have $(\itl{x}, \mathbf{1}, \itl{x}) \in r$},
  \\
  r \text{ is transitive}
  &\iff
    \text{for all $(\itl{x}, a, \itl{y}) \in r$ and $(\itl{y}, b,\itl{z}) \in r$, we have $(\itl{x}, a \otimes b,\itl{z}) \in r$.}
\end{align*}
Here, $\mathbf{1}$ is the unit element and $a \otimes b$ is the multiplication of $Q$. For example, the unit element of the Lawvere quantale $[0,+\infty]$ is $0$, and its multiplication $a \otimes b$ is the addition $a + b$. If we read $(x, a, y) \in r$ as ``the distance between $x$ and $y$ is at most $a$'', the above conditions coincide with the reflexivity condition and the triangle inequality of a metric. It is clear then that we can translate the notion of a preorder into the quantitative setting by following the above translations. This is a so-called \emph{quasi-metric} (see \cite[Definition~6.1.1]{Goubault-Larrecq_2013}), essentially, a metric without the symmetry condition, indeed a very well-studied class of metrics.

In this respect, while the construction of logical relations lifts preorders well to all simple types, their quantitative counterparts, the quasi-metrics, are \emph{not} preserved by the higher-order lifting of differential logical relations. Indeed, in the previous section, we observed that an essential ingredient in the lifting of the reflexivity condition is the fundamental lemma; yet, in the framework of differential logical relations, the fundamental lemma only produces, for each term $\itl{t} \in \itl{A}$, the ``reflexivity'' condition $(\itl{t}, \prt{\itl{t}}, \itl{t}) \in \rho_{\itl{A}}$, which differs from standard reflexivity in that the distance $\prt{\itl{t}}$ may \emph{not} be the unit element of $\mathcal{Q}_{\itl{A}}$. In fact, differential logical relations are inherently non-reflexive; as observed in \cite{DGY19}, a function $t \in A \Rightarrow B$ enjoys the reflexivity condition $(t, \mathbf{1}, t) \in \rho_{A \Rightarrow B}$ if and only if $t$ is a constant function. % The following example further illustrates how reflexivity fails: we investigate the self-distance of the sine function. It follows from the definition of differential logical relations that a function $d \colon \mathbb{R} \times [0,+\infty] \to [0,+\infty]$ bounds the self-distance of the sine function if and only if for all reals $x, y \in \mathbb{R}$ and for all $a \in [0,+\infty]$ such that $|x - y| \leq a$, we have $|\sin(x) - \sin(y)| \leq d(x, a)$. In particular, such bound $d$ must satisfy $|\sin(0) - \sin(1)| \leq d(0, 1)$. This means that $d(0, 1)$ must be larger than or equal to $\sin(1)$. Since the unit element $\mathbf{1} \colon \mathbb{R} \times [0,+\infty] \to [0,+\infty]$ is the constant function $\mathbf{1}(x, a) = 0$, we see $(\sin,\mathbf{1},\sin) \notin \rho_{\mathbb{R}}$, that is, reflexivity for the sine function fails. By the same argument, we can show that the self-distance between any non-constant function must be greater than the unit $\mathbf{1}$.

\subsection{Weakening the Reflexivity Condition}
\label{sec:weak-refl-cond}

The observations so far imply that the metric structure naturally arising from differential logical relations cannot be that of standard (quasi-)metrics. Rather, it must be something close to the \emph{partial} metrics \cite{matthews, Stubbe2018}, where the reflexivity condition $d(x,x) = 0$ is replaced by the left quasi-reflexivity condition $d(x, x) \leq d(x, y)$ and the right quasi-reflexivity condition $d(y, y) \leq d(x, y)$. In what follows, we restrict our attention primarily to left quasi-reflexivity as either variant suffices to lift transitivity for differential logical relations (see below).
% as the presence of both left and right quasi-reflexivity is an obstacle to establishing the fundamental lemma.
For brevity, we refer to left quasi-reflexivity simply as quasi-reflexivity when no ambiguity arises.

% and the transitivity condition $d(x, z) \leq d(x, y) + d(y, z)$ is replaced by the \emph{strong transitivity} condition $d(x, z) + d(y, y) \leq d(x, y) + d(y, z)$.

Let us re-examine the lifting of transitivity for logical relations. As we have observed, transitivity of $\rho_{A \Rightarrow B}$ follows from reflexivity of $\rho_{\itl{A}}$ and transitivity of $\rho_{B}$. However, in fact, in order to derive transitivity of $\rho_{A \Rightarrow B}$, we only need \emph{quasi-reflexivity} of $\rho_{\itl{A}}$ rather than reflexivity of $\rho_{A}$. Here, we say that $\rho_{\itl{A}}$ is quasi-reflexive when for every $(\itl{t}, \itl{s}) \in \rho_{\itl{A}}$, we have $(\itl{t}, \itl{t}) \in \rho_{\itl{A}}$, which intuitively means that only the points smaller than someone are smaller than themselves. Indeed, the crucial step in proving transitivity for $\rho_{A \Rightarrow B}$ in Section~\ref{sec:logical-relations} is showing that $(p, p) \in \rho_{A}$ holds for every $(p, q) \in \rho_{A}$; quasi-reflexivity suffices for this step.
% From quasi-reflexivity of $\rho_{A}$, we can deduce transitivity of $\rho_{A \Rightarrow B}$ as follows: given $(\itl{t}, \itl{s}) \in \rho_{\itl{A} \Rightarrow \itl{B}}$ and $(\itl{s}, \itl{u}) \in \rho_{\itl{A} \Rightarrow \itl{B}}$, then for any $(\itl{p}, \itl{q}) \in \rho_{\itl{A}}$, it follows from quasi-reflexivity of $\rho_{\itl{A}}$ that $(\itl{p}, \itl{p}) \in \rho_{\itl{A}}$. Then, by relation preservation, we obtain $(\itl{t} \, \itl{p}, \itl{s} \, \itl{p}) \in \rho_{\itl{B}}$ and $(\itl{s} \, \itl{p}, \itl{u} \, \itl{q}) \in \rho_{\itl{B}}$. Finally, by transitivity of $\rho_{\itl{B}}$, we obtain $(\itl{t} \, \itl{p}, \itl{u} \, \itl{q}) \in \rho_{\itl{B}}$.

% In what follows, we use the term quasi-reflexive in place of left-quasi-reflexive as we do not consider \emph{right}-quasi-reflexivity in this paper.
In this way, one can develop a theory of logical relations for \emph{quasi-preorder}s, i.e., quasi-reflexive and transitive relations. The sole delicate point is that, in order to let quasi-reflexive relations be preserved by the construction of function spaces, one has to slightly modify the relation lifting as follows:
\begin{equation}\label{eq:tostar}
  (\itl{t} , \itl{s})
  \in \rho_{\itl{A}\Rightarrow \itl{B}}
  \iff
  \text{for every }
  (\itl{p}, \itl{q})
  \in \rho_{\itl{A}},\,
  \text{we have }
  (\itl{t} \, \itl{p},
  \itl{s} \, \itl{q})
  \in  \rho_{\itl{B}}
  \text{ and }
  (\itl{t} \, \itl{p},
  \itl{t} \, \itl{q})
  \in \rho_{\itl{B}}
  \tag{${\Rightarrow}^{\mathrm{q}}$}.
\end{equation}
Compared to \eqref{eq:to}, the lifting \eqref{eq:tostar} includes the second clause $(\itl{t} \, \itl{p}, \itl{t} \, \itl{q}) \in \rho_{\itl{B}}$ relating the action of $\itl{t}$ on $\itl{p}$ and $\itl{q}$. With this definition, one can check that if $\rho_{\itl{A}}$ and $\rho_{\itl{B}}$ are quasi-preorders, then $\rho_{\itl{A} \Rightarrow \itl{B}}$ is a quasi-preorder as well. In the quantitative setting, we need to change the definition of the differential logical relation $\rho_{A \Rightarrow B}$ as follows:
\begin{multline*}
  \qquad
  (\itl{t}, d, \itl{s})
  \in \rho_{\itl{A} \Rightarrow  \itl{B}}
  \iff
  \text{for every }
  (\itl{p}, a, \itl{q})
  \in \rho_{\itl{A}},
  \\
  (\itl{t} \, \itl{p},
  d(\itl{p}, a),
  \itl{s} \, \itl{q})
  \in \rho_{\itl{B}}
  \text{ and }
  (\itl{t} \, \itl{p},
  d(\itl{p}, a),
  \itl{t} \, \itl{q})
  \in \rho_{\itl{B}}.
  \qquad
\end{multline*}
Notice that the definition of $\rho_{\itl{A}\Rightarrow \itl{B}}$ closely imitates the clause \eqref{eq:tostar} for quasi-preorders. Thanks to this modification, we can show that $\rho_{A \Rightarrow B}$ is quasi-reflexive and transitive whenever $\rho_{A}$ is quasi-reflexive and $\rho_{B}$ is transitive. Here, we say that $\rho_{A}$ is quasi-reflexive when for every $(t, a, s) \in \rho_{A}$, we have $(t, a, t) \in \rho_{A}$.

In this paper, we also refine the definition of differential logical relations in another key respect: we replace the definition of $\mathcal{Q}_{A \Rightarrow B}$ as follows:
\begin{equation*}
  \mathcal{Q}_{A \Rightarrow B} =
  (A, \mathcal{Q}_{A}) \rightarrowtriangle \mathcal{Q}_{B},
\end{equation*}
where the right hand side is the set of functions $d \colon A \times \mathcal{Q}_{A} \to \mathcal{Q}_{B}$ such that $d(x, -) \colon \mathcal{Q}_{A} \to \mathcal{Q}_{B}$ is \emph{monotone} for every $x \in A$. The monotonicity condition means that larger input errors result in larger output errors. This restriction plays a crucial role in the preservation of variants of strong transitivity and indistancy, as well as a weak form of symmetry by the construction of differential logical relations (See Section~\ref{sec:relating-quasi-quasi-metric}).

%%% Local Variables:
%%% mode: LaTeX
%%% TeX-master: "lmcs_main.tex"
%%% End:

\section{Quasi-Quasi-Metric}
\label{sec:quasi-quasi-metric}

In the following sections, we develop a theory of quantale-valued metrics for differential logical relations, with a specific focus on \emph{quasi}-reflexivity and transitivity. Motivated by the observation on the lifting of transitivity, we introduce the notion of quasi-quasi-metrics as quasi-reflexive and transitive quantale-valued predicates, and
  examine how they are related to the well-established notions of quasi-metrics and partial quasi-metrics.
% examine how they are related to the well-established notion of quasi-metrics.
% examine how they are related to the well-established notions of quasi-metrics and partial quasi-metrics.

\subsection{Quantales and the Definition of Quasi-Quasi-Metric}
\label{sec:quant-defin-quasi}

Following the approach of differential logical relations, we make use of quantales to measure discrepancies between programs. We first recall the definition of a quantale.

\begin{defi}
  A \emph{quantale} is a complete lattice $(Q,\sqsubseteq)$ endowed with the structure of a monoid $(Q, \mathbf{1}, \otimes)$ such that for every $x \in Q$, both $x \otimes (-)$ and $(-) \otimes x$ preserve arbitrary suprema. A quantale $Q$ is \emph{integral} (cf.~\cite{Hofmann2014}, p.~148) when the unit $\mathbf{1}$ is the greatest element, and $Q$ is \emph{commutative} when the multiplication $(-) \otimes (-)$ is commutative.
\end{defi}

Throughout this paper, we simply denote a quantale $(Q, \sqsubseteq, \mathbf{1},\otimes)$ by its underlying set $Q$ and always assume that quantales are commutative and integral.

\begin{rem}
  We reserve the square symbols $\sqsubseteq$, $\bigsqcap$ and $\bigsqcup$ for the partial order, the meet operation and the join operation of quantales, respectively.
\end{rem}

Let $Q$ be a quantale. We define the \emph{residual} $x \multimap y$ of $x, y \in Q$ by
\begin{equation*}
  x \multimap y = \bigsqcup \{z \in Q \mid x \otimes z \sqsubseteq y\}.
\end{equation*}
For every $x \in Q$, the function $x \multimap (-)$ is right adjoint to $x \otimes (-)$, and hence, for all $y, z \in Q$, we have
\begin{equation*}
  z \sqsubseteq x \multimap y
  \iff x \otimes z \sqsubseteq y.
\end{equation*}
% A commutative quantale $Q$ is \emph{divisible} \cite{Stubbe2018} if for all $x, y \in Q$, we have $x \sqsubseteq y$ whenever $y \otimes (y \multimap x) = x$. Equivalently, $Q$ is divisible iff, whenever $x \sqsubseteq y$, there exists $z$ such that $x = y \otimes z$.

\begin{exa}\label{eg:Law}
  The \emph{Lawvere quantale} consists of the set of non-negative extended reals $[0,+\infty]$ and the \emph{reversed} order $a \sqsubseteq b$ if and only if $a \geq b$ equipped with the addition of real numbers. Notice that the ordering of quantales is \emph{reversed} with respect to usual metric intuitions: the unit $0$ is the greatest element, and joins correspond to infima. For $a, b \in [0,+\infty]$, the residual $a \multimap b$ is given by
  \begin{equation*}
    a \multimap b =
    \begin{cases}
      b \dotdiv a, & \text{if $a \neq \infty$}, \\
      0, & \text{if $a = \infty$},
    \end{cases}
  \end{equation*}
  where $b \dotdiv a$ is the truncated subtraction. For example, we have $\infty \multimap \infty = 0$ and $a \multimap \infty = \infty$ whenever $a$ is finite. Throughout the remainder of this paper, when working with the Lawvere quantale, we write $\lawsqcap{a}{b}$ (the join of $a$ and $b$ with respect to the standard order on $[0, +\infty]$) and $\lawbigsqcap_{i \in I} a_{i}$ (the join of the family $\{a_{i} \in [0,+\infty]\}_{i \in I}$ with respect to the standard order on $[0, +\infty]$) for the meet and the infimum with respect to the quantale order $\sqsubseteq$, so that
  \begin{equation*}
    a \sqcap b = \lawsqcap{a}{b},
    \qquad
    \bigsqcap_{i \in I} a_{i} = \lawbigsqcap_{i \in I} a_{i}.
  \end{equation*}
\end{exa}

\begin{exa}\label{eg:xqq}
  For a set $X$, and for quantales $Q$ and $P$, let $(X, Q) \rightarrowtriangle P$ be the quantale consisting of functions $f \colon X \times Q \to P$ such that $f(x, -) \colon Q \to P$ is monotone for every $x \in X$. We equip $(X, Q) \rightarrowtriangle P$ with the pointwise order and the pointwise multiplication. The unit of $(X, Q) \rightarrowtriangle P$ is the constant function $\mathbf{1}(x, a) = \mathbbm{1}$ where $\mathbbm{1}$ is the unit of $P$. The residual $f \multimap g$ is given by
  \begin{equation*}
    (f \multimap g)(x, a) =
    \bigsqcap_{b \sqsupseteq a}
    \bigl(f(x, b) \multimap g(x, b)\bigr).
  \end{equation*}
  In this paper, quantales of this form arise from the interpretation of function types.
\end{exa}

We then define the notion of quantale-valued quasi-quasi-metrics as the quantitative counterpart of the notion of quasi-preorders (being ``quasi'' both in the sense of quasi-metrics, i.e. the rejection of symmetry, and of quasi-preorders, i.e. the weakening of reflexivity).

\begin{defi}
  Let $X$ be a set, and let $Q$ be a quantale. We say that a ternary relation $r \subseteq X \times Q \times X$ is \emph{reflexive}, \emph{quasi-reflexive}, or \emph{transitive} if the following conditions hold, respectively:
  \begin{align*}
    \text{$r$ is \emph{reflexive}}
    &\iff \text{for every $x \in X$, we have $(x, \mathbf{1}, x) \in r$};
    \\
    \text{$r$ is \emph{quasi-reflexive}}
    &\iff \text{for every $(x, a, y) \in r$, we have $(x, a, x) \in r$};
    \\
    \text{$r$ is \emph{transitive}}
    &\iff \text{for all $(x, a, y) \in r$ and $(y, b, z) \in r$, we have $(x, a \otimes b, z) \in r$.}
  \end{align*}
\end{defi}
\begin{defi}
  For a set $X$ and a quantale $Q$, a \emph{$Q$-valued (pseudo) quasi-quasi-metric} on $X$ is a ternary relation $r \subseteq X \times Q \times X$ that is quasi-reflexive and transitive. A \emph{quasi-quasi-metric space} is a triple $(X, Q, r)$ consisting of a set $X$, a quantale $Q$ and a quasi-quasi-metric $r \subseteq X \times Q \times X$.
\end{defi}
The ``pseudo'' prefix stands for the fact that the usual separation property, namely that $(x, \mathbf{1}, y) \in r$ implies $x = y$, needs not hold. As we do not investigate the separation property throughout this paper, all metric notions discussed in the rest of the paper are to be understood as implicitly ``pseudo''. When $Q$ and $X$ are clear from the context, we simply refer to $r$ as a quasi-quasi-metric, rather than a $Q$-valued quasi-quasi-metric on $X$. To provide intuition for quasi-quasi-metrics, we give an example.
% The construction of differential logical relations preserves the structure of quasi-quasi-metrics. Indeed, in Section~\ref{sec:category-quasi-quasi}, we will observe that the argument showing that the quasi-preorders lift to all simple types scales well to the quantitative setting, that is, a quasi-quasi-metric on the ground type induces quasi-quasi-metrics on all simple types.

\begin{exa}\label{eg:quasi-qqm}
  Let $(X, d)$ be a hemi-metric space (see \cite[Definition~6.1.1]{Goubault-Larrecq_2013}), i.e., $X$ is a set and $d \colon X \times X \to [0,+\infty]$ is a function satisfying $d(x, x) = 0$ and $d(x, y) \lawotimes d(y, z) \lawsqsubseteq d(x, z)$ for all $x, y, z \in X$. % (Recall that $\sqsubseteq$ denotes the reversed order, and $\otimes$ denotes the addition. See Example~\ref{eg:Law}.)
  Then, for every function $w \colon X \to [0,+\infty]$, the ternary relation $r \subseteq X \times [0,+\infty] \times X$ given by
  \begin{equation*}
    (x, a, y) \in r \iff
    \lawsqcap{wx}{d(x, y)} \lawsqsupseteq a
  \end{equation*}
  is a quasi-quasi-metric on $X$. Let us check this claim. We first show that $r$ is quasi-reflexive. Suppose that $(x, a, y) \in r$. By the definition of $r$, we have
  \begin{equation*}
    \lawsqcap{wx}{d(x, x)} = wx \lawsqsupseteq \lawsqcap{wx}{d(x, y)} \lawsqsupseteq a.
  \end{equation*}
  Thus, $(x, a, x) \in r$. We next show that $r$ is transitive. Suppose that $(x, a, y) \in r$ and $(y, b, z) \in r$. Then, we have
  \begin{equation*}
    wx \lawsqsupseteq \lawsqcap{wx}{d(x, y)} \lawsqsupseteq a \lawsqsupseteq a \lawotimes b,
  \end{equation*}
  and
  \begin{equation*}
    d(x, z) \lawsqsupseteq d(x, y) \lawotimes d(y, z) \lawsqsupseteq
    \lawsqcapparen{wx}{d(x, y)} \lawotimes \lawsqcapparen{wy}{d(y, z)}
    \lawsqsupseteq a \lawotimes b.
  \end{equation*}
  Thus, $\lawsqcap{wx}{d(x, z)} \lawsqsupseteq a \lawotimes b$, and therefore, we conclude that $(x, a \lawotimes b, z) \in r$.
\end{exa}

The introduction of quasi-quasi-metrics naturally raises several questions: how can we formalize these structures in the form of standard metrics, and how do their structures relate to more established notions
like quasi-metrics and partial quasi-metrics?
% like quasi-metrics and partial metrics?
In the following sections, we explore these relationships.

\subsection{Quantale-Valued Relations}
\label{sec:quant-valu-relat}

While we introduced quasi-quasi-metrics as ternary relations following the presentation of differential logical relations \cite{DGY19} and quantitative equational theories \cite{Plotk,Dahlqvist2023}, the standard metrics are presented as functions from the product of the underlying set with itself to a quantale of distances. We bridge this gap by defining translations between these two presentations of ``metrics''. We first prepare terminology.

\begin{defi}
  Given a quantale $Q$ and a set $X$, a \emph{$Q$-relation over $X$} is a function $\phi \colon X \times X \to Q$, which can be understood as a matrix with values in $Q$; and a \emph{$Q$-predicate over $X$} is a ternary relation $r \subseteq X \times Q \times X$. % We denote the set of $Q$-relations over $X$ and $Y$ by $\qrel{Q}(X, Y)$ and denote the set of $Q$-predicates between $X$ and $Y$ by $\qpred{Q}(X, Y)$.
\end{defi}

We then define translations between $Q$-relations and $Q$-predicates. Let $X$ be a set, and let $Q$ be a quantale. For a $Q$-relation $\phi \colon X \times X \to Q$, we define a $Q$-predicate $\check{\phi} \subseteq X \times Q \times X$ by
\begin{equation*}
  \check{\phi} =
  \{(x, a, y) \in
  X \times Q \times X
  \mid a \sqsubseteq \phi(x, y)\}.
\end{equation*}
Conversely, for a $Q$-predicate $r \subseteq X \times Q \times X$, we define a $Q$-relation $\hat{r} \colon X \times X \to Q$ by
\begin{equation*}
  \hat{r}(x,y) =
  \bigsqcup
  \{ a \in Q
  \mid
  (x, a, y) \in r\}.
\end{equation*}
Intuitively, $(x, a, y) \in \check{\phi}$ means that $a$ is a bound of the distance between $x$ and $y$, and $\hat{r}(x,y)$ is the distance between $x$ and $y$. This correspondence becomes bijective when we restrict the translations to \emph{closed} $Q$-predicates.

\begin{defi}
  Let $X$ be a set, and let $Q$ be a quantale. A \emph{closed $Q$-predicate} over $X$ is a $Q$-predicate $r \subseteq X \times Q \times X$ such that
  \begin{itemize}
  \item if $(x, a, y) \in r$ and $b \sqsubseteq a$, then $(x, b, y) \in r$; and
  \item for every family $\{a_{i} \in Q\}_{i \in I}$, if $(x, a_{i}, y) \in r$ for every $i \in I$, then $(x, \bigsqcup_{i \in I} a_{i}, y) \in r$.
  \end{itemize}
\end{defi}

It is straightforward to see that for every $Q$-relation $\phi$, the $Q$-predicate $\check{\phi}$ is closed. Now, we prove that the translations restricted to closed $Q$-predicates form a bijection.
\begin{prop}\label{prop:rhd}
  Let $X$ be a set, and let $Q$ be a quantale. For every closed $Q$-predicate $r \subseteq X \times Q \times X$, and for all $x,y \in X$ and $a \in Q$,
  \begin{equation*}
    a \sqsubseteq \hat{r}(x, y)
    \iff
    (x, a, y) \in r.
  \end{equation*}
  Furthermore, $\check{(-)}$ and $\hat{(-)}$ form a bijection between the set of $Q$-relations over $X$ and the set of closed $Q$-predicates over $X$.
\end{prop}
\begin{proof}
  We first show that $a \sqsubseteq \hat{r}(x, y)$ if and only if $(x, a, y) \in r$. Suppose that $a \sqsubseteq \hat{r}(x, y)$. It follows from $r$ being closed under joins that $(x, \hat{r}(x, y), y)$ is an element of $r$. Therefore, since $r$ is downward closed, we obtain $(x, a, y) \in r$. The other implication follows from the definition of $\hat{r}(x, y)$. We next check that the constructions $\check{(-)}$ and $\hat{(-)}$ form a bijection. For every $Q$-relation $\phi \colon X \times X \to Q$, we have
  \begin{equation*}
    \hat{\check{\phi}}(x, y)
    = \bigsqcup
    \{a \in Q \mid (x, a, y) \in \check{\phi}\}
    = \bigsqcup
    \{a \in Q \mid a \sqsubseteq \phi(x, y)\}
    = \phi(x, y).
  \end{equation*}
  Conversely, for every closed $Q$-predicate $r \subseteq X \times Q \times X$,
  \begin{equation*}
    (x, a, y) \in \check{\hat{r}}
    \iff a \sqsubseteq \hat{r}(x, y)
    \iff (x, a, y) \in r.
  \end{equation*}
  The latter equivalence follows from the first part of this proposition.
\end{proof}

We say that a $Q$-relation $\phi$ is reflexive, quasi-reflexive, or transitive whenever $\check{\phi}$ is. Similarly, we say that $\phi$ is a quasi-quasi-metric whenever $\check{\phi}$ is a quasi-quasi-metric. From the bijective correspondence between $Q$-relations and closed $Q$-predicates, we obtain the following characterization of reflexivity, quasi-reflexivity and transitivity in terms of quantale-valued relations.
\begin{prop}\label{prop:rqt}
  Let $X$ be a set, and let $Q$ be a quantale. For every $Q$-relation $\phi \colon X \times X \to Q$, the following equivalences hold:
  \begin{align*}
    \text{$\phi$ is reflexive}
    &\iff
      \text{for every $x \in X$, we have $\mathbf{1} \sqsubseteq \phi(x, x)$},
    \\
    \text{$\phi$ is quasi-reflexive}
    &\iff
      \text{for all $x,y \in X$, we have $\phi(x, y) \sqsubseteq \phi(x, x)$},
    \\
    \text{$\phi$ is transitive}
    &\iff
      \text{for all $x,y,z \in X$, we have $\phi(x, y) \otimes \phi(y, z) \sqsubseteq \phi(x, z)$}.
  \end{align*}
\end{prop}
\begin{proof}
  We only check the second equivalence. We can similarly check other equivalences. Let $x$ and $y$ be elements of $X$. By the definition of $\check{\phi}$, we have $(x, \phi(x, y), y) \in \check{\phi}$. If $\phi$ is quasi-reflexive, then from $(x, \phi(x, y), y) \in \check{\phi}$, we obtain $(x, \phi(x, y), x) \in \check{\phi}$, and therefore, $\phi(x, y) \sqsubseteq \phi(x, x)$. Conversely, if $\phi(x, y) \sqsubseteq \phi(x, x)$, then for every $(x, a, y) \in \check{\phi}$, we have $a \sqsubseteq \phi(x, y) \sqsubseteq \phi(x, x)$. Hence, $(x, a, x) \in \check{\phi}$.
\end{proof}

Table~\ref{table:summary} summarizes the formulations of reflexivity, quasi-reflexivity and transitivity
% and left strong transitivity
in terms of $Q$-relations, $Q$-predicates, and binary relations regarded as $\{\bot,\top\}$-relations. Here, $\{\bot,\top\}$ is the quantale given by $\bot \sqsubseteq \top$ and $x \otimes y = x \sqcap y$ for all $x, y \in \{\bot, \top\}$, so its unit is $\top$; we identify a $\{\bot,\top\}$-relation $\phi \colon X \times X \to \{\bot,\top\}$ with the following binary relation on $X$:
\begin{equation*}
  \{(x, y) \in X \times X \mid \phi(x, y) = \top\}.
\end{equation*}

\begin{table}
  \centering
  \begin{tabular}{l}
    \begin{tabular}{|c|c|rcl|}
      \cline{2-5}
      \multicolumn{1}{c|}{}
      & Reflexivity
      & \multicolumn{3}{c|}{Quasi-reflexivity}
      \\ \hline
      $Q$-relation
      & $\phi(x, x) = \mathbf{1}$
      & $\phi(x, y)$ & $\sqsubseteq$ & $\phi(x, x)$
      \\ \hline
      $Q$-predicate
      & $(x, \mathbf{1}, x) \in r$
      & $(x, a, y) \in r$ & $\implies$ & $(x, a, x) \in r$
      \\ \hline
      Binary relation
      & $(x, x) \in R$
      & $(x, y) \in R$ & $\implies$ & $(x, x) \in R$
      \\ \hline
    \end{tabular}
    \\\\
    \begin{tabular}{|c|rclcl|}
      \cline{2-6}
      \multicolumn{1}{c|}{}
      & \multicolumn{5}{c|}{Transitivity}
      \\ \hline
      $Q$-relation
      & $\phi(x, y)$ & {\!\!$\otimes$\!\!} & $\phi(y, z)$ & $\sqsubseteq$ & $\phi(x, z)$
      \\ \hline
      $Q$-predicate
      & $(x, a, y) \in r$ & {\!\!$\&$\!\!} & $(y, b, z) \in r$ & $\implies$ & $(x, a \otimes b, z) \in r$
      \\ \hline
      Binary relation
      & $(x, y) \in R$ & {\!\!$\&$\!\!} & $(y, z) \in R$ & $\implies$ & $(x, z) \in R$
      \\ \hline
    \end{tabular}
  \end{tabular}
  \caption{Summary of Properties for $Q$-relations, $Q$-predicates and Binary Relations}
  \label{table:summary}
\end{table}

\subsection{Quasi-Metrics and Partial Quasi-Metrics}
\label{sec:quasi-and-partial-metrics}

In this section, we use the language of quantale-valued relations and quantale-valued predicates to explore the connections between the notion of quasi-quasi-metrics and more well-established
notions of quasi-metrics and partial quasi-metrics.
% notion of quasi-metrics.
% notions of quasi-metrics.
We first recall the notion of quasi-metrics.
\begin{defi}\label{def:quasi}
  Let $Q$ be a quantale. A \emph{quasi-metric on a set $X$} is a reflexive and transitive $Q$-predicate $r \subseteq X \times Q \times X$. We say that a $Q$-relation is a \emph{quasi-metric} when the associated $Q$-predicate is a quasi-metric, i.e., a $Q$-relation $\phi \colon X \times X \to Q$ is a \emph{quasi-metric on $X$} whenever $\phi$ satisfies the reflexivity condition $\mathbf{1} \sqsubseteq \phi(x, x)$ and the transitivity condition $\phi(x, y) \otimes \phi(y, z) \sqsubseteq \phi(x, z)$.
\end{defi}
When $Q$ is the Lawvere quantale $[0,+\infty]$, the above definition coincides with that of a hemi-metric, i.e., a (pseudo-)quasi-metric (a metric that need not satisfy symmetry or the separation axiom), in \cite[Definition~6.1.1]{Goubault-Larrecq_2013}. Since the reflexivity condition implies the quasi-reflexivity condition, every quasi-metric is a quasi-quasi-metric. Indeed, if $\phi$ is a quasi-metric on a set $X$, then for all $x,y \in X$, we have
\begin{equation*}
  \phi(x, y) \sqsubseteq \mathbf{1} \sqsubseteq \phi(x,x).
\end{equation*}
Therefore, $\phi$ satisfies the quasi-reflexivity condition. Below, we give two examples of quasi-metrics.
\begin{exa}\label{eg:obsmet}
  Let $Q$ be a quantale, and let $\phi \colon X \times X \to Q$ be a $Q$-relation over a set $X$. We define $Q$-relations $\phi^{l}, \phi^{r} \colon X \times X \to Q$ by
  \begin{align*}
    \phi^{l}(x, y)
    &= \bigsqcap_{z \in X}
      (\phi(y, z) \multimap \phi(x, z)),
    \\
    \phi^{r}(x, y)
    &= \bigsqcap_{z \in X}
      (\phi(z, x) \multimap \phi(z, y)).
  \end{align*}
  It is straightforward to check that $\phi^{l}$ and $\phi^{r}$ are quasi-metrics on $X$. For reflexivity, $\mathbf{1} \sqsubseteq \phi(x, z) \multimap \phi(x,z)$ and $\mathbf{1} \sqsubseteq \phi(z,x) \multimap \phi(z,x)$ for all $x, z \in X$, and taking meets over $z$ gives reflexivity. For transitivity, for all $x,y,z,w \in X$, we have
  \begin{align*}
    \phi^{l}(x,y) \otimes \phi^{l}(y,z) \sqsubseteq
    (\phi(y,w) \multimap \phi(x, w)) \otimes (\phi(z,w) \multimap \phi(y,w)) &\sqsubseteq \phi(z,w) \multimap \phi(x,w), \\
    \phi^{r}(x,y) \otimes \phi^{r}(y,z) \sqsubseteq
    (\phi(w,x) \multimap \phi(w, y)) \otimes (\phi(w,y) \multimap \phi(w,z)) &\sqsubseteq \phi(w,x) \multimap \phi(w,z),
  \end{align*}
  and taking meets over $w$ gives transitivity. We call $\phi^{l}$ and $\phi^{r}$ the \emph{left and right observational quasi-metrics} of $\phi$, respectively. Intuitively, $\phi^{l}(x, y)$ measures how $x$ is separated from $y$ by observing discrepancies $\phi(y, z)$ and $\phi(x, z)$ for every $z \in X$; similarly, $\phi^{r}(x, y)$ measures how $x$ is separated from $y$ in the other direction. We will use these examples to explore several variants of transitivity conditions in Section~\ref{sec:remark-trans-cond}.
\end{exa}

We then recall the definition of a non-symmetric variant of the partial metrics from \cite{matthews}, called partial \emph{quasi}-metrics \cite{KUNZI2006}. As we anticipated, partial quasi-metrics are metrics $\phi$ on a set $X$ where the usual reflexivity condition $\phi(x, x) = \mathbf{1}$ is replaced by the quasi-reflexivity condition $\phi(x, y) \sqsubseteq \phi(x, x)$ (and $\phi(x, y) \sqsubseteq \phi(y, y)$). Moreover, unlike quasi-quasi-metrics or quasi-metrics, partial quasi-metrics satisfy a \emph{strong} transitivity condition. For a $[0,+\infty]$-relation $\phi \colon X \times X \to [0,+\infty]$, this condition reads as
\begin{align*}
  % \label{eq:pqm}
  &\phi(x, y) \lawotimes \phi(y, z)
  \lawsqsubseteq \phi(x, z) \lawotimes \phi(y, y).
  \tag{Strong transitivity in $[0,+\infty]$}
\end{align*}
This means that the self-distance $\phi(y, y)$ of the central term $y$ is ``added to the distance $\phi(x, z)$ in the triangle inequality''. By integrality of $[0,+\infty]$, the strong transitivity condition implies the transitivity condition. Indeed, for all $x, y, z \in X$, we have
\begin{equation*}
  \phi(x, y) \lawotimes \phi(y, z)
  \lawsqsubseteq \phi(x, z) \lawotimes \phi(y, y)
  \lawsqsubseteq \phi(x, z).
\end{equation*}
Thus, every partial quasi-metric is a quasi-quasi-metric; that is, the notion of quasi-quasi-metrics subsumes that of partial quasi-metrics.

%%% Local Variables:
%%% mode: LaTeX
%%% TeX-master: "lmcs_main.tex"
%%% End:

\section{The Category of Quasi-Quasi-Metric Spaces}
\label{sec:category-quasi-quasi}

\subsection{Cartesian Closed Structure}
\label{sec:cart-clos-struct}

We introduce a cartesian closed category of quasi-quasi-metrics, observing a correspondence between the cartesian closed structure and the construction of differential logical relations. We define the category $\qqm$ of quasi-quasi-metric spaces as follows:
\begin{itemize}
\item
  objects are quasi-quasi-metric spaces $A = (\uset{A}, \mathcal{Q}_{A}, \rho_{A})$ such that $\rho_{A}$ is closed; and
\item
  morphisms from $A$ to $B$ are triples $(f,d,g)$ consisting of functions $f,g \colon \uset{A} \to \uset{B}$ and
  $d \colon \uset{A} \times \mathcal{Q}_{A} \to \mathcal{Q}_{B}$ subject to the following conditions:
  \begin{itemize}
  \item
    $d \in (\uset{A}, \mathcal{Q}_{A}) \rightarrowtriangle \mathcal{Q}_{B}$, i.e., $d(x, -) \colon \mathcal{Q}_{A} \to \mathcal{Q}_{B}$ is monotone for every $x \in \uset{A}$; and
  \item
    for every $(x,a,y) \in \rho_{A}$, we have $(fx, d(x,a), gy) \in \rho_{B}$ and $(fx, d(x,a), fy) \in \rho_{B}$.
  \end{itemize}
  Below, we simply say that $d \colon \uset{A} \times \mathcal{Q}_{A} \to \mathcal{Q}_{B}$ is monotone when $d$ satisfies the first condition in the definition of morphisms, and we refer to the second condition as the compatibility condition.
\end{itemize}

The identity morphism on an object $A$ is $(i_{A}, j_{A}, i_{A})$ consisting of the identity function $i_{A}$ on $\uset{A}$ and the second projection $j_{A} \colon \uset{A} \times \mathcal{Q}_{A} \to \mathcal{Q}_{A}$. The composition of $(f, d, g) \colon A \to B$ and $(h, e, k) \colon B \to C$ is $(h \circ f, c, k \circ g)$ where $c \colon \uset{A} \times \mathcal{Q}_{A} \to \mathcal{Q}_{C}$ is given by
\begin{equation*}
  c(x,a) = e(fx, d(x,a)).
\end{equation*}
Intuitively, $c(x, a)$ first computes the error $d(x, a)$ introduced by the first morphism, and then computes the total error $e(fx, d(x,a))$ produced by the second morphism with respect to the intermediate value $fx$ and the error $d(x, a)$.

Below, we show that $\qqm$ is cartesian closed. The cartesian closed structure closely matches the construction of differential logical relations in Section~\ref{sec:from-logic-relat}. We first show that $\qqm$ has a terminal object and products.
\begin{lem}\label{lem:terminal}
  The category $\qqm$ has a terminal object.
\end{lem}
\begin{proof}
  We define an object $T$ by
  \begin{equation*}
    \uset{T} = \{0\},
    \qquad
    \mathcal{Q}_{T} = \{\mathbf{1}\},
    \qquad
    \rho_{T} = \{(0, \mathbf{1}, 0)\}.
  \end{equation*}
  Since $\rho_{T}$ is equal to $\uset{T} \times \mathcal{Q}_{T} \times \uset{T}$, we can easily check that $T$ is an object of $\qqm$. We next check that $T$ is a terminal object. For an object $A$, since $\uset{T}$ and $\mathcal{Q}_{T}$ are singletons, the only possible morphism from $A$ to $T$ is $(!,t,!)$ given by
  \begin{equation*}
    !(x) = 0,
    \qquad
    t(x,a) = \mathbf{1}.
  \end{equation*}
  It remains to check that $(!, t, !)$ is a morphism from $A$ to $T$. Since $\mathcal{Q}_{T}$ is a singleton, $t$ is monotone. Furthermore, since $\rho_{T} = \uset{T} \times \mathcal{Q}_{T} \times \uset{T}$, we see that $(!,t,!)$ automatically satisfies the compatibility condition. Therefore, $(!,t,!)$ is the unique morphism from $A$ to $T$, and $T$ is a terminal object.
\end{proof}

\begin{lem}\label{lem:product}
  The category $\qqm$ has products.
\end{lem}
\begin{proof}
  Let $A$ and $B$ be objects of $\qqm$. We define an object $A \times B \in \qqm$ by
  \begin{equation*}
    \uset{A \times B} = \uset{A} \times \uset{B},
    \qquad
    \mathcal{Q}_{A \times B} = \mathcal{Q}_{A} \times \mathcal{Q}_{B},
  \end{equation*}
  and
  \begin{equation*}
    ((x,y),(a,b),(z,w))
    \in
    \rho_{A \times B}
    \iff
    (x,a,z) \in \rho_{A}
    \text{ and }
    (y,b,w) \in \rho_{B},
  \end{equation*}
  where we endow $\mathcal{Q}_{A \times B}$ with the componentwise quantale structure. We first verify that $A \times B$ is an object of $\qqm$. The quantale $\mathcal{Q}_{A \times B}$ is commutative and integral since both $\mathcal{Q}_{A}$ and $\mathcal{Q}_{B}$ are commutative and integral. Furthermore, $\rho_{A \times B}$ inherits quasi-reflexivity, transitivity, and closedness from $\rho_A$ and $\rho_B$ componentwise. Hence, $\rho_{A \times B}$ is a closed quasi-quasi-metric. We next define the projections $\mathrm{fst}_{A, B} \colon A \times B \to A$ and $\mathrm{snd}_{A, B} \colon A \times B \to B$ by $\mathrm{fst}_{A, B} = (p_{A,B},r_{A,B},p_{A,B})$ and $\mathrm{snd}_{A, B} = (q_{A,B},s_{A,B},q_{A,B})$ where the functions
  \begin{align*}
    p_{A,B} &\colon \uset{A \times B} \to \uset{A},
    &
    r_{A,B} &\colon \uset{A \times B} \times \mathcal{Q}_{A \times B} \to \mathcal{Q}_{A},
    \\
    q_{A,B} &\colon \uset{A \times B} \to \uset{B},
    &
    s_{A,B} &\colon \uset{A \times B} \times \mathcal{Q}_{A \times B} \to \mathcal{Q}_{B}
  \end{align*}
  are given by
  \begin{align*}
    p_{A,B}(x, y) &= x,
    &
    r_{A,B}((x, y), (a, b)) &= a,
    \\
    q_{A,B}(x,y) &= y,
    &
    s_{A,B}((x,y),(a,b)) &= b.
  \end{align*}
  These projections are indeed morphisms in $\qqm$: the functions $r_{A,B}$ and $s_{A,B}$ are monotone since both functions are projections; and for every $((x,y),(a,b),(z,w)) \in \rho_{A \times B}$, it follows from the definition of $\rho_{A \times B}$ that we have
  \begin{align*}
    (p_{A,B}(x,y),r_{A, B}((x,y),(a,b)),p_{A,B}(z,w))
    &= (x,a,z)
    \in \rho_{A},
    \\
    (q_{A,B}(x,y),s_{A, B}((x,y),(a,b)),q_{A,B}(z,w))
    &= (y,b,w)
    \in \rho_{B}.
  \end{align*}
  Thus, the compatibility conditions for $\mathrm{fst}_{A, B}$ and $\mathrm{snd}_{A, B}$ hold. Finally, we prove that $A \times B$ with these projections satisfies the universal property. Given any pair of morphisms $(f,d,g) \colon C \to A$ and $(k,e,h) \colon C \to B$, we define their pairing $(u, c, v) \colon C \to A \times B$ by
  \begin{equation*}
    u(x) = (fx,kx),
    \qquad
    v(x) = (gx,hx),
    \qquad
    c(x, a) = (d(x, a), e(x, a)).
  \end{equation*}
  We verify that the tuple $(u,c,v)$ is a morphism from $C$ to $A \times B$: monotonicity of $c$ follows from monotonicity of $d$ and $e$; and for every $(x,a,y) \in \rho_{C}$, since $(fx, d(x,a),gy) \in \rho_{A}$ and $(kx,e(x,a),hy) \in \rho_{B}$, we have
  \begin{equation*}
    (ux, c(x, a), vy)
    = ((fx,kx), (d(x, a), e(x, a)), (gy,hy)) \in \rho_{A \times B}.
  \end{equation*}
  Similarly, by using quasi-reflexivity, we can show that $(ux, c(x, a), uy) \in \rho_{A \times B}$. It remains to check the universal property. By construction, $\mathrm{fst}_{A,B} \circ (u, c, v) = (f, d, g)$ and $\mathrm{snd}_{A,B} \circ (u, c, v) = (k, e, h)$. To see uniqueness, suppose that $(u', c', v') \colon C \to A \times B$ is a morphism such that $\mathrm{fst}_{A, B} \circ (u', c', v') = (f, d, g)$ and $\mathrm{snd}_{A, B} \circ (u', c', v') = (k, e, h)$. Then, for all $x \in \uset{C}$ and $a \in \mathcal{Q}_{C}$, we have
  \begin{equation*}
    u'(x) = (fx,kx)
    \qquad
    v'(x) = (gx,hx)
    \qquad
    c'(x,a) = (d(x,a),e(x,a)).
  \end{equation*}
  Thus, $(u',c',v') = (u,c,v)$.
\end{proof}

Given morphisms $(f, d, g) \colon A \to B$ and $(k, e, h) \colon C \to D$, their product morphism is defined as $(u, c, v) \colon A \times C \to B \times D$, where the component functions are given by
\begin{equation*}
  u(x, y) = (fx, ky),
  \qquad
  v(x, y) = (gx, hy),
  \qquad
  c((x, y),(a,b)) = (d(x, a), e(y, b)).
\end{equation*}
Indeed, the definition of the product morphism is derived by applying the pairing construction in the proof of Lemma~\ref{lem:product} to the compositions $(f',d',g') = (f,d,g) \circ \mathrm{fst}_{A,C}$ and $(k',e',h') = (k,e,h) \circ \mathrm{snd}_{A,C}$. To be concrete, these morphisms are given as follows:
\begin{align*}
  f'(x,y)
  &= fx,
  &
    g'(x,y)
  &= gx,
  &
    d'((x,y),(a,b))
  &= d(x, a),
  \\
  k'(x,y)
  &= ky,
  &
    h'(x,y)
  &= hy,
  &
    e'((x,y),(a,b))
  &= e(y, b).
\end{align*}
% To confirm that $(u,c,v)$ is the product morphism, it suffices to verify that we have
% \begin{equation*}
%   \mathrm{fst}_{B, D} \circ (u, c, v) = (f, d, g) \circ \mathrm{fst}_{A, C},
%   \qquad
%   \mathrm{snd}_{B, D} \circ (u, c, v) = (k, e, h) \circ \mathrm{snd}_{A, C}.
% \end{equation*}
% Indeed, evaluating these compositions yeilds
% \begin{align*}
%   p_{B,D}(u(x, y))
%   &= fx = f(p_{A,C}(x, y)),
%   &
%   q_{B,D}(u(x,y))
%   &= ky = k(q_{A,C}(x,y)),
%   \\
%   p_{B,D}(v(x, y))
%   &= gx = g(p_{A,C}(x, y)),
%   &
%   q_{B,D}(v(x,y))
%   &= hy = h(q_{A,C}(x,y)),
% \end{align*}
% and
% \begin{align*}
%   r_{B,D}(u(x,y), c((x,y),(a,b)))
%   &= d(x, a) = d(p_{A,C}(x,y),r_{A,C}((x,y),(a,b))),
%   \\
%   s_{B,D}(u(x,y), c((x,y),(a,b)))
%   &= e(y, b) = e(q_{A,C}(x,y),s_{A,C}((x,y),(a,b)))
% \end{align*}
% for all $x \in |A|$, $y \in |B|$, $a \in \mathcal{Q}_{A}$ and $b \in
% \mathcal{Q}_{B}$.
We denote the product of $(f,d,g)$ and $(k,e,h)$ by $(f,d,g) \times (k,e,h)$. We will use the product of morphisms in the proof of the following proposition.
\begin{prop}\label{prop:ccc}
  The category $\qqm$ is cartesian closed.
\end{prop}
\begin{proof}
  By Lemma~\ref{lem:terminal} and Lemma~\ref{lem:product}, $\qqm$ has finite products. It remains to construct an exponential object $A \Rightarrow B$ for every pair of objects $A$ and $B$ in $\qqm$. We define the object $A \Rightarrow B$ by
  \begin{equation*}
    \uset{A \Rightarrow B}
    = \uset{A} \Rightarrow \uset{B},
    \qquad
    \mathcal{Q}_{A \Rightarrow B}
    = (\uset{A}, \mathcal{Q}_{A}) \rightarrowtriangle \mathcal{Q}_{B},
  \end{equation*}
  and
  \begin{multline*}
    \qquad
    (f,d,g)
    \in
    \rho_{A \Rightarrow B}
    \iff
    \text{for every }
    (x,a,y) \in \rho_{A}, \,
    \\
    % h \in \{f,g\}, \,
    (fx, d(x,a), gy) \in \rho_{B}
    \text{ and }
    (fx, d(x,a), fy) \in \rho_{B}.
    \qquad
  \end{multline*}
  Here, $\uset{A} \Rightarrow \uset{B}$ is the set of functions from $\uset{A}$ to $\uset{B}$, and we endow $\mathcal{Q}_{A \Rightarrow B}$ with the pointwise quantale structure. We first verify that $A \Rightarrow B$ is an object of $\qqm$. The quantale $\mathcal{Q}_{A \Rightarrow B}$ is commutative and integral since $\mathcal{Q}_{B}$ is commutative and integral, and the quantale structure of $\mathcal{Q}_{A \Rightarrow B}$ is defined pointwise. We next show that $\rho_{A \Rightarrow B}$ is a closed quasi-quasi-metric:
  \begin{itemize}
  \item (Closedness)
    Since  the order of $\mathcal{Q}_{A \Rightarrow B}$ is pointwise, it follows from closedness of $\rho_{B}$ that $\rho_{A \Rightarrow B}$ is closed.
  \item (Quasi-reflexivity)
    The second clause of the definition of $\rho_{A \Rightarrow B}$ implies that $\rho_{A \Rightarrow B}$ is quasi-reflexive.
  \item (Transitivity)
    Let $(f, d, g), (g, e, h) \in \rho_{A \Rightarrow B}$. For every $(x, a, y) \in \rho_{A}$, we have $(x, a, x) \in \rho_{A}$ by the quasi-reflexivity of $\rho_{A}$. Hence,
    \begin{align}
      (fx, d(x, a), gx) &\in \rho_{B},
                          \label{eqn:fdg}
      \\
      (gx, e(x, a), hy) &\in \rho_{B},
                          \label{eqn:geh}
      \\
      (fx, d(x, a), fy) &\in \rho_{B}.
                          \label{eqn:fdf}
    \end{align}
    By applying transitivity of $\rho_{B}$ to \eqref{eqn:fdg} and \eqref{eqn:geh}, we obtain
    \begin{equation*}
      (fx, (d \otimes e)(x, a), hy)
      = (fx, d(x, a) \otimes e(x, a), hy) \in \rho_{B}.
    \end{equation*}
    Since $\rho_{B}$ is closed and $\mathcal{Q}_{B}$ is integral, it follows from \eqref{eqn:fdf} that
    \begin{equation*}
      (fx, (d \otimes e)(x, a), fy)
      = (fx, d(x, a) \otimes e(x, a), fy) \in \rho_{B}.
    \end{equation*}
    Therefore, $\rho_{A \Rightarrow B}$ is transitive.
  \end{itemize}
  We now define the evaluation morphism $\mathrm{eval}_{A, B} \colon (A \Rightarrow B) \times A \to B$ to be $(\epsilon_{A,B}, \delta_{A,B}, \epsilon_{A,B})$ where the functions
  \begin{equation*}
    \epsilon_{A,B} \colon \uset{(A \Rightarrow B) \times A} \to \uset{B},
    \qquad
    \delta_{A,B} \colon \uset{(A \Rightarrow B) \times A} \times \mathcal{Q}_{(A \Rightarrow B) \times A} \to \mathcal{Q}_{B}
  \end{equation*}
  are given by $\epsilon_{A,B}(f, x) = fx$ and $\delta_{A,B}((f,x),(d,a)) = d(x,a)$. We verify that $\mathrm{eval}_{A, B}$ is a morphism from $(A \Rightarrow B) \times A$ to $B$:
  \begin{itemize}
  \item (Monotonicity)
    For all $(f,x) \in \uset{(A \Rightarrow B) \times A}$ and $(d,a),(e,b) \in \mathcal{Q}_{(A \Rightarrow B) \times A}$, if $(d,a) \sqsubseteq (e,b)$, then
    \begin{equation*}
      \delta_{A,B}((f,x),(d,a)) = d(x,a) \sqsubseteq e(x, a) \sqsubseteq e(x, b)
      = \delta_{A,B}((f,x),(e,b)).
    \end{equation*}
    Hence, $\delta_{A, B}$ is monotone.
  \item (Compatibility)
    For every $((f,x),(d,a),(g,y)) \in \rho_{(A \Rightarrow B) \times A}$, we have
    \begin{align*}
      (\epsilon_{A,B}(f,x),\delta_{A,B}((f,x),(d,a)),\epsilon_{A,B}(g,y))
      = (fx, d(x,a),gy) \in \rho_{B}.
    \end{align*}
  \end{itemize}
  Thus, we see that $\mathrm{eval}_{A, B}$ is a morphism from $(A \Rightarrow B) \times A$ to $B$. Finally, we prove that $A \Rightarrow B$ with the evaluation morphism satisfies the universal property. Given a morphism $(f, d, g) \colon C \times A \to B$, we define its transpose $(f^{\wedge}, d^{\wedge}, g^{\wedge}) \colon C \to A \Rightarrow B$ by
  \begin{equation*}
    f^{\wedge}(x)(y) = f(x, y),
    \qquad
    g^{\wedge}(x)(y) = g(x, y),
    \qquad
    d^{\wedge}(x, a)(y, b) = d((x, y), (a, b)).
  \end{equation*}
  We note that since $d \colon \uset{C \times A} \times \mathcal{Q}_{C \times A} \to \mathcal{Q}_{B}$ is monotone, $d^{\wedge}(x,a)$ is an element of $\mathcal{Q}_{A \Rightarrow B}$ for all $x \in \uset{C}$ and $a \in \mathcal{Q}_{C}$. Let us verify that $(f^{\wedge},d^{\wedge},g^{\wedge})$ is a morphism from $C$ to $A \Rightarrow B$.
  \begin{itemize}
  \item (Monotonicity)
    By monotonicity of $d$, we see that $d^{\wedge}$ is monotone.
  \item (Compatibility)
    Let $(x,a,y) \in \rho_{C}$. For every $(z,b,w) \in \rho_{A}$, the two compatibility clauses of $(f,d,g)$ yield
    \begin{align*}
      (f^{\wedge}(x)(z),d^{\wedge}(x,a)(z,b),g^{\wedge}(y)(w))
      &= (f(x,z),d((x,z),(a,b)),g(y,w)) \in \rho_{B},
      \\
      (f^{\wedge}(x)(z),d^{\wedge}(x,a)(z,b),f^{\wedge}(x)(w))
      &= (f(x,z),d((x,z),(a,b)),f(x,w)) \in \rho_{B}.
    \end{align*}
    In the second clause, we used $(x, a, x) \in \rho_{C}$, which follows from quasi-reflexivity of $\rho_{C}$. Hence, we obtain
    \begin{equation*}
      (f^{\wedge}(x),d^{\wedge}(x,a),g^{\wedge}(y)) \in \rho_{A \Rightarrow B}.
    \end{equation*}
    We can similarly check that for every $(x,a,y) \in \rho_{C}$, we have
    \begin{equation*}
      (f^{\wedge}(x),d^{\wedge}(x,a),f^{\wedge}(y)) \in \rho_{A \Rightarrow B}.
    \end{equation*}
  \end{itemize}
  Hence, $(f^{\wedge},d^{\wedge},g^{\wedge})$ is a morphism from $C$ to $A \Rightarrow B$. We then check that $(f^{\wedge}, d^{\wedge}, g^{\wedge})$ is indeed a transpose of $(f, d, g)$. Let $(k,e,h)$ denote
  \begin{equation*}
    (\epsilon_{A,B}, \delta_{A,B}, \epsilon_{A,B}) \circ
    ((f^{\wedge},d^{\wedge},g^{\wedge}) \times (i_{A},j_{A},i_{A}))
    \colon C \times A \to B.
  \end{equation*}
  Recall that $(i_{A},j_{A},i_{A})$ is the identity morphism on $A$. For every $(x,y) \in \uset{C \times A}$ and $(a,b) \in \mathcal{Q}_{C \times A}$,
  \begin{align*}
    k(x,y)
    &= \epsilon_{A,B}(f^{\wedge}(x), y) = f(x,y),
    \\
    h(x,y)
    &= \epsilon_{A,B}(g^{\wedge}(x), y) = g(x,y),
    \\
    e((x,y),(a,b))
    &= \delta_{A,B}((f^{\wedge}(x),y),(d^{\wedge}(x,a),b)) = d((x,y),(a,b)).
  \end{align*}
  Therefore, $(k, e, h) = (f, d, g)$. To show uniqueness, suppose that $(u,c,v) \colon C \to A \Rightarrow B$ is a morphism such that
  \begin{equation*}
    (\epsilon_{A,B}, \delta_{A,B}, \epsilon_{A,B}) \circ
    ((u,c,v) \times (i_{A},j_{A},i_{A})) = (f,d,g).
  \end{equation*}
  For every $(x,y) \in \uset{C \times A}$ and $(a,b) \in \mathcal{Q}_{C \times A}$,
  \begin{align*}
    u(x)(y)
    &= \epsilon_{A,B}(ux, y) = f(x,y),
    \\
    v(x)(y)
    &= \epsilon_{A,B}(vx, y) = g(x,y),
    \\
    c(x,a)(y,b)
    &= \delta_{A,B}((ux,y),(c(x,a),b)) = d((x, y), (a, b)).
  \end{align*}
  Therefore, $(u,c,v) = (f^{\wedge},d^{\wedge},g^{\wedge})$.
\end{proof}

Every metric space induces a quasi-quasi-metric space. The following example corresponds to the Euclidean metric on $\mathbb{R}$.
\begin{exa}\label{eg:R}
  We define $R \in \qqm$ by
  \begin{equation*}
    \uset{R} = \mathbb{R},
    \qquad
    \mathcal{Q}_{R} = [0,+\infty],
    \qquad
    (x, a, y) \in \rho_{R}
    \iff
    |x - y| \lawsqsupseteq a.
  \end{equation*}
  We check that $\rho_{R}$ is a closed quasi-quasi-metric. Recall that the Lawvere quantale $[0,+\infty]$ is equipped with the reversed order. Therefore, by definition, $\rho_{R}$ is a closed $[0,+\infty]$-predicate, and the $[0,+\infty]$-relation $\hat{\rho}_{R} \colon \mathbb{R} \times \mathbb{R} \to [0,+\infty]$ associated with $\rho_{R}$ is the Euclidean metric $\hat{\rho}_{R}(x, y) = |x - y|$. Since reflexivity implies quasi-reflexivity, by Proposition~\ref{prop:rqt}, we conclude that $\rho_{R}$ is a quasi-quasi-metric.
  % Observe that the $[0,+\infty]$-relation $\hat{\rho}_{R} \colon \mathbb{R} \times \mathbb{R} \to [0,+\infty]$ associated with $\rho_{R}$ is the Euclidean metric $\hat{\rho}_{R}(x, y) = |x - y|$.
\end{exa}

\begin{exa}\label{eg:RR}
  To see how we measure distances between functions, we concretely describe the function space $R \Rightarrow R \in \qqm$. The underlying set of $R \Rightarrow R$ is the set of functions on $\mathbb{R}$, and $\mathcal{Q}_{R \Rightarrow R}$ is the set of functions $d \colon \mathbb{R} \times [0,+\infty] \to [0,+\infty]$ such that for every $x \in \mathbb{R}$, the function $d(x,-) \colon [0,+\infty] \to [0,+\infty]$ is monotone. For functions $f$ and $g$ on $\mathbb{R}$, we have $(f,d,g) \in \rho_{R \Rightarrow R}$ if and only if $d \in \mathcal{Q}_{R \Rightarrow R}$ satisfies
  \begin{equation*}
    d(x,a) \lawsqsubseteq
    \lawbigsqcap_{|x - y| \lawsqsupseteq a}
    \lawsqcapparen{|fx - gy|}{|fx - fy|},
  \end{equation*}
  i.e., $d(x,a)$ bounds gaps between $fx$ and values of $g$ and $f$ on the disk with center $x$ and radius $a$. The unit element $\mathbf{1} \in \mathcal{Q}_{R \Rightarrow R}$ is the constant function $\mathbf{1}(x, a) = 0$, and we have $(f, \mathbf{1}, f) \in \rho_{R \Rightarrow R}$ if and only if $f$ is a constant function. In particular, $\rho_{R \Rightarrow R}$ is not reflexive.
\end{exa}
We can generalize the characterization of $\rho_{R \Rightarrow R}$ in Example~\ref{eg:RR} as follows.
\begin{prop}\label{lem:arrow}
  Let $A$ and $B$ be objects of $\qqm$. For all $f, g \in \uset{A \Rightarrow B}$, $x \in \uset{A}$, and $a \in \mathcal{Q}_{A}$, we have
  \begin{equation*}
    \hat{\rho}_{A \Rightarrow B}(f, g)(x, a)
    = \bigsqcap_{(x, a, y) \in \rho_{A}}
    \bigl(\hat{\rho}_{B}(fx, gy) \sqcap \hat{\rho}_{B}(fx, fy)\bigr).
  \end{equation*}
\end{prop}
\begin{proof}
  We denote the right hand side by $d(x, a)$. We first check
  \begin{math}
    \hat{\rho}_{A \Rightarrow B}(f, g)(x, a)
    \sqsubseteq d(x, a).
  \end{math}
  By Proposition~\ref{prop:rhd}, $(f, \hat{\rho}_{A \Rightarrow B}(f, g), g) \in \rho_{A \Rightarrow B}$. Hence, for every $y \in \uset{A}$ such that $(x, a, y) \in \rho_{A}$, we have
  \begin{equation*}%\label{eq:hatrho}
    (fx, \hat{\rho}_{A \Rightarrow B}(f, g)(x, a), gy) \in \rho_{B},
    \qquad
    (fx, \hat{\rho}_{A \Rightarrow B}(f, g)(x, a), fy) \in \rho_{B}.
  \end{equation*}
  Thus,
  \begin{equation*}
    \hat{\rho}_{A \Rightarrow B}(f, g)(x, a) \sqsubseteq
    \hat{\rho}_{B}(fx, gy) \sqcap \hat{\rho}_{B}(fx, fy).
  \end{equation*}
  Since this inequality holds for every $y \in \uset{A}$ satisfying $(x, a, y) \in \rho_{A}$, we obtain
  \begin{equation*}
    \hat{\rho}_{A \Rightarrow B}(f, g)(x, a) \sqsubseteq d(x, a).
  \end{equation*}
  We next check the other inequality. Let $(x, a, y) \in \rho_{A}$. By Proposition~\ref{prop:rhd}, we have
  \begin{equation*}
    (fx, \hat{\rho}_{B}(fx, gy), gy) \in \rho_{B},
    \qquad
    (fx, \hat{\rho}_{B}(fx, fy), fy) \in \rho_{B}.
  \end{equation*}
  Hence, since $\rho_{B}$ is closed, by the definition of $d$, we obtain
  \begin{equation*}
    (fx, d(x, a), gy) \in \rho_{B},
    \qquad
    (fx, d(x, a), fy) \in \rho_{B}
  \end{equation*}
  for arbitrary $(x, a, y) \in \rho_{A}$. Moreover, closedness of $\rho_{A}$ implies that $d(x,-)$ is monotone for every $x \in \uset{A}$. Hence, $d$ is an element of $\mathcal{Q}_{A \Rightarrow B}$. Thus, $(f, d, g) \in \rho_{A \Rightarrow B}$, and we obtain $d \sqsubseteq \hat{\rho}_{A \Rightarrow B}(f, g)$.
\end{proof}

\begin{exa}\label{eg:id-sin}
  We observe that $\hat{\rho}_{R \Rightarrow R}$ is not symmetric. Let $\mathrm{id}_{\mathbb{R}}$ be the identity function on $\mathbb{R}$. It follows from Proposition~\ref{lem:arrow} that we have
  \begin{align*}
    \hat{\rho}_{R \Rightarrow R}
    (\sin, \mathrm{id}_{\mathbb{R}})(x, a)
    &=
      \lawbigsqcap_{|x - y| \lawsqsupseteq a}
      \lawsqcapparen{|\sin(x) - y|}{|\sin(x) - \sin(y)|},
    \\
    \hat{\rho}_{R \Rightarrow R}
    (\mathrm{id}_{\mathbb{R}}, \sin)(x, a)
    &=
    \lawbigsqcap_{|x - y| \lawsqsupseteq a}
      \lawsqcapparen{|x - \sin(y)|}{a}.
  \end{align*}
  Since
  \begin{equation*}
    \hat{\rho}_{R \Rightarrow R}
    (\sin, \mathrm{id}_{\mathbb{R}})
    (\pi, 2\pi)
    = 3\pi
    \lawsqsubsetneq
    2\pi =
    \hat{\rho}_{R \Rightarrow R}
    (\mathrm{id}_{\mathbb{R}}, \sin)
    (\pi, 2\pi),
  \end{equation*}
  $\hat{\rho}_{R \Rightarrow R}$ is not symmetric.
\end{exa}

We show that $\qqm$ has ``\emph{weak}'' coproducts. It remains open whether $\qqm$ has coproducts or not. As preparation, we provide a construction that adjoins a new least element to a quantale. For a quantale $Q$, we define a set $LQ$ by
\begin{equation*}
  LQ = \{\varnothing\} \cup \{\{a\} \mid a \in Q\},
\end{equation*}
and we endow $LQ$ with the following partial order:
\begin{equation*}
  \varnothing \sqsubseteq \{a\},
  \qquad
  \{a\} \sqsubseteq \{b\}
  \iff
  a \sqsubseteq b,
\end{equation*}
and the multiplication of $LQ$ given by extending the multiplication of $Q$ as follows:
\begin{equation*}
  \{a\} \otimes \{b\} = \{a \otimes b\}, \qquad
  \{a\} \otimes \varnothing = \varnothing \otimes \{a\} = \varnothing \otimes \varnothing = \varnothing.
\end{equation*}
The unit element of $LQ$ is $\{\mathbf{1}\}$.

\begin{lem}
  For every quantale $Q$, the set $LQ$ with the partial order $\sqsubseteq$ and the multiplicative structure $(\mathbf{1}, \otimes)$ is a quantale.
\end{lem}
\begin{proof}
  We first show that $LQ$ is a complete lattice. Since $\varnothing$ is the least element of $LQ$, the empty set has the least upper bound $\varnothing \in LQ$. Let $U \subseteq LQ$ be a non-empty subset. If $U = \{\varnothing\}$, then the least upper bound of $U$ is $\varnothing$. If $U \setminus \{\varnothing\}$ is not empty, then for every element $\alpha \in LQ$,
   \begin{align*}
     \forall \beta \in U,\, \beta \sqsubseteq \alpha
     &\iff \exists a \in Q,\, \alpha = \{a\} \text{ and }
     \forall b \in Q,\, \text{if } \{b\} \in U
     \text{ then } b \sqsubseteq a
     \\
     &\iff \exists a \in Q,\, \alpha = \{a\} \text{ and }
     \bigsqcup_{\{b\} \in U} b \sqsubseteq a
     \\
     &\iff
     \left\{\bigsqcup_{\{b\} \in U} b \right\}
     \sqsubseteq \alpha.
   \end{align*}
   Hence, the least upper bound $\bigsqcup U$ is given by
  \begin{equation*}
    \bigsqcup U =
    \left\{\bigsqcup_{\{b\} \in U} b\right\}.
  \end{equation*}
  We next check that for every $\alpha \in LQ$, the map $\alpha \otimes (-) \colon LQ \to LQ$ is continuous. The map $\varnothing \otimes (-) \colon LQ \to LQ$ is continuous since $\varnothing \otimes (-)$ is a constant map. Thus, it suffices to check continuity of $\{a\} \otimes (-) \colon LQ \to LQ$ for every $a \in Q$. Let $U \subseteq LQ$. There are three cases:
  \begin{itemize}
  \item
    If $U$ is empty, then
    \begin{equation*}
      \{a\} \otimes \bigsqcup U = \{a\} \otimes \varnothing = \varnothing =
      \bigsqcup \varnothing = \bigsqcup \{\{a\} \otimes \alpha \mid \alpha \in U\}.
    \end{equation*}
  \item
    If $U = \{\varnothing\}$, then
    \begin{equation*}
      \{a\} \otimes \bigsqcup U = \{a\} \otimes \varnothing = \varnothing = \bigsqcup \{\varnothing\} = \bigsqcup \{\{a\} \otimes \varnothing\} = \bigsqcup \{\{a\} \otimes \alpha \mid \alpha \in U\}.
    \end{equation*}
  \item
    If there is $b \in Q$ such that $\{b\} \in U$, then
    \begin{equation*}
      \{a\} \otimes \bigsqcup U
      = \left\{
        a \otimes \bigsqcup_{\{b\} \in U} b
      \right\}
      = \left\{
        \bigsqcup_{\{b\} \in U} a \otimes b
      \right\}
      = \bigsqcup \{\{a\} \otimes \alpha \mid \alpha \in U\}.
    \end{equation*}
  \end{itemize}
  Finally, the multiplication of $LQ$ is commutative since the multiplication of $Q$ is commutative. The unit laws $\{\mathbf{1}\} \otimes \alpha = \alpha \otimes \{\mathbf{1}\} = \alpha$ follow from the definition of the multiplication of $LQ$ and the unit laws $\mathbf{1} \otimes a = a \otimes \mathbf{1} = a$. The unit $\{\mathbf{1}\}$ is the greatest element of $LQ$ since $Q$ is integral. Hence, $LQ$ is integral.
\end{proof}

\begin{prop}
  The category $\qqm$ has weak coproducts: for all objects $A$ and $B$ in $\qqm$, there is an object $A + B$ in $\qqm$ equipped with a pair of a section and a retraction:
  \begin{equation*}
    s_{C} \colon \qqm(A,C) \times \qqm(B,C) \leftrightarrows \qqm(A + B,C) :\! r_{C}
    \qquad
    (C \in \qqm)
  \end{equation*}
  such that the retraction $r_{C}$ is natural in $C$.
\end{prop}
\begin{proof}
  For $A$ and $B$ in $\qqm$, let $A + B$ be
  \begin{align*}
    \uset{A + B}
    &= (\{0\} \times \uset{A})
      \cup (\{1\} \times \uset{B}),
    \\
    \mathcal{Q}_{A + B}
    &= L\mathcal{Q}_{A}
      \times
      L\mathcal{Q}_{B},
    \\
    \rho_{A + B}
    &= \{(u, (\varnothing, \varnothing), v)
      \mid u,v \in \uset{A + B}\}
    \\
    &\phantom{=}
      \quad \cup
      \{((0,x),(\{a\},\varnothing),(0,y))
      \mid (x,a,y) \in \rho_{A}\}
    \\
    &\phantom{=}
      \quad \cup
      \{((1,z),(\varnothing,\{b\}),(1,w))
      \mid (z,b,w) \in \rho_{B}\}
      .
  \end{align*}
  We first verify that $\rho_{A + B}$ is a closed quasi-quasi-metric:
  \begin{itemize}
  \item (Quasi-reflexivity)
    It follows from the quasi-reflexivity of $\rho_{A}$ and $\rho_{B}$ that $\rho_{A + B}$ is quasi-reflexive.
  \item (Transitivity)
    Let $((i,x),(\alpha,\beta),(j,y)),((j,y),(\gamma,\delta),(k,z)) \in \rho_{A + B}$. When $(\alpha,\beta) = (\varnothing, \varnothing)$ or $(\gamma,\delta) = (\varnothing, \varnothing)$,
    \begin{equation*}
      ((i,x),(\alpha \otimes \gamma, \beta \otimes \delta), (k,z))
      =
      ((i,x),(\varnothing,\varnothing),(k,z))
      \in \rho_{A + B}.
    \end{equation*}
    Otherwise, we have $i = j = k$, and $((i,x),(\alpha \otimes \gamma, \beta \otimes \delta), (k,z))$ follows from the transitivity of $\rho_{A}$ and
    $\rho_{B}$.
  \item (Closedness)
    The closedness of $\rho_{A + B}$ follows from the closedness of $\rho_{A}$ and $\rho_{B}$, together with $(u,(\varnothing,\varnothing),v) \in \rho_{A + B}$ for all $u,v \in \uset{A + B}$.
  \end{itemize}
 Hence, $\rho_{A + B}$ is a closed quasi-quasi-metric.

 We next construct a section and a retraction. Given a morphism $(f, d, g) \colon A + B \to C$, we define $r_{C}(f, d, g)$ to be the pair of morphisms $(h, e, k) \colon A \to C$ and $(h', e', k') \colon B \to C$ obtained by restricting the domains of $f$, $d$, and $g$: we define $(h, e, k) \colon A \to C$ by
  \begin{equation*}
    h(x) = f(0, x), \qquad
    e(x, a) = d((0,x),(\{a\},\varnothing)), \qquad
    k(x) = g(0, x),
  \end{equation*}
  and we define $(h', e', k')$ in the same manner.
  By monotonicity of $d$, the map $e$ is monotone;
  it follows from compatibility of $(f,d,g)$ that $(h, e, k)$ satisfies the compatibility condition. Hence, $(h, e, k)$ is a morphism from $A$ to $C$. In the same way, we can check that $(h',e',k')$ is a morphism from $B$ to $C$. Conversely, given morphisms $(f, d, g) \colon A \to C$ and $(h, e, k) \colon B \to C$, we define $s_{C}((f,d,g),(h,e,k))$ to be
  \begin{equation*}
    ([f,h], \lbag d,e \rbag, [g,k]) \colon A + B \to C
  \end{equation*}
  consisting of the cotuplings $[f,h],[g,k] \colon \uset{A + B} \to \uset{C}$ given by
  \begin{align*}
    [f,h](0,x)
    &= f(x), &
    [g,k](0,x)
    &= g(x), \\
    [f, h](1,z)
    &= h(z), &
    [g, k](1,z)
    &= k(z),
  \end{align*}
  and the function ${\lbag d,e \rbag} \colon \uset{A + B} \times \mathcal{Q}_{A + B} \to \mathcal{Q}_{C}$ given by
  \begin{align*}
    {\lbag d,e \rbag}((0,x),(\alpha,\beta))
    &=
    \begin{cases}
      d(x, a), & \text{if $\alpha = \{a\}$}, \\
      \bot, & \text{otherwise},
    \end{cases}
    \\
    {\lbag d,e \rbag}((1, z), (\alpha,\beta))
    &=
    \begin{cases}
      e(z, b), & \text{if $\beta = \{b\}$}, \\
      \bot, & \text{otherwise},
    \end{cases}
  \end{align*}
  where $\bot$ is the least element of $\mathcal{Q}_{C}$. We verify that $([f,h],{\lbag d,e \rbag},[g,k])$ is a morphism from $A + B$ to $C$. We first check monotonicity of $\lbag d, e \rbag$. Suppose that we are given $x \in \uset{A}$ and $(\alpha,\beta),(\gamma,\delta) \in \mathcal{Q}_{A + B}$ such that $(\alpha,\beta) \sqsubseteq (\gamma,\delta)$. There are two cases to consider.
  \begin{itemize}
    \item
      If $\alpha = \{a\}$ for some $a \in \mathcal{Q}_{A}$, then $\gamma = \{b\}$ for some $b \in \mathcal{Q}_{A}$ satisfying $a \sqsubseteq b$. Hence,
      \begin{align*}
        {\lbag d, e \rbag} ((0,x),(\alpha,\beta))
        &= d(x, a)
        \sqsubseteq d(x, b)
        = {\lbag d, e \rbag}((0,x),(\gamma,\delta)).
        % \lbag d, e \rbag ((1,z),(\alpha,\beta))
        % &= \bot
        % = \lbag d, e \rbag((1,z),(\gamma,\delta)), \\
      \end{align*}
    % \item
    %   If $\beta = \{b\}$ for some $b \in \mathcal{Q}_{B}$, then $\delta = \{b'\}$ for some $b' \in \mathcal{Q}_{B}$ satisfying $b \sqsubseteq b'$. Hence,
    %   \begin{align*}
    %     \lbag d, e \rbag ((0,x),(\alpha,\beta))
    %     &=
    %     \bot
    %     = \lbag d, e \rbag((0,x),(\gamma,\delta)),
    %     \\
    %     \lbag d, e \rbag ((1,z),(\alpha,\beta))
    %     &= e(z, b)
    %     \sqsubseteq e(z, b')
    %     = \lbag d, e \rbag((1,z),(\gamma,\delta)).
    %   \end{align*}
    \item
      If $\alpha = \varnothing$, then
      \begin{align*}
        {\lbag d, e \rbag}((0,x),(\alpha,\beta))
        &= \bot
        \sqsubseteq {\lbag d, e \rbag}((0,x),(\gamma,\delta)).
        % \lbag d, e \rbag((1,z),(\alpha,\beta))
        % &= \bot
        % \sqsubseteq \lbag d, e \rbag((1,z),(\gamma,\delta)).
      \end{align*}
  \end{itemize}
  Hence, ${\lbag d,e \rbag}((0,x),-)$ is monotone for every $x \in \uset{A}$. We can similarly check the monotonicity of ${\lbag d,e \rbag}((1,z),-)$ for every $z \in \uset{B}$. Thus, ${\lbag d, e \rbag}$ is monotone. We next check the compatibility condition for $([f,h],{\lbag d, e \rbag}, [g,k])$. Let $(u,(\alpha,\beta),v) \in \rho_{A + B}$. We need to show
  \begin{equation*}
    ([f,h](u), {\lbag d,e \rbag}(u,(\alpha,\beta)),[g,k](v))
    \in \rho_{C},
    \qquad
    ([f,h](u), {\lbag d,e \rbag}(u,(\alpha,\beta)),[f,h](v))
    \in \rho_{C}.
  \end{equation*}
  Here, we only check the first clause as the second clause is derived in the same manner. We have three cases:
  \begin{itemize}
    \item
      If $(\alpha,\beta) = (\varnothing,\varnothing)$, then since $\rho_{C}$ is closed, we have
      \begin{equation*}
        ([f,h](u), {\lbag d,e \rbag}(u,(\alpha,\beta)),[g,k](v))
        =
        ([f,h](u), \bot, [g,k](v))
        \in \rho_{C}.
      \end{equation*}
    \item
      If $(u,(\alpha,\beta),v) = ((0,x),(\{a\},\varnothing),(0,y))$, then since $(f,d,g)$ is a morphism from $A$ to $C$ and $(x, a,y) \in \rho_{A}$, we have
      \begin{equation*}
        ([f,h](u), {\lbag d,e \rbag}(u,(\alpha,\beta)),[g,k](v))
        =
        (fx, d(x,a), gy)
        \in \rho_{C}.
      \end{equation*}
    \item
      If $(u,(\alpha,\beta),v) = ((1,z),(\varnothing, \{b\}),(1,w))$, then since $(h,e,k)$ is a morphism from $B$ to $C$ and $(z, b,w) \in \rho_{B}$, we have
      \begin{equation*}
        ([f,h](u), {\lbag d,e \rbag}(u,(\alpha,\beta)),[g,k](v))
        =
        (hz, e(z,b), kw)
        \in \rho_{C}.
      \end{equation*}
  \end{itemize}
  Thus, $([f,h],{\lbag d, e \rbag}, [g, k])$ is a morphism from $A + B$ to $C$.

  We then check that $r_{C}$ is a left inverse of $s_{C}$. For all $(f,d,g) \colon A \to C$ and $(h,e,k) \colon B \to C$, let
  \begin{equation*}
    ((f',d',g'), (h',e',k'))  =
    r_{C}(s_{C}((f,d,g),(h,e,k))).
  \end{equation*}
  By the definition of $r_{C}$ and $s_{C}$, for all $x \in \uset{A}$ and $a \in \mathcal{Q}_{A}$, we have
  \begin{align*}
    f'x &= [f,h](0,x) = fx, \\
    g'x &= [g,k](0,x) = gx, \\
    d'(x,a) &= {\lbag d,e \rbag}((0,x),(\{a\},\varnothing)) = d(x,a).
  \end{align*}
  Hence, $(f',d',g') = (f,d,g)$. In the same way, we can show that $(h',e',k') = (h,e,k)$. Therefore, $r_{C}$ is a left inverse of $s_{C}$. We finally check naturality of $r_{C}$. Suppose that we are given morphisms $(f,d,g) \colon A + B \to C$ and $(h,e,k) \colon C \to D$. Let
  \begin{align*}
    ((f_{1},d_{1},g_{1}), (f_{2},d_{2},g_{2}))
    &= r_{C}(f,d,g),
    \\
    ((u_{1},c_{1},v_{1}),(u_{2},c_{2},v_{2}))
    &= r_{D}((h,e,k) \circ (f,d,g)),
  \end{align*}
  and we define $(u'_{1},c'_{1},v'_{1})$ and $(u'_{2},c'_{2},v'_{2})$ by
  \begin{equation*}
    (u'_{1},c'_{1},v'_{1}) = (h,e,k) \circ (f_{1}, d_{1}, g_{1}),
    \qquad
    (u'_{2},c'_{2},v'_{2}) = (h,e,k) \circ (f_{2}, d_{2}, g_{2}).
  \end{equation*}
  We show that $(u'_{i},c'_{i},v'_{i}) = (u_{i},c_{i},v_{i})$ for $i \in \{1,2\}$. Let $x \in \uset{A}$ and $a \in \mathcal{Q}_{A}$. Then,
  \begin{align*}
    u'_{1}x &= (h \circ f_{1})(x) = h(f(0,x)) = u_{1}x, \\
    v'_{1}x &= (k \circ g_{1})(x) = k(g(0,x)) = v_{1}x, \\
    c'_{1}(x, a) &= e(f_{1}(x),d_{1}(x,a)) = e(f(0,x),d((0,x),(\{a\},\varnothing))) = c_{1}(x, a).
  \end{align*}
  We can similarly check that $(u_{2}',c_{2}',v_{2}') = (u_{2},c_{2},v_{2})$.
\end{proof}

\subsection{Quasi-Quasi-Metric Semantics}
\label{sec:target-language-its}

We introduce a simply typed $\lambda$-calculus $\LL$ together with its interpretation in the cartesian closed category $\qqm$. We show that the interpretation of types yields differential logical relations on $\LL$, and furthermore, we derive a fundamental lemma from the categorical interpretation.

\paragraph{Target Calculus}
\label{sec:target-language}

We describe our target calculus $\LL$ that comprises reals and functions of arbitrary arity on $\mathbb{R}$. To define our target calculus, we fix a set $\mathbb{F}$ of multi-arity functions of the form $\alpha \colon \mathbb{R}^{n} \to \mathbb{R}$ for positive integers $n$. We call $n$ the \emph{arity} of $\alpha \in \mathbb{F}$ and denote it by $\mathrm{ar}(\alpha)$.

\begin{figure}
  \footnotesize \centering
  \fbox{$\begin{array}{c}
    \prftree{
    r \in \mathbb{R}
    }{
    \mathsf{\Gamma}
    \vdash
    \underline{r} : \real
    }
    \qquad
    \prftree{
    \mathsf{\Gamma} \vdash
    \mathsf{t}_{1}
    : \real
    }{\ldots}{
    \mathsf{\Gamma} \vdash
    \mathsf{t}_{n}
    : \real
    }{
    \mathrm{ar}(\alpha) = n
    }{
    \mathsf{\Gamma} \vdash
    \alpha
    (\mathsf{t}_{1},
    \ldots,
    \mathsf{t}_{n})
    : \real
    }
    \\[7pt]
    \prftree{
    \mathsf{x}:\mathsf{A} \in \mathsf{\Gamma}
    }{
    \mathsf{\Gamma} \vdash
    \mathsf{x}:\mathsf{A}
    }
    \qquad
    \prftree{
    \mathsf{\Gamma} \vdash
    \mathsf{t}
    : \mathsf{A} \Rightarrow \mathsf{B}
    }{
    \mathsf{\Gamma} \vdash
    \mathsf{s}
    : \mathsf{A}
    }{
    \mathsf{\Gamma} \vdash
    \mathsf{t} \,
    \mathsf{s}
    : \mathsf{B}
    }
    \qquad
    \prftree{
    \mathsf{\Gamma},\mathsf{x}:\mathsf{A} \vdash
    \mathsf{t} : \mathsf{B}
    }{
    \mathsf{\Gamma} \vdash
    \lam{\mathsf{x}}{\mathsf{A}}{\mathsf{t}}
    : \mathsf{A} \Rightarrow \mathsf{B}
    }
    \\[7pt]
    \prftree{
    \mathsf{\Gamma} \vdash
    \mathsf{t} : \mathsf{A}
    }{
    \mathsf{\Gamma} \vdash
    \mathsf{s} : \mathsf{B}
    }{
    \mathsf{\Gamma} \vdash
    \pair{
    \mathsf{t}
    }{
    \mathsf{s}
    }
    : \mathsf{A} \times \mathsf{B}
    }
    \qquad
    \prftree{
    \mathsf{\Gamma} \vdash
    \mathsf{t}
    :\mathsf{A} \times \mathsf{B}
    }{
    \mathsf{\Gamma} \vdash
    \fst(\mathsf{t})
    :\mathsf{A}
    }
    \qquad
    \prftree{
    \mathsf{\Gamma} \vdash
    \mathsf{t}
    :\mathsf{A} \times \mathsf{B}
    }{
    \mathsf{\Gamma} \vdash
    \snd(\mathsf{t})
    :\mathsf{B}
    }
  \end{array}$}
\caption{Typing Rules}
\label{fig:typing-rules}
\end{figure}

Let $\mathsf{Var}$ be a countably infinite set of variables. We define \emph{type}s and \emph{term}s as follows:
\begin{align*}
  (\text{Type})
  && \mathsf{A},\mathsf{B}
  & \coloneq
    \real
    \mid \mathsf{A} \times \mathsf{B}
    \mid \mathsf{A} \Rightarrow \mathsf{B} , \\
  (\text{Term})
  && \mathsf{t},\mathsf{s}
  &\coloneq
    \mathsf{x} \in \mathsf{Var}
    \mid \underline{r}
    \mid \alpha
    (\mathsf{t}_{1},\ldots,\mathsf{t}_{n})
    \mid \mathsf{t} \, \mathsf{s}
    \mid \lam{\mathsf{x}}{\mathsf{A}}{\mathsf{t}}
    \mid \pair{\mathsf{t}}{\mathsf{s}}
    \mid \fst(\mathsf{t})
    \mid \snd(\mathsf{t})
  .
\end{align*}
In the definition of terms, $r$ varies over $\mathbb{R}$, and $\alpha$ varies over $\mathbb{F}$. We adopt the standard typing rules given in Figure~\ref{fig:typing-rules}. We write $\mathsf{\Gamma} \vdash \mathsf{t} : \mathsf{A}$ when the judgment is derivable from those rules. In what follows, we also use $\mathsf{u}, \mathsf{p}$ and $\mathsf{q}$ as meta-variables for terms. We denote the set of types by $\type$ and the set of closed terms of type $\mathsf{A}$ by $\term{\mathsf{A}}$.

\paragraph{Notations}

Let $\mathsf{\Gamma}$ be a typing environment of the form $\mathsf{\Gamma} = (\mathsf{x}_{1} : \mathsf{A}_{1}, \ldots, \mathsf{x}_{n} : \mathsf{A}_{n})$. For a type-indexed family of sets $\{X_{\mathsf{A}}\}_{\mathsf{A} \in \type}$, we define a set $X_{\mathsf{\Gamma}}$ by
\begin{equation*}
  X_{\mathsf{\Gamma}} =
  % \begin{cases}
    X_{\mathsf{A}_{1}}
    \times \cdots \times
    X_{\mathsf{A}_{n}},
  %   & \text{if $n \geq 1$}; \\
  %   \{\ast\}
  %   & \text{if $n = 0$};
  % \end{cases}
\end{equation*}
and given an element $\gamma \in X_{\mathsf{\Gamma}}$, we denote its $i$th element by $\gamma|_{i}$. For example, $\term{\mathsf{\Gamma}}$ denotes the set of finite sequences $\gamma = (\mathsf{t}_{1}, \ldots, \mathsf{t}_{n})$ consisting of closed terms $\gamma|_{i} = \mathsf{t}_{i} \in \term{\mathsf{A}_{i}}$. Furthermore, for a typing environment $\mathsf{\Gamma}$ and a type $\mathsf{B}$, given $\gamma = (x_{1},\ldots,x_{n}) \in X_{\mathsf{\Gamma}}$ and $y \in X_{\mathsf{B}}$, we write $\gamma \add y$ for $(x_{1},\ldots,x_{n},y) \in X_{\mathsf{\Delta}}$ where $\mathsf{\Delta} = (\mathsf{x}_{1}:\mathsf{A}_{1},\ldots,\mathsf{x}_{n}:\mathsf{A}_{n},\mathsf{y}:\mathsf{B})$.

Similarly, given a type-indexed family of quantale-valued predicates
\begin{equation*}
  \{r_{\mathsf{A}} \subseteq X_{\mathsf{A}} \times Q_{\mathsf{A}} \times X_{\mathsf{A}}\}_{\mathsf{A} \in \type},
\end{equation*}
we extend it to $\mathsf{\Gamma}$ by defining $r_{\mathsf{\Gamma}} \subseteq X_{\mathsf{\Gamma}} \times Q_{\mathsf{\Gamma}} \times X_{\mathsf{\Gamma}}$ by
\begin{equation*}
  (\gamma, \xi, \epsilon) \in r_{\mathsf{\Gamma}}
  \iff
  \text{for every } i \in \{1,\ldots,n\}, \,
  (\gamma|_{i},\xi|_{i},\epsilon|_{i}) \in r_{\mathsf{A}_{i}},
\end{equation*}
where we endow $Q_{\mathsf{\Gamma}}$ with the componentwise quantale structure. It is straightforward to see that if $r_{\mathsf{A}}$ are (closed) quasi-quasi-metrics for all types $\mathsf{A}$, then $r_{\mathsf{\Gamma}}$ is a (closed) quasi-quasi-metric on $X_{\mathsf{\Gamma}}$. We adopt the same convention for any type-indexed family of $k$-ary relations $\{s_{\mathsf{A}} \subseteq X_{\mathsf{A}}^{k}\}_{\mathsf{A} \in \type}$.

\paragraph{Set Theoretic Semantics}
\label{sec:set-theor-semant}

We recall the standard set theoretic interpretation of $\LL$, which is subsumed by the interpretation of $\LL$ in the cartesian closed category $\qqm$. The set theoretic interpretation associates types with sets and terms with functions as follows: we interpret a type $\mathsf{A}$ as a set $\sem{\mathsf{A}}$ inductively given by
\begin{equation*}
  \sem{\real} = \mathbb{R},
  \qquad
  \sem{\mathsf{A} \times \mathsf{B}} =
  \sem{\mathsf{A}} \times \sem{\mathsf{B}},
  \qquad
  \sem{\mathsf{A} \Rightarrow \mathsf{B}} =
  \sem{\mathsf{A}} \Rightarrow \sem{\mathsf{B}};
\end{equation*}
and we interpret a term $\mathsf{\Gamma} \vdash \mathsf{t} : \mathsf{A}$ as the function $\sem{\mathsf{t}}$ from $\sem{\mathsf{\Gamma}}$ to $\sem{\mathsf{A}}$ given in Figure~\ref{fig:semt}. For brevity, given a closed term $\vdash \mathsf{t} : \mathsf{A}$, we identify the interpretation $\sem{\mathsf{t}} \colon \{\ast\} \to \sem{\mathsf{A}}$ with its value $\semp{\mathsf{t}}{\ast} \in \sem{\mathsf{A}}$.

\begin{figure}
  \centering
  \fbox{
    $\begin{aligned}
      \semp{\mathsf{x}_{i}}{\gamma}
      &= \gamma|_{i} \\
      \semp{\underline{r}}{\gamma}
      &= r \\
      \semp{\alpha
      (\mathsf{t}_{1},\ldots,
      \mathsf{t}_{n})}{\gamma}
      &= \alpha(
        \semp{\mathsf{t}_{1}}{\gamma},
        \ldots,
        \semp{\mathsf{t}_{n}}{\gamma})
      \\
      \semp{\mathsf{t} \, \mathsf{s}}{\gamma}
      &=
        \semp{\mathsf{t}}{\gamma}
        (\semp{\mathsf{s}}{\gamma})
      \\
      \semp{\lam{\mathsf{x}}
      {\mathsf{A}}{\mathsf{t}}}{\gamma}(x)
      &= \semp{\mathsf{t}}{\gamma \add x}
      \\
      \semp{\fst(\mathsf{t})}{\gamma}
      &=
        \text{the first component of }
        \semp{\mathsf{t}}{\gamma}
      \\
      \semp{\snd(\mathsf{t})}{\gamma}
      &=
        \text{the second component of }
        \semp{\mathsf{t}}{\gamma}
      \\
      \semp{\pair{\mathsf{t}}{\mathsf{s}}}{\gamma}
      &= (
        \semp{\mathsf{t}}{\gamma},
        \semp{\mathsf{s}}{\gamma})
    \end{aligned}$}
  \caption{Set Theoretic Interpretation}
  \label{fig:semt}
\end{figure}

\paragraph{Quasi-Quasi-Metric Semantics}
\label{sec:qqm-semantics}

We now describe the interpretation of $\LL$ in the cartesian closed category $\qqm$. We interpret the ground type $\real$ as $R$ given in Example~\ref{eg:R} and lift the interpretation to all types as follows:
\begin{equation*}
  \psem{\real}
  = R,
  \qquad
  \psem{\mathsf{A} \times \mathsf{B}}
  = \psem{\mathsf{A}} \times \psem{\mathsf{B}},
  \qquad
  \psem{\mathsf{A} \Rightarrow \mathsf{B}}
  = \psem{\mathsf{A}} \Rightarrow \psem{\mathsf{B}}.
\end{equation*}
Below, we simply denote each component of $\psem{\mathsf{A}}$ by $(|\mathsf{A}|, \mathcal{Q}_{\mathsf{A}}, \rho_{\mathsf{A}})$. Since the underlying sets of $\psem{\mathsf{A} \times \mathsf{B}}$ and $\psem{\mathsf{A} \Rightarrow \mathsf{B}}$ are $\uset{\mathsf{A}} \times \uset{\mathsf{B}}$ and $\uset{\mathsf{A}} \Rightarrow \uset{\mathsf{B}}$ respectively,
%cartesian closed structure of $\qqm$ subsumes the cartesian closed structure of the category $\mathbf{Set}$ of sets and functions,
a straightforward induction on $\mathsf{A}$ shows that the underlying set $\uset{\mathsf{A}}$ coincides with its set-theoretic interpretation $\sem{\mathsf{A}}$. In this sense, the standard set theoretic interpretation of $\LL$ is subsumed by the interpretation of $\LL$ in $\qqm$. By expanding the definition of $\psem{\mathsf{A}}$, we obtain the following inductive definition of the quantales $\mathcal{Q}_{\mathsf{A}}$:
\begin{equation*}
  \mathcal{Q}_{\real}
  = [0,+\infty],
  \qquad
  \mathcal{Q}_{\mathsf{A} \times \mathsf{B}}
  = \mathcal{Q}_{\mathsf{A}}
  \times \mathcal{Q}_{\mathsf{B}},
  \qquad
  \mathcal{Q}_{\mathsf{A} \Rightarrow \mathsf{B}}
  =
  (\uset{\mathsf{A}}, \mathcal{Q}_{\mathsf{A}})
  \rightarrowtriangle
  \mathcal{Q}_{\mathsf{B}},
\end{equation*}
and the following inductive definition of the quasi-quasi-metrics $\rho_{\mathsf{A}}$:
\begin{align*}
  (x, a, y) \in \rho_{\real}
  &\iff |x - y| \lawsqsupseteq a,
  \\
  ((x, y), (a, b), (z, w))
  \in \rho_{\mathsf{A} \times \mathsf{B}}
  &\iff (x, a, z) \in \rho_{\mathsf{A}}
    \text{ and } (y, b, w) \in \rho_{\mathsf{B}},
  \\
  (f, d, g)
  \in \rho_{\mathsf{A} \Rightarrow \mathsf{B}}
  &\iff \text{for every } (x, a, y) \in \rho_{\mathsf{A}},
  \\
  &\mathbin{\phantom{\iff}}
    \qquad
    (fx, d(x, a), gy) \in \rho_{\mathsf{B}}
    \text{ and }
    (fx, d(x, a), fy) \in \rho_{\mathsf{B}}.
\end{align*}
We can check by induction on the typing derivation that the interpretation $\psem{\mathsf{t}} \colon \psem{\mathsf{\Gamma}} \to \psem{\mathsf{A}}$ of a term $\mathsf{\Gamma} \vdash \mathsf{t} : \mathsf{A}$ is equal to $(\sem{\mathsf{t}}, \fsem{\mathsf{t}}, \sem{\mathsf{t}})$, where $\psem{\mathsf{\Gamma}} = (\uset{\mathsf{\Gamma}}, \mathcal{Q}_{\mathsf{\Gamma}},\rho_{\mathsf{\Gamma}})$. The function $\fsem{\mathsf{t}} \in (|\mathsf{\Gamma}|, \mathcal{Q}_{\mathsf{\Gamma}}) \rightarrowtriangle \mathcal{Q}_{\mathsf{A}}$ is called the \emph{derivative} \cite{DLG21}, and we give the definition of $\fsem{\mathsf{t}}$ in Figure~\ref{fig:diff}. The function
\begin{equation*}
  \prt{\alpha} \colon R^{n} \times [0,+\infty]^{n} \to [0,+\infty]
\end{equation*}
in the definition of $\fsem{\alpha(\mathsf{t}_{1},\ldots,\mathsf{t}_{n})}$ is defined by
\begin{equation*}
  \prt{\alpha}((x_{1}, \ldots, x_{n}),
  (a_{1}, \ldots, a_{n})) =
  \lawbigsqcap_{|x_{1} - y_{1}| \lawsqsupseteq a_{1}}
  \cdots
  \lawbigsqcap_{|x_{n} - y_{n}| \lawsqsupseteq a_{n}}
  |\alpha(x_{1}, \ldots, x_{n})
  - \alpha(y_{1}, \ldots, y_{n})|.
\end{equation*}
The map $\prt{\alpha}$ calculates the self-distance of $\alpha$ (see Proposition~\ref{lem:arrow}). In the sequel, for brevity, given a closed term $\vdash \mathsf{t} : \mathsf{A}$, we identify $\fsem{\mathsf{t}} \in (\{\ast\}, \{\mathbf{1}\}) \rightarrowtriangle \mathcal{Q}_{\mathsf{A}}$ with $\fsemp{\mathsf{t}}{\ast}{\mathbf{1}} \in \mathcal{Q}_{\mathsf{A}}$.

We now state the fundamental lemma.
\begin{thm}[Fundamental Lemma]\label{thm:fundamental-lemma}
  For every $\mathsf{\Gamma} \vdash \mathsf{t} : \mathsf{A}$ and $(\gamma,\xi,\varepsilon) \in \rho_{\mathsf{\Gamma}}$, we have
  \begin{equation*}
    (\sem{\mathsf{t}}_{\gamma},
    \fsem{\mathsf{t}}_{\gamma,\xi},
    \sem{\mathsf{t}}_{\varepsilon})
    \in \rho_{\mathsf{A}}.
  \end{equation*}
  In particular, when $\mathsf{A}$ is equal to $\real$, we have
  \begin{equation*}
    |\sem{\mathsf{t}}_{\gamma} -
    \sem{\mathsf{t}}_{\varepsilon}|
    \lawsqsupseteq \fsem{\mathsf{t}}_{\gamma, \xi}.
  \end{equation*}
\end{thm}
\begin{proof}
  We can prove the claim by standard induction on the derivation of $\mathsf{\Gamma} \vdash \mathsf{t} : \mathsf{A}$. We omit the proof, as we will give a proof for a more general result, namely Theorem~\ref{thm:fl-strong}.
\end{proof}
\begin{figure}
  \centering
  \fbox{\begin{math}
      \begin{array}{c}
        \begin{aligned}
          \fsemp{\mathsf{x}_{i}}{\gamma}{\xi}
          &= \xi|_{i}
          \\
          \fsemp{\underline{r}}{\gamma}{\xi}
          &= 0
          \\
          \fsemp{\alpha
          (\mathsf{t}_{1},\ldots,\mathsf{t}_{n})}{\gamma}{\xi}
          &= \prt{\alpha}(
            (\semp{\mathsf{t}_{1}}{\gamma},
            \ldots,
            \semp{\mathsf{t}_{n}}{\gamma}),
            (\fsemp{\mathsf{t}_{1}}{\gamma}{\xi},
            \ldots,
            \fsemp{\mathsf{t}_{n}}{\gamma}{\xi}))
          \\
          \fsemp{\mathsf{t} \, \mathsf{s}}{\gamma}{\xi}
          &= \fsemp{\mathsf{t}} {\gamma}{\xi}
            (\semp{\mathsf{s}}{\gamma},
            \fsemp{\mathsf{s}}{\gamma}{\xi})
          \\
          \fsemp{\lam{\mathsf{x}}{\mathsf{A}}{\mathsf{t}}}
          {\gamma}{\xi} (x, a)
          &= \fsemp{\mathsf{t}}{\gamma \add x}{\xi \add a}
          \\
          \fsemp{\pair{\mathsf{t}}{\mathsf{s}}}{\gamma}{\xi}
          &= (\fsemp{\mathsf{t}} {\gamma}{\xi},
            \fsemp{\mathsf{s}} {\gamma}{\xi})
          \\
          \fsemp{\fst(\mathsf{t})}{\gamma}{\xi}
          &=
            \text{the first component of }
            \fsemp{\mathsf{t}}{\gamma}{\xi}
          \\
          \fsemp{\snd(\mathsf{t})}{\gamma}{\xi}
          &=
            \text{the second component of }
            \fsemp{\mathsf{t}}{\gamma}{\xi}
        \end{aligned}
      \end{array}
    \end{math}}
  \caption{The Definition of Derivatives}
  \label{fig:diff}
\end{figure}

The fundamental lemma provides a compositional framework for reasoning about distances between the set theoretic interpretations of terms. In the following example, we suppose that $\mathbb{F}$ is the set of all multi-arity functions on $\mathbb{R}$.
\begin{exa}\label{eg:phipsi}
  We fix a real number $\epsilon > 0$, and define a function $F \colon (\mathbb{R} \Rightarrow \mathbb{R}) \times \mathbb{R} \to \mathbb{R}$ by
  \begin{equation*}
    F(f, x) =
    \frac{
      f(x + \epsilon) - f(x)
    }{
      \epsilon
    }.
  \end{equation*}
  The function $F$ approximates the derivative of $f \colon \mathbb{R} \to \mathbb{R}$ at $x \in \mathbb{R}$ and is equal to the set theoretic interpretation of
  \begin{equation*}
    \mathsf{f} : \real \Rightarrow \real,
    \mathsf{x} : \real \vdash
    \mathsf{let} \ \mathsf{y} \
    \mathsf{be} \
    \mathrm{add} (\mathsf{x}) \
    \mathsf{in} \
    \mathrm{diff}
    (\mathsf{f} \, \mathsf{y}
    ,
    \mathsf{f} \, \mathsf{x})
    : \real,
  \end{equation*}
  where $\mathsf{let} \ \mathsf{y} \ \mathsf{be} \ \mathsf{t} \ \mathsf{in} \ \mathsf{s}$ is a syntactic sugar for $(\lam{\mathsf{y}}{\mathsf{A}}{\mathsf{s}}) \, \mathsf{t}$, and $\mathsf{A}$ is the type of $\mathsf{t}$. The functions $\mathrm{diff} \colon \mathbb{R} \times \mathbb{R} \to \mathbb{R}$ and $\mathrm{add} \colon \mathbb{R} \to \mathbb{R}$ are given by $\mathrm{diff}(x,y) = (x - y)/\epsilon$ and $\mathrm{add}(x) = x + \epsilon$. For all $x \in \mathbb{R}$ and $a \in [0,+\infty]$, we have
  \begin{equation*}
    \prt{\mathrm{add}}(x,a)
    =
      \lawbigsqcap_{|x - y| \lawsqsupseteq a} |(x + \epsilon) - (y + \epsilon)|
    =
      \lawbigsqcap_{|x - y| \lawsqsupseteq a} |x - y|
    = a,
  \end{equation*}
  and for all $x, y \in \mathbb{R}$ and $a, b \in [0,+\infty]$, we have
  \begin{align*}
    \prt{\mathrm{diff}}((x, y), (a, b))
    =
      \lawbigsqcap_{|x - z| \lawsqsupseteq a}
      \lawbigsqcap_{|y - w| \lawsqsupseteq b}
      \left|\dfrac{(x - y) - (z - w)}{\epsilon}
      \right|
    =\dfrac{a + b}{\epsilon}.
  \end{align*}
  Now, it follows from the fundamental lemma that
  \begin{equation*}
    D =
      \fsem{
      \mathsf{let} \ \mathsf{y} \
      \mathsf{be} \
      \mathrm{add} (\mathsf{x}) \
      \mathsf{in} \
      \mathrm{diff}
      (\mathsf{f} \, \mathsf{y}
      ,
      \mathsf{f} \, \mathsf{x})}
    \colon \uset{(\real \Rightarrow \real) \times \real} \times \mathcal{Q}_{(\real \Rightarrow \real) \times \real} \to \mathcal{Q}_{\real}
  \end{equation*}
  satisfies
  \begin{equation*}
    (F(f,x), D((f,x),(d,a)), F(g,y)) \in \rho_{\real}
  \end{equation*}
  for all $(f, d, g) \in \rho_{\real \Rightarrow \real}$ and $(x, a, y) \in \rho_{\real}$. To be concrete, $D$ is given by
  \begin{align*}
    D((f, x), (d, a))
    &=
      \fsemp{
      \mathsf{let} \ \mathsf{y} \
      \mathsf{be} \
      \mathrm{add} (\mathsf{x}) \
      \mathsf{in} \
      \mathrm{diff}
      (\mathsf{f} \, \mathsf{y}
      ,
      \mathsf{f} \, \mathsf{x})}{(f,x)}{(d,a)}
    \\
    &=
      \fsemp{
      \lam{\mathsf{y}}{\real}{
      \mathrm{diff}
      (\mathsf{f} \, \mathsf{y}
      , \mathsf{f} \, \mathsf{x})
      }}{(f,x)}{(d,a)}
      (\semp{\mathrm{add}(\mathsf{x})}{(f,x)},
      \fsemp{\mathrm{add}(\mathsf{x})}{(f,x)}{(d,a)}
      )
    \\
    &=
      \fsemp{
      \lam{\mathsf{y}}{\real}{
      \mathrm{diff}
      (\mathsf{f} \, \mathsf{y}
      , \mathsf{f} \, \mathsf{x})
      }}{(f,x)}{(d,a)}
      (x + \epsilon, a)
    \\
    &=
      \fsemp{
      \mathrm{diff}
      (\mathsf{f} \, \mathsf{y}
      , \mathsf{f} \, \mathsf{x})
      }{(f,x,x+\epsilon)}{(d,a,a)}
    \\
    &=
      \prt{\mathrm{diff}}(
      (\semp{\mathsf{f} \, \mathsf{y}}{(f,x,x+\epsilon)},
      \semp{\mathsf{f} \, \mathsf{x}}{(f,x,x+\epsilon)}),
      (\fsemp{\mathsf{f} \, \mathsf{y}}{(f,x,x+\epsilon)}{(d,a,a)},
      \fsemp{\mathsf{f} \, \mathsf{x}}{(f,x,x+\epsilon)}{(d,a,a)})
      )
    \\
    &=
      \prt{\mathrm{diff}}(
      (f(x + \epsilon),f(x)),
      (d(x+\epsilon,a),d(x,a))
      )
    \\
    &=
      \frac{d(x + \epsilon, a)
      + d(x, a)}{\epsilon}.
  \end{align*}
 We note that $D$ depends on the specific term chosen to denote $F$. Since we have
  \begin{equation*}
    (\mathrm{id}_{\mathbb{R}}, \hat{\rho}_{R \Rightarrow R}(\mathrm{id}_{\mathbb{R}}, \sin),
    \sin)
    \in
    \rho_{\real \Rightarrow \real},
    \qquad
    (0,a,a)
    \in
    \rho_{\real},
  \end{equation*}
  $D$ yields the following bound of the distance between $F(\mathrm{id}_{\mathbb{R}}, 0)$ and $F(\sin, a)$:
  \begin{align}
    |F(\mathrm{id}_{\mathbb{R}}, 0)
    - F(\sin, a)|
    &\lawsqsupseteq
      D((\mathrm{id}_{\mathbb{R}},0),
      (\hat{\rho}_{\real \Rightarrow \real}(\mathrm{id}_{\mathbb{R}},\sin),a))
      \notag
    \\
    &=
      \dfrac{1}{\epsilon}
      \left(
      \lawbigsqcap_{|\epsilon - y| \lawsqsupseteq a}
      \lawsqcapparen{|\epsilon - \sin(y)|}{a}
      + \lawbigsqcap_{|y| \lawsqsupseteq a}
      \lawsqcapparen{|\sin(y)|}{a}
      \right)
      \notag
    \\
    &=
      \dfrac{1}{\epsilon}
      \left(
      \lawbigsqcap_{|\epsilon - y| \lawsqsupseteq a}
      \lawsqcapparen{|\epsilon - \sin(y)|}{a}
      \right)
      + \dfrac{a}{\epsilon},
      \label{eqn:id_sin_bound}
  \end{align}
  where the first equality follows from the explicit presentation of $\hat{\rho}_{\real \Rightarrow \real}(\mathrm{id}_{\mathbb{R}}, \sin)$ given in Example~\ref{eg:id-sin}, and the second equality follows from $|\sin(y)| \lawsqsupseteq |y|$ for all $y \in \mathbb{R}$. Furthermore, let us assume $a = \epsilon^{2}$ and $\epsilon < 1 / 2$. Then, from \eqref{eqn:id_sin_bound}, we obtain
  \begin{align}
    |F(\mathrm{id}_{\mathbb{R}},0) - F(\sin,\epsilon^{2})|
    &\lawsqsupseteq
      \dfrac{1}{\epsilon}
      \left(
      \lawbigsqcap_{|\epsilon - y| \lawsqsupseteq \epsilon^{2}}
      \lawsqcapparen{|\epsilon - \sin(y)|}{\epsilon^{2}}
      \right)
      + \dfrac{\epsilon^{2}}{\epsilon},
      \notag
    \\
    &=
      \lawbigsqcap_{|\epsilon - y| \lawsqsupseteq \epsilon^{2}}
      \lawsqcapparen{\left|1 - \dfrac{\sin(y)}{\epsilon}\right|}{\epsilon}
      + \epsilon
      \notag
    \\
    &\lawsqsupseteq
      \left(
      \lawsqcap{
      \lawsqcap{
      \left|1 - \dfrac{\sin(\epsilon - \epsilon^{2})}{\epsilon}\right|}{
      \left|1 - \dfrac{\sin(\epsilon + \epsilon^{2})}{\epsilon}\right|}}{
      \epsilon
      }
      \right)
      + \epsilon,
      \label{eqn:sin_e_e}
  \end{align}
  where the second inequality follows from the fact that $\sin(y)$ is monotonically increasing for $y \in [\epsilon - \epsilon^{2},\epsilon+\epsilon^{2}] \subseteq [-\pi/2, \pi/2]$. Since
  \begin{align*}
    \lim_{\epsilon \to 0} \dfrac{\sin(\epsilon - \epsilon^{2})}{\epsilon}
    &=
      \lim_{\epsilon \to 0}
      \left(
      \dfrac{\sin(\epsilon - \epsilon^{2}) - \sin(0)}{\epsilon - \epsilon^{2}}\right)\left(
      \dfrac{\epsilon - \epsilon^{2}}{\epsilon}\right)
      = \cos(0) = 1,
    \\
    \lim_{\epsilon \to 0} \dfrac{\sin(\epsilon + \epsilon^{2})}{\epsilon}
    &=
      \lim_{\epsilon \to 0}
      \left(
      \dfrac{\sin(\epsilon + \epsilon^{2}) - \sin(0)}{\epsilon + \epsilon^{2}}\right)\left(
      \dfrac{\epsilon + \epsilon^{2}}{\epsilon}\right)
      = \cos(0) = 1,
  \end{align*}
  the right hand side of \eqref{eqn:sin_e_e} converges to $0$ as $\epsilon \to 0$. Therefore, $|F(\mathrm{id}_{\mathbb{R}},0) - F(\sin,\epsilon^{2})|$ also converges to $0$ as $\epsilon \to 0$.
\end{exa}

We can also leverage the fundamental lemma to bound the distances between different functions.
\begin{prop}\label{prop:fundamental2}
  For all closed terms $\vdash \mathsf{t} : \mathsf{A} \Rightarrow \real$ and $\vdash \mathsf{s} : \mathsf{A} \Rightarrow \real$, we have
  \begin{equation*}
    (\sem{\mathsf{t}},
    d,
    \sem{\mathsf{s}})
    \in \rho_{\mathsf{A} \Rightarrow \real},
  \end{equation*}
  where $d \in \mathcal{Q}_{\mathsf{A} \Rightarrow \real}$ is given by
  \begin{equation*}
    d(x, a) =
    \lawsqcap{
      \bigl(|\sem{\mathsf{t}}(x) - \sem{\mathsf{s}}(x)|
      +
      \fsem{\mathsf{s}}(x, a)\bigr)
    }{\fsem{\mathsf{t}}(x, a)} .
  \end{equation*}
\end{prop}
\begin{proof}
  By the fundamental lemma, we have
  \begin{equation*}
    (\sem{\mathsf{t}},
    \fsem{\mathsf{t}},
    \sem{\mathsf{t}})
    \in \rho_{\mathsf{A} \Rightarrow \real},
    \qquad
    (\sem{\mathsf{s}},
    \fsem{\mathsf{s}},
    \sem{\mathsf{s}})
    \in \rho_{\mathsf{A} \Rightarrow \real}.
  \end{equation*}
  Therefore, for every $(x, a, y) \in \rho_{\mathsf{A}}$, the following inequalities hold:
  \begin{align*}
    d(x, a)
    &\lawsqsubseteq
      |\sem{\mathsf{t}}(x) - \sem{\mathsf{s}}(x)|
      + \fsem{\mathsf{s}}(x, a)
    \\
    &\lawsqsubseteq
      |\sem{\mathsf{t}}(x) - \sem{\mathsf{s}}(x)|
      + |\sem{\mathsf{s}}(x) - \sem{\mathsf{s}}(y)|
    \\
    &\lawsqsubseteq
      |\sem{\mathsf{t}}(x) - \sem{\mathsf{s}}(y)|,
  \end{align*}
  and
  \begin{align*}
    d(x, a)
    &\lawsqsubseteq
      \fsem{\mathsf{t}}(x, a)
      \lawsqsubseteq
      |\sem{\mathsf{t}}(x) - \sem{\mathsf{t}}(y)|.
  \end{align*}
  Hence, by the definition of $\rho_{\mathsf{A} \Rightarrow \real}$, we obtain $(\sem{\mathsf{t}}, d, \sem{\mathsf{s}}) \in \rho_{\mathsf{A} \Rightarrow \real}$.
\end{proof}

\subsection{Some Additional Properties in Quasi-Quasi-Metric Semantics}
\label{sec:relating-quasi-quasi-metric}

In this section, we investigate the quasi-quasi-metrics arising from the interpretation of $\LL$. We show that each quasi-quasi-metric of the form $\rho_{\mathsf{A}}$ enjoys several additional properties, namely, a variant of indistancy, left strong transitivity, and a weak form of symmetry. Below, we first introduce these three properties, comparing them with the notion of partial (quasi-)metrics \cite{matthews,KUNZI2006}, and then show that $\rho_{\mathsf{A}}$ enjoys them for every type $\mathsf{A}$.

\subsubsection{Self-distance and Self-indistancy}
\label{sec:selfdistance}

For a $Q$-predicate $r \subseteq X \times Q \times X$, the \emph{self-distance} of $x \in X$ is defined to be $\hat{r}(x, x)$; for brevity, we denote $\hat{r}(x, x)$ by $\sigma^{r}_{x}$:
\begin{equation*}
  \sigma^{r}_{x} = \hat{r}(x, x) = \bigsqcup \{a \in Q \mid (x, a, x) \in r\}.
\end{equation*}
In what follows, we write $\sigma_{x}$ for $\sigma^{r}_{x}$ when $r$ is clear from context. We introduce a variant of the indistancy condition.

\begin{defi}
  We say that a $Q$-predicate $r \subseteq X \times Q \times X$ satisfies \emph{self-indistancy} when for all $x, y \in X$, if $(x, \sigma_{x}, y) \in r$, then $x = y$.
  % Similarly, we say that a $Q$-predicate $r \subseteq X \times Q \times X$ satisfies \emph{right self-indistancy} when for all $x, y \in X$, if $(x, \sigma_{y}, y) \in r$, then $x = y$.
\end{defi}

When $r$ is a reflexive $[0,+\infty]$-predicate on a set $X$, self-indistancy reduces to standard indistancy. Indeed, when the self-distance of every $x \in X$ is $0$, self-indistancy means that if the distance between $x$ and $y$ is $0$, then $x$ is equal to $y$. In terms of $Q$-relations, we can characterize self-indistancy as follows:
\begin{prop}
  Let $\phi \colon X \times X \to Q$ be a $Q$-relation. The associated $Q$-predicate $\check{\phi} \subseteq X \times Q \times X$ satisfies the self-indistancy condition if and only if for all $x, y \in X$, if $\phi(x, x) \sqsubseteq \phi(x, y)$, then $x = y$.
\end{prop}
\begin{proof}
  Suppose that $\check{\phi}$ satisfies the self-indistancy condition. For arbitrary $x, y \in X$, if $\phi(x, x) \sqsubseteq \phi(x, y)$, then $(x, \sigma_{x}, y) = (x, \phi(x,x), y) \in \check{\phi}$. Therefore, by self-indistancy of $\check{\phi}$, we obtain $x = y$. Conversely, suppose that for all $x, y \in X$, if $\phi(x, x) \sqsubseteq \phi(x, y)$, then $x = y$. To show self-indistancy, let $(x, \sigma_{x}, y) \in \check{\phi}$, which implies that $\phi(x, x) \sqsubseteq \phi(x, y)$. Then, by the assumption, we obtain $x = y$. Hence, $\check{\phi}$ satisfies the self-indistancy condition.
\end{proof}

We note that a quasi-reflexive $Q$-relation $\phi \colon X \times X \to Q$ satisfies the self-indistancy condition if and only if for all $x, y \in X$, if $\phi(x, x) = \phi(x, y)$, then $x = y$. A similar indistancy condition appears in the definition of a partial (quasi-)metric space \cite{matthews,KUNZI2006}: if $\phi(x, x) = \phi(x, y)$ and $\phi(y, x) = \phi(y, y)$, then $x = y$.

\subsubsection{Left Strong Transitivity}
\label{sec:quasi-metr-part}

We next introduce \emph{left strong transitivity} as a variant of the strong transitivity condition in the definition of partial (quasi-)metric space. Recall that the strong transitivity condition for a $[0, +\infty]$-relation $\phi \colon X \times X \to [0,+\infty]$ is given by
\begin{align*}
  % \label{eq:pqm}
  &\phi(x, y) \lawotimes \phi(y, z)
  \lawsqsubseteq \phi(x, z) \lawotimes \phi(y, y).
  \tag{Strong transitivity in $[0, +\infty]$}
\end{align*}
This strong transitivity, together with quasi-reflexivity $\phi(x,y) \lawsqsubseteq \phi(x,x)$ implies the following inequality (see Proposition~\ref{prop:ST1_LST}):
\begin{equation*}
  (\phi(y, y) \multimap \phi(x, y)) \lawotimes \phi(y, z)
  \lawsqsubseteq \phi(x, z).
\end{equation*}
We call this property \emph{left strong transitivity}.
\begin{prop}\label{prop:ST1_LST}
  For every $[0, +\infty]$-relation $\phi \colon X \times X \to [0,+\infty]$, if
  \begin{equation*}% \label{eq:pqm}
    \phi(x, y) \lawsqsubseteq \phi(x, x),
    \qquad
    \phi(x, y) \lawotimes \phi(y, z)
    \lawsqsubseteq \phi(x, z) \lawotimes \phi(y, y)
  \end{equation*}
  hold for all $x, y, z \in X$, then
  \begin{equation*}
    (\phi(y, y) \multimap \phi(x, y)) \lawotimes \phi(y, z)
    \lawsqsubseteq \phi(x, z)
  \end{equation*}
  holds for all $x, y, z \in X$.
\end{prop}
\begin{proof}
  We consider two cases according to the value of $\phi(y, y)$:
  \begin{itemize}
  \item
    Suppose that $\phi(y, y) = \infty$. The quasi-reflexivity $\phi(y, z) \lawsqsubseteq \phi(y, y)$ yields $\phi(y,z) = \infty$. Thus, we obtain
    \begin{equation*}
      (\phi(y, y) \multimap \phi(x, y)) \lawotimes \phi(y, z) = \infty \lawsqsubseteq \phi(x, z).
    \end{equation*}
  \item
    Suppose that $\phi(y, y) < \infty$. We can derive left strong transitivity as follows:
    \begin{align*}
      (\phi(y, y) \multimap \phi(x, y)) \lawotimes \phi(y, z)
      &\lawsqsubseteq
      \phi(y, y) \multimap (\phi(x, y) \lawotimes \phi(y, z))
      \\
      &\lawsqsubseteq
      \phi(y, y) \multimap (\phi(x, z) \lawotimes \phi(y, y))
      \\
      &= \phi(x, z),
    \end{align*}
    where the second inequality follows from strong transitivity, and the last line follows from the definition of the residual. Indeed, for all $a, b \in [0,+\infty]$, if $b$ is finite, then $a$ is equal to $b \multimap (a \lawotimes b)$.
  \end{itemize}
  This completes the proof.
\end{proof}

We generalize left strong transitivity to arbitrary quantale-valued relations as follows:
\begin{defi}
  We say that a $Q$-relation $\phi \colon X \times X \to Q$ is \emph{left strong transitive} if for all $x, y, z \in X$, we have
  \begin{equation*}
    (\phi(y, y) \multimap \phi(x, y)) \otimes \phi(y, z)
    \sqsubseteq \phi(x, z).
  \end{equation*}
  A $Q$-relation $\phi \colon X \times X \to Q$ is a \emph{left strong quasi-quasi-metric} on $X$ if $\phi$ is quasi-reflexive and left strong transitive.
\end{defi}

It follows from integrality of $Q$ that if a $Q$-relation $\phi \colon X \times X \to Q$ is left strong transitive, then $\phi$ is transitive. Hence, left strong quasi-quasi-metrics are quasi-quasi-metrics. Let us give an example of a left strong quasi-quasi-metric.
\begin{exa}
  Let $Q$ be a quantale, and let $X$ be a set. For every function $\omega \colon X \to Q$ and every quasi-metric $\phi \colon X \times X \to Q$ (see Definition~\ref{def:quasi}), a $Q$-relation $\psi \colon X \times X \to Q$ given by
  \begin{equation*}
    \psi(x, y) = \omega x \otimes \phi(x, y)
  \end{equation*}
  is a left strong quasi-quasi-metric on $X$. A similar construction for partial metrics can be found in \cite{matthews}. Let us check that $\psi$ satisfies quasi-reflexivity and left strong transitivity. For all $x, y \in X$, by integrality of $Q$ and reflexivity of $\phi$, we have
  \begin{align*}
    \psi(x, y)
    \sqsubseteq \omega x = \omega x \otimes \phi(x, x) = \psi(x,x).
  \end{align*}
  Hence, $\psi$ is quasi-reflexive. We next deduce left strong transitivity of $\psi$. For all $x, y, z \in X$,
  \begin{align*}
    (\psi(y, y) \multimap \psi(x,y)) \otimes \psi(y, z)
    &= ((\omega y \otimes \phi(y,y)) \multimap \psi(x, y)) \otimes \psi(y, z)
    \\
    &\sqsubseteq (\omega y \multimap \psi(x, y)) \otimes \psi(y, z)
    \\
    &= (\omega y \multimap \psi(x, y)) \otimes \omega y \otimes \phi(y, z)
    \\
    &\sqsubseteq \psi(x, y) \otimes \phi(y, z)
    \\
    &= \omega x \otimes \phi(x, y) \otimes \phi(y, z)
    \\
    &\sqsubseteq \omega x \otimes \phi(x, z)
    = \psi(x, z).
  \end{align*}
  Here, the first inequality follows from reflexivity of $\phi$, and the last inequality follows from transitivity of $\phi$.
\end{exa}

We can characterize left strong transitivity in terms of $Q$-predicates and self-distances. %: here, for a $Q$-predicate $r \subseteq X \times Q \times X$, the \emph{self-distance} of $x \in X$ is defined to be $\hat{r}(x, x)$. For brevity, we denote $\hat{r}(x, x)$ by $\sigma^{r}_{x}$.
It remains open whether we can characterize left strong transitivity purely in terms of $Q$-predicates without relying on self-distances.

\begin{prop}\label{prop:lstrong}
  For every closed $Q$-predicate $r \subseteq X \times Q \times X$, the associated $Q$-relation $\hat{r}$ is left strong transitive if and only if
  \begin{equation*}
    \forall x, y, z \in X,\, \forall a,b \in Q,\, \text{if }
    (x, a \otimes \sigma_{y}, y) \in r \text{ and }
    (y, b, z) \in r, \text{ then } (x, a \otimes b, z) \in r.
  \end{equation*}
\end{prop}
% \begin{prop}\label{prop:lpqm}
%   For every closed $Q$-predicate $r \subseteq X \times Q \times X$, the associated $Q$-relation $\hat{r}$ is a left strong quasi-quasi-metric if and only if
%   \begin{itemize}
%   \item
%     $r$ is quasi-reflexive; and
%   \item
%     for all $x, y \in X$ and $a,b \in Q$, if $(x, a \otimes \hat{r}(y, y), y) \in r$ and $(y, b, z) \in r$, then $(x, a \otimes b, z) \in r$.
%   \end{itemize}
%   We also refer to the second property as \emph{left strong transitivity}, and we call $r$ a \emph{left strong quasi-quasi-metric} whenever $r$ is quasi-reflexive and left strong transitive.
% \end{prop}
\begin{proof}
  % Since $\hat{r}$ is quasi-reflexive if and only if $r$ is quasi-reflexive, we only need to check that $\hat{r}$ is left strong transitive if and only if $r$ is left strong transitive.
  We first suppose that $\hat{r}$ is left strong transitive. If $(x, a \otimes \sigma_{y}, y) \in r$ and $(y, b, z) \in r$, then by the definition of $\hat{r}$, we have $a \otimes \sigma_{y} \sqsubseteq \hat{r}(x, y)$ and $b \sqsubseteq \hat{r}(y, z)$. Hence, the left strong transitivity of $\hat{r}$ yields
  \begin{equation*}
    a \otimes b \sqsubseteq
    (\sigma_{y} \multimap \hat{r}(x, y)) \otimes \hat{r}(y, z)
    \sqsubseteq \hat{r}(x, z).
  \end{equation*}
  By Proposition~\ref{prop:rhd}, since $r$ is closed, we obtain $(x, a \otimes b, z) \in r$. Conversely, we suppose that $r$ is left strong transitive. By the definition of $\hat{r}$, it holds that
  \begin{equation}
    (x, \hat{r}(x, y), y) \in r,
    \qquad
    (y, \hat{r}(y, z), z) \in r.
    \label{eq:prop-lst}
  \end{equation}
  Since $(\sigma_{y} \multimap \hat{r}(x, y)) \otimes \sigma_{y} \sqsubseteq \hat{r}(x, y)$ and $r$ is closed, the first clause of \eqref{eq:prop-lst} implies
  \begin{equation*}
    (x, (\sigma_{y} \multimap \hat{r}(x, y)) \otimes \sigma_{y}, y)
    \in r.
  \end{equation*}
  Hence, the left strong transitivity of $r$ yields
  \begin{equation*}
    (x, (\sigma_{y} \multimap \hat{r}(x, y)) \otimes \hat{r}(y, z), z)
    \in r.
  \end{equation*}
  Therefore, $(\sigma_{y} \multimap \hat{r}(x, y)) \otimes \hat{r}(y, z) \sqsubseteq \hat{r}(x, z)$.
\end{proof}

We also refer to this property as \emph{left strong transitivity}, and we call $r$ a \emph{left strong quasi-quasi-metric} whenever $r$ is quasi-reflexive and left strong transitive. In Table~\ref{table:summary2}, we summarize the formulations of self-indistancy and left strong transitivity in terms of $Q$-relation, $Q$-predicate and binary relations.

\begin{table}
  \centering
  \begin{tabular}{l}
    \begin{tabular}{|c|rcl|}
      \cline{2-4}
      \multicolumn{1}{c|}{}
      & \multicolumn{3}{c|}{Self-indistancy}
      \\ \hline
      $Q$-relation
      & $\phi(x, x) \sqsubseteq \phi(x, y)$ & $\implies$ & $x = y$
      \\ \hline
      $Q$-predicate
      & $(x, \sigma_{x}, y) \in r$ & $\implies$ & $x = y$
      \\ \hline
      Binary relation
      & $((x, x) \in R \Rightarrow (x, y) \in R)$
      & $\implies$ & $x = y$
      \\ \hline
    \end{tabular}
    \\\\
    \begin{tabular}{|c|rclcl|}
      \cline{2-6}
      \multicolumn{1}{c|}{}
      & \multicolumn{5}{c|}{Left strong transitivity}
      \\ \hline
      $Q$-relation
      & $(\phi(y, y) \multimap \phi(x, y))$ & {\!\!$\otimes$\!\!} & $\phi(y ,z)$ & $\sqsubseteq$ & $\phi(x, z)$
      \\ \hline
      $Q$-predicate
      & $(x, a \otimes \sigma_{y}, y) \in r$ & {\!\!$\&$\!\!} & $ (y, b, z) \in r$ & $\implies$ & $(x, a \otimes b, z) \in r$
      \\ \hline
      Binary relation
      & $((y, y) \in R \Rightarrow (x, y) \in R)$ & {\!\!$\&$\!\!} & $(y, z) \in R$ & $\implies$ & $(x, z) \in R$
      \\ \hline
    \end{tabular}
  \end{tabular}
  \caption{Summary of Properties for $Q$-relations, $Q$-predicates and Binary Relations}
  \label{table:summary2}
\end{table}

\subsubsection{Remarks on Transitivity Conditions}
\label{sec:remark-trans-cond}

Before we introduce a weak form of symmetry in Section~\ref{sec:wsym} and give our main observation in Section~\ref{sec:mainstrong}, we make a few remarks on transitivity conditions. While we focus on left strong transitivity in this paper, there are several ways to adapt the notion of strong transitivity to the quantale-valued setting, such as
\begin{align*}
  \phi(x, y) \otimes \phi(y, z)
  &\sqsubseteq
    \phi(x, z) \otimes \phi(y, y),
    \tag{ST1}
    \label{eqn:ST1}
  \\
  \phi(x, y) \otimes (\phi(y, y) \multimap \phi(y, z))
  &\sqsubseteq
    \phi(x, z),
    \tag{ST2}
    \label{eqn:ST2}
  \\
  \phi(y, y) \multimap (\phi(x, y) \otimes \phi(y, z))
  &\sqsubseteq
    \phi(x, z),
    \tag{ST3}
    \label{eqn:ST3}
  \\
  (\phi(y, y) \multimap \phi(x, y)) \otimes \phi(y, y) \otimes (\phi(y, y) \multimap \phi(y, z))
  &\sqsubseteq
    \phi(x, z),
    \tag{ST4}
    \label{eqn:ST4}
  \\
  (\phi(y, y) \multimap \phi(x, y)) \otimes \phi(y, z)
  &\sqsubseteq
    \phi(x, z).
    \tag{LST}
    \label{eqn:LST}
\end{align*}
The condition \eqref{eqn:ST1}, a straightforward adaptation of the strong transitivity condition, appears in \cite{Geoffroy2020}. The conditions \eqref{eqn:ST2} and \eqref{eqn:ST4} are related to the notion of \emph{diagonal} \cite{Stubbe2018}. The condition \eqref{eqn:ST2} can be found in \cite{Kopperman_Matthews_Pajoohesh_2004,PistoneLICS}, and \eqref{eqn:ST4} can be found in \cite{Stubbe2018}. The last condition \eqref{eqn:LST} is left strong transitivity. For $[0,+\infty]$-relations, these conditions are related as follows:
\begin{prop}
  For every left quasi-reflexive $[0,+\infty]$-relation $\phi \colon X \times X \to [0,+\infty]$, we have the following implications:
  \begin{equation*}
    \eqref{eqn:ST3} \implies
    \eqref{eqn:ST2} \implies
    \eqref{eqn:ST1} \implies
    \eqref{eqn:ST4} \iff
    \eqref{eqn:LST}.
  \end{equation*}
  Moreover, if $\phi$ is both left and right quasi-reflexive, i.e., we have $\phi(x, y) \sqsubseteq \phi(x, x)$ and $\phi(x, y) \sqsubseteq \phi(y, y)$ for all $x,y \in X$, then \eqref{eqn:ST1}, \eqref{eqn:ST2}, \eqref{eqn:ST4}, and \eqref{eqn:LST} are equivalent. If we further assume that $\phi(x, y)$ is finite for all $x, y \in X$, then the above five conditions are equivalent.
\end{prop}
\begin{proof}
  We can deduce \eqref{eqn:ST2} from \eqref{eqn:ST3} as follows:
  \begin{align*}
    \phi(x, y) \lawotimes (\phi (y, y) \multimap \phi(y, z))
    &\lawsqsubseteq
    \phi(y, y) \multimap (\phi(x, y) \lawotimes \phi(y, z))
    \\
    &\lawsqsubseteq
    \phi(x, z).
  \end{align*}
  The first inequality follows from that the function $\phi(y,y) \multimap (-)$ is right adjoint to $\phi(y,y) \otimes (-)$, and the second inequality is \eqref{eqn:ST3}. We next deduce \eqref{eqn:ST1} from \eqref{eqn:ST2}. For all $a, b \in [0,+\infty]$, if $a \lawsqsubseteq b$, then by the definition of the residual $b \multimap a$, we have
  \begin{equation*}
    a = (b \multimap a) \lawotimes b.
  \end{equation*}
  This property is called divisibility. Since $\phi(y, z) \lawsqsubseteq \phi(y, y)$, applying divisibility gives
  \begin{align*}
    \phi(x, y) \lawotimes \phi(y, z)
    &=
    \phi(x, y) \lawotimes (\phi(y, y) \multimap \phi(y, z)) \lawotimes \phi(y, y)
    \\
    &\lawsqsubseteq
    \phi(x, z) \lawotimes \phi(y, y).
  \end{align*}
 %  We then deduce \eqref{eqn:ST4} from \eqref{eqn:ST1}. If $\phi(y, y) = \infty$, then \eqref{eqn:ST4} holds since the left hand side of \eqref{eqn:ST4} is $\infty$. Therefore, we suppose $\phi(y, y)$ is finite. Then we can derive \eqref{eqn:ST4} from \eqref{eqn:ST1} as follows:
 % \begin{align*}
 %    (\phi(y, y) \multimap \phi(x, y)) \otimes \phi(y, y) \otimes (\phi(y, y) \multimap \phi(y, z))
 %    &\sqsubseteq
 %    (\phi(y, y) \multimap \phi(x, y)) \otimes \phi(y, z)
 %    \\
 %    &\sqsubseteq
 %    \phi(y, y) \multimap (\phi(x, y) \otimes \phi(y, z))
 %    \\
 %    &\sqsubseteq
 %    \phi(y, y) \multimap (\phi(x, z) \otimes \phi(y, y))
 %    \\
 %    &= \phi(x, z)
 %  \end{align*}
 %  where the last line follows from the definition of the residual. Indeed, for any $a, b \in [0,+\infty]$, if $b$ is finite, then we have
 %  \begin{equation*}
 %    a = b \multimap (a \otimes b).
 %  \end{equation*}
  We have already checked that \eqref{eqn:ST1} implies \eqref{eqn:LST} in Proposition~\ref{prop:ST1_LST}. The equivalence between \eqref{eqn:ST4} and \eqref{eqn:LST} follows from divisibility. Indeed, since $\phi(y, z) \lawsqsubseteq \phi(y,y)$, we have
  \begin{equation*}
    \phi(y, y) \lawotimes (\phi(y, y) \multimap \phi(y, z)) = \phi(y, z),
  \end{equation*}
  and therefore, the left hand side of \eqref{eqn:ST4} is equal to the left hand side of \eqref{eqn:LST}. If we further assume that $\phi$ is right quasi-reflexive, i.e., $\phi(x, y) \lawsqsubseteq \phi(y, y)$, then it follows from divisibility that we have
  \begin{equation*}
    \phi(x, y) = (\phi(y, y) \multimap \phi(x, y)) \lawotimes \phi(y, y).
  \end{equation*}
  This implies that the left hand side of \eqref{eqn:ST4} and the left hand side of \eqref{eqn:ST2} are the same. Hence, \eqref{eqn:ST4} is equivalent to \eqref{eqn:ST2}. Since we have already checked that \eqref{eqn:ST2} is the strongest condition and \eqref{eqn:ST4} is the weakest condition among \eqref{eqn:ST2}, \eqref{eqn:ST1}, \eqref{eqn:ST4} and \eqref{eqn:LST}, we see that \eqref{eqn:ST1}, \eqref{eqn:ST2}, \eqref{eqn:ST4}, and \eqref{eqn:LST} are equivalent. Finally, we suppose that $\phi(x, y)$ is finite for all $x, y \in X$, and we show that \eqref{eqn:LST} implies \eqref{eqn:ST3}. It follows from the definition of residual that for all $a,b \in [0,+\infty]$ and $c \in [0,+\infty)$, if $a \lawsqsubseteq c$, then
  \begin{equation*}
    c \multimap (a \lawotimes b) = (c \multimap a) \lawotimes b.
  \end{equation*}
  By this equality and the right quasi-reflexivity of $\phi$, we obtain
  \begin{equation*}
    \phi(y, y) \multimap (\phi(x, y) \lawotimes \phi(y, z)) = (\phi(y,y) \multimap \phi(x, y)) \lawotimes \phi(y,z),
  \end{equation*}
  which implies that the left hand side of \eqref{eqn:ST3} is equal to the left hand side of \eqref{eqn:LST}. Hence, \eqref{eqn:ST3} is equivalent to \eqref{eqn:LST}. Since \eqref{eqn:ST3} is the strongest condition and \eqref{eqn:LST} is the weakest condition among the five conditions, we see that the five conditions are equivalent.
\end{proof}

In general, for every quantale-valued relation $\phi \colon X \times X \to Q$, if $\phi$ is left strong transitive, then $\phi$ satisfies \eqref{eqn:ST4}. Indeed, for all $x, y, z \in X$,
\begin{align*}
  (\phi(y,y) \multimap \phi(x,y)) \otimes \phi(y,y)
  \otimes (\phi(y,y) \multimap \phi(y,z))
  &\sqsubseteq
    (\phi(y,y) \multimap \phi(x,y)) \otimes \phi(y,z)
  \\
  &\sqsubseteq
    \phi(x,z),
\end{align*}
where the first inequality holds because $a \otimes (a \multimap b) \sqsubseteq b$, and the second inequality is \eqref{eqn:LST}. In Section~\ref{sec:mainstrong}, we will observe that all quasi-quasi-metrics arising from the interpretation of simple types satisfy the left strong transitivity condition. Consequently, these quasi-quasi-metrics also satisfy \eqref{eqn:ST4}. On the other hand, the other conditions \eqref{eqn:ST1}, \eqref{eqn:ST2}, and \eqref{eqn:ST3} fail in the quasi-quasi-metric space $\psem{\real \Rightarrow \real}$ interpreting the function type over the reals, as the next example illustrates.
\begin{exa}\label{eg:st}
  For brevity, we denote $\hat{\rho}_{\real \Rightarrow \real}$ by $\phi$, and we write $\phi_{f,g}$ for $\phi(f, g)$. We show that the $\mathcal{Q}_{\real \Rightarrow \real}$-relation $\phi$ satisfies none of \eqref{eqn:ST1}, \eqref{eqn:ST2}, or \eqref{eqn:ST3}. To this end, we define functions $f,g,h \colon \mathbb{R} \to \mathbb{R}$ by
  \begin{equation*}
    f(x) = 1, \qquad
    g(x) = |x|, \qquad
    h(x) = x.
  \end{equation*}
  We first check that \eqref{eqn:ST1} for $\phi$ does not hold:
  \begin{equation*}
    \bigl(\phi_{f, g}
    \otimes
    \phi_{g, h}\bigr)(0,2)
    = 1 \lawotimes 2
    \lawsqsupsetneq
    3 \lawotimes 2
    =
    \bigl(\phi_{f, h}
    \otimes
    \phi_{g, g}\bigr)(0,2).
  \end{equation*}
  We next observe a failure of \eqref{eqn:ST2} for $\phi$:
  \begin{align*}
    \bigl(\phi_{f, g} \otimes
    (\phi_{g, g}
    \multimap
    \phi_{g, h})\bigr)(0,2)
    &= 1 \lawotimes \lawbigsqcap_{a \lawsqsupseteq 2}
      \bigl(
      \phi_{g, g}(0, a)
      \multimap
      \phi_{g, h}(0, a)
      \bigr)
    \\
    &= 1 \lawotimes \lawbigsqcap_{a \lawsqsupseteq 2}
      \bigl(
      a
      \multimap
      a
      \bigr)
    \\
    &= 1 \lawotimes 0
    = 1 \lawsqsupsetneq 3 = \phi_{f, h}(0,2).
  \end{align*}
  Finally, we observe a failure of \eqref{eqn:ST3} for $\phi$:
  \begin{align*}
    \bigl(\phi_{g, g} \multimap
    (\phi_{f, g}
    \otimes
    \phi_{g, h})\bigr)(0,2)
    &= \lawbigsqcap_{a \lawsqsupseteq 2}
      \bigl(
      \phi_{g, g}(0, a)
      \multimap
      (\phi_{f, g}(0, a)
      + \phi_{g, h}(0, a))
      \bigr)
    \\
    &= \lawbigsqcap_{a \lawsqsupseteq 2}
      \bigl(
      a
      \multimap
      (1
      \lawotimes a)
      \bigr)
    \\
    &= 1 \lawsqsupsetneq 3 = \phi_{f, h}(0,2).
  \end{align*}
\end{exa}

Of the various transitivity conditions presented so far, we can characterize transitivity, \eqref{eqn:ST2}, and \eqref{eqn:LST} in terms of \emph{observational quasi-metrics} $\phi^{r}$ and $\phi^{l}$ given in Example~\ref{eg:obsmet}.

\begin{prop}\label{prop:left-right}
  For every $Q$-relation $\phi \colon X \times X \to Q$, the following statements hold:
  \begin{enumerate}[label=(\roman*)]
  \item
    $\phi$ is transitive iff for all $x, y \in X$, we have $\phi(x, y) \sqsubseteq \phi^{r}(x, y)$;
  \item
    $\phi$ is transitive iff for all $x, y \in X$, we have $\phi(x, y) \sqsubseteq \phi^{l}(x, y)$;
  \item
    $\phi$ satisfies \eqref{eqn:ST2} iff for all $x, y \in X$, we have $\phi(x, x) \multimap \phi(x, y) \sqsubseteq \phi^{r}(x, y)$;
  \item
    $\phi$ satisfies \eqref{eqn:LST} iff for all $x, y \in X$, we have $\phi(y, y) \multimap \phi(x, y) \sqsubseteq \phi^{l}(x, y)$.
  \end{enumerate}
\end{prop}
\begin{proof}
  By the definition of $\phi^{l}$ and $\phi^{r}$, the following inequalities hold for all $x,y,z \in X$:
  \begin{align*}
    \phi^{l}(x, y) \otimes \phi(y, z)
    & \sqsubseteq \phi(x, z),
    % \tag{left transitivity}
    \\
    \phi(x, y) \otimes \phi^{r}(y, z)
    & \sqsubseteq \phi(x, z).
    % \tag{right transitivity}
  \end{align*}
  Moreover, $\phi^{l}$ and $\phi^{r}$ are the greatest $Q$-relations that satisfy the above conditions. The claims (i)--(iv) follow from these characterizations. Here, we prove (i). We can check the other statements in the same manner. If $\phi$ is transitive, then since $\phi(x, y) \otimes \phi(y, z)\sqsubseteq \phi(x, z)$ holds for all $x, y, z \in X$, by the above characterization of $\phi^{r}$, we have $\phi(x, y) \sqsubseteq \phi^{r}(x, y)$. Conversely, if for all $x, y \in X$, we have $\phi(x, y) \sqsubseteq \phi^{r}(x, y)$, then for all $x, y, z \in X$, we have
  \begin{equation*}
    \phi(x, y) \otimes \phi(y, z) \sqsubseteq
    \phi(x, y) \otimes \phi^{r}(y, z) \sqsubseteq \phi(x, z).
  \end{equation*}
  Hence, $\phi$ is transitive.
\end{proof}

The result above suggests that the left strong quasi-quasi-metrics can be seen as limit cases of the quasi-quasi-metrics, namely, those for which the quasi-metric $\phi^{l}(x, y)$ can be written under the simpler form $\phi(y, y) \multimap \phi(x, y)$. Indeed, if $\phi$ is left strong transitive, then $\phi(y, y) \multimap \phi(x, y) \sqsubseteq \phi^{l}(x, y)$ follows from Proposition~\ref{prop:left-right}, and $\phi^{l}(x, y) \sqsubseteq \phi(y, y) \multimap \phi(x, y)$ follows from the definition of $\phi^{l}$. Hence, a quasi-quasi-metric $\phi$ satisfies the left strong transitivity condition if and only if $\phi^{l}(x, y)$ is equal to $\phi(y, y) \multimap \phi(x, y)$.

\subsubsection{Weak Symmetry}
\label{sec:wsym}

We introduce a weak form of symmetry.
\begin{defi}
  Let $r \subseteq X \times Q \times X$ be a $Q$-predicate. We say that $r$ satisfies \emph{weak symmetry} when for all $x,y,z,w \in X$ and $a,b \in Q$, if $(x, a \otimes \sigma_{w}, y) \in r$ and $(x, b, z) \in r$ and $a \otimes b \sqsubseteq \sigma_{y}$, then $(y, a \otimes b, z) \in r$.
\end{defi}
We provide some explanation of weak symmetry. We call this condition \emph{weak symmetry} because, roughly speaking, it is a combination of symmetry and transitivity: if we ignore the middle component of $r$, the weak symmetry property states that if $(x, y) \in r$ and $(x, z) \in r$, then $(y, z) \in r$. We also explain the first assumption, $(x, a \otimes \sigma_{w}, y) \in r$, of the weak symmetry property, where $w \in X$ is arbitrary and independent of $x, y, z \in X$. Intuitively, the assumption implies that $a$ bounds the ``pointwise distance'' between $x$ and $y$: for example, for functions $f,g \in \uset{(\real \Rightarrow \real) \Rightarrow \real}$ and $d \in \mathcal{Q}_{(\real \Rightarrow \real) \Rightarrow \real}$, if $(f, d \otimes \sigma_{h}, g) \in \rho_{(\real \Rightarrow \real) \Rightarrow \real}$ for some $h \in \uset{(\real \Rightarrow \real) \Rightarrow \real}$, then for every $u \in \uset{\real \Rightarrow \real}$, we have
\begin{equation*}
  (fu, d(u, \sigma_{u}), gu) =
  (fu, d(u, \sigma_{u}) \lawotimes \sigma_{hu}, gu) =
  (fu, d(u, \sigma_{u}) \lawotimes \sigma_{h}(u, \sigma_{u}), gu) \in \rho_{\real}.
\end{equation*}
On the other hand, thanks to $\sigma_{h}$, the function $d$ does not need to bound the distance between $f$ and $g$, that is, $d$ does not need to satisfy $(fu, d(u, a), gv) \in \rho_{\mathsf{Real}}$ for all $(u, a, v) \in \rho_{\real \Rightarrow \real}$. This relaxation is crucial in the proof of Corollary~\ref{cor:another} where we show that for every type $\mathsf{A}$, the quasi-quasi-metric $\rho_{\mathsf{A}}$ also satisfies another form of self-indistancy. By the definition of weak symmetry, it is straightforward to see that for every $Q$-relation $\phi \colon X \times X \to Q$, the associated $Q$-predicate $\check{\phi}$ satisfies the weak symmetry condition if and only if for all $x,y,z,w \in X$ and $a,b \in Q$, if $a \otimes \sigma_{w} \in \phi(x, y)$ and $b \in \phi(x, z)$ and $a \otimes b \sqsubseteq \sigma_{y}$, then $a \otimes b \in \phi(y,z)$. It remains open how to characterize weak symmetry in terms of $Q$-relations without auxiliary bounds $a,b \in Q$.
% Since $\rho_{\real}$ is symmetric, we also have $(gu, d(u, \sigma_{u}), fu) \in \rho_{\real}$. In this way, we
% lift symmetry of $\rho_{\real}$ to weak symmetry of $\rho_{\mathsf{A}}$ for all types $\mathsf{A}$.

The following proposition gives a consequence of weak symmetry.
\begin{prop}\label{cor:sym}
  Let $r \subseteq X \times Q \times X$ be a closed $Q$-predicate. For all $x, y, z \in X$ and $a \in Q$, if $r$ satisfies weak symmetry and we have $(x, a \otimes \sigma_{z}, y) \in r$ and $a \otimes \sigma_{x} \sqsubseteq \sigma_{y}$, then $(y, a \otimes \sigma_{x}, x) \in r$.
  % \begin{equation*}
  %   a^{\circ} = \sigma_{x} \otimes ((\sigma_{x} \sqcap \sigma_{y}) \multimap a).
  % \end{equation*}
\end{prop}
\begin{proof}
  Since $(x, \sigma_{x}, x) \in r$, the claim follows from the definition of weak symmetry.
  % We write $a^{\bullet}$ for
  % \begin{equation*}
  %   (\sigma_{x} \sqcap \sigma_{y}) \otimes ((\sigma_{x} \sqcap \sigma_{y}) \multimap a).
  % \end{equation*}
  % Then, we have $a^{\bullet} \sqsubseteq a$, and therefore, $(x, a^{\bullet}, y) \in \rho_{\mathsf{A}}$.
  % Since $(x, \sigma_{x}, x) \in \rho_{\mathsf{A}}$, it follows from weak symmetry that $(y, \sigma_{y} \sqcap a^{\circ}, x)$ is an element of $\rho_{\mathsf{A}}$.
\end{proof}

\subsubsection{Self-indistancy, Left Strong Transitivity and Weak Symmetry for $\rho_{\mathsf{A}}$}
\label{sec:mainstrong}

We now show that every $\rho_{\mathsf{A}}$ satisfies self-indistancy, left strong transitivity and weak symmetry.
We first prepare a lemma.
\begin{lem}\label{lem:self}
  Let $\mathsf{A}$ and $\mathsf{B}$ be types such that $\rho_{\mathsf{A}}$ satisfies the self-indistancy condition. Then, for all $f \colon \uset{\mathsf{A}} \to \uset{\mathsf{B}}$ and $x \in \uset{\mathsf{A}}$, we have $\sigma_{f}(x, \sigma_{x}) = \sigma_{fx}$.
\end{lem}
\begin{proof}
  For all $f \colon \uset{\mathsf{A}} \to \uset{\mathsf{B}}$ and $x \in \uset{\mathsf{A}}$, we have
  \begin{equation*}
    \sigma_{f}(x, \sigma_{x})
    = \bigsqcap_{(x, \sigma_{x}, y) \in \rho_{\mathsf{A}}}
    \hat{\rho}_{\mathsf{B}}(fx, fy)
    = \hat{\rho}_{\mathsf{B}}(fx, fx)
    = \sigma_{fx},
  \end{equation*}
  where the first equality follows from Proposition~\ref{lem:arrow}, and the second equality follows from the self-indistancy condition of $\rho_{\mathsf{A}}$.
\end{proof}

We give our main observation in this section.

\begin{thm}\label{thm:self}
  For every type $\mathsf{A}$, the quasi-quasi-metric $\rho_{\mathsf{A}}$ satisfies the following properties.
  \begin{itemize}
  \item
    (Self-indistancy) For all $x,y \in \uset{\mathsf{A}}$,
    if $(x, \sigma_{x}, y) \in \rho_{\mathsf{A}}$, then $x = y$.
  \item
    (Left strong transitivity) For all $x, y, z \in \uset{\mathsf{A}}$ and $a, b \in \mathcal{Q}_{\mathsf{A}}$, if $(x, a \otimes \sigma_{y}, y) \in \rho_{\mathsf{A}}$ and $(y, b, z) \in \rho_{\mathsf{A}}$, then $(x, a \otimes b, z) \in \rho_{\mathsf{A}}$.
  \item
    (Weak symmetry) For all $x,y,z,w \in \uset{\mathsf{A}}$ and $a,b \in \mathcal{Q}_{\mathsf{A}}$, if $(x, a \otimes \sigma_{w}, y) \in \rho_{\mathsf{A}}$ and $(x, b, z) \in \rho_{\mathsf{A}}$ and $a \otimes b \sqsubseteq \sigma_{y}$, then $(y, a \otimes b, z) \in \rho_{\mathsf{A}}$.
  % \item
  %   (Weak symmetry) For any $x, y, z \in \uset{\mathsf{A}}$, and for any $a, b \in \mathcal{Q}_{\mathsf{A}}$,
  %   if $(x, (\sigma_{x} \sqcap \sigma_{y}) \otimes a, y) \in \rho_{\mathsf{A}}$ and $(x, b, z) \in \rho_{\mathsf{A}}$, then $(y, \sigma_{y} \sqcap (a \otimes b), z) \in \rho_{\mathsf{A}}$.
  \end{itemize}
  % Here, for $x \in \uset{\mathsf{A}}$, we define $\sigma_{x} \in \mathcal{Q}_{\mathsf{A}}$ to be $\hat{\rho}_{\mathsf{A}}(x, x)$, i.e.,
  % \begin{equation*}
  %   \sigma_{x} = \bigsqcup \{a \in \mathcal{Q}_{\mathsf{A}}
  %   \mid (x, a, x) \in \rho_{\mathsf{A}}\}.
  % \end{equation*}
  % We call $\sigma_{x}$ the \emph{self-distance} of $x$.
\end{thm}
\begin{proof}
  We prove the claim by induction on $\mathsf{A}$. For the base case $\mathsf{A} = \real$, for every $x \in \mathbb{R}$, the self-distance $\sigma_{x}$ is equal to $0$. Therefore, the claim follows from the fact that $\rho_{\real}$ is a metric on $\mathbb{R}$. For the case $\mathsf{A} \times \mathsf{B}$, for every $(x, y) \in \uset{\mathsf{A} \times \mathsf{B}}$, we have $\sigma_{(x, y)} = (\sigma_{x}, \sigma_{y})$, and the claim follows from the induction hypotheses for $\mathsf{A}$ and $\mathsf{B}$. It remains to check the case $\mathsf{A} \Rightarrow \mathsf{B}$ where we use Lemma~\ref{lem:self}: by the induction hypothesis, for every $f \colon \uset{\mathsf{A}} \to \uset{\mathsf{B}}$ and $x \in \uset{\mathsf{A}}$, we have $\sigma_{f}(x, \sigma_{x}) = \sigma_{fx}$. % We defer the proof of this fact to Lemma~\ref{lem:self}.
  \begin{itemize}
  \item
    (Self-indistancy) If $(f, \sigma_{f}, g) \in \rho_{\mathsf{A} \Rightarrow \mathsf{B}}$, then for every $x \in \uset{\mathsf{A}}$,
    \begin{equation*}
      (fx, \sigma_{fx}, gx)
      = (fx, \sigma_{f}(x, \sigma_{x}), gx)
      \in \rho_{\mathsf{B}}.
    \end{equation*}
    Hence, it follows from the self-indistancy of $\rho_{\mathsf{B}}$ that $fx$ is equal to $gx$ for every $x \in \uset{\mathsf{A}}$.
  \item
    (Left strong transitivity) We suppose $(f, d \otimes \sigma_{g}, g) \in \rho_{\mathsf{A} \Rightarrow \mathsf{B}}$ and $(g, e, h) \in \rho_{\mathsf{A} \Rightarrow \mathsf{B}}$. By the assumptions, for every $(x, a, y) \in \rho_{\mathsf{A}}$, we obtain
    \begin{equation*}
      (fx, d(x,\sigma_{x}) \otimes \sigma_{gx}, gx)
      =
      (fx,
      d(x,\sigma_{x}) \otimes
      \sigma_{g}(x,\sigma_{x}),
      gx)
      \in \rho_{\mathsf{B}},
      \quad
      (gx, e(x,a), hy) \in \rho_{\mathsf{B}}.
    \end{equation*}
    Thus, it follows from left strong transitivity of $\rho_{\mathsf{B}}$ that
    \begin{equation*}
      (fx, d(x, \sigma_{x}) \otimes
      e(x, a), hy) \in \rho_{\mathsf{B}}.
    \end{equation*}
    By quasi-reflexivity and $(x, a, y) \in \rho_{\mathsf{A}}$, we have $(x, a, x) \in \rho_{\mathsf{A}}$, and therefore, $a \sqsubseteq \sigma_{x}$. Then, by monotonicity of $d$, we obtain $d(x, a) \sqsubseteq d(x, \sigma_{x})$. Therefore, by closedness of $\rho_{\mathsf{B}}$,
    \begin{equation}\label{eq:fx_hy}
      (fx, d(x, a) \otimes
      e(x, a), hy) \in \rho_{\mathsf{B}}.
    \end{equation}
    It remains to check $(fx, d(x, a) \otimes e(x, a), fy)
    \in \rho_{\mathsf{B}}$ to deduce $(f, d \otimes e, h) \in \rho_{\mathsf{A} \Rightarrow \mathsf{B}}$. By applying quasi-reflexivity to the assumptions, we obtain
    \begin{equation*}
      (f, d \otimes \sigma_{g}, f) \in \rho_{\mathsf{A} \Rightarrow \mathsf{B}},
      \qquad
      (g, e, g) \in \rho_{\mathsf{A} \Rightarrow \mathsf{B}}.
    \end{equation*}
    Thus, we have
    \begin{equation*}
      (fx, d(x, a) \otimes \sigma_{g}(x, a), fy)
      \in
      \rho_{\mathsf{B}},
      \qquad
      e \sqsubseteq \sigma_{g}.
    \end{equation*}
    Hence, by closedness of $\rho_{\mathsf{B}}$,
    \begin{equation}\label{eq:fx_fy}
      (fx, d(x, a) \otimes e(x, a), fy)
      \in \rho_{\mathsf{B}}.
    \end{equation}
    Since \eqref{eq:fx_hy} and \eqref{eq:fx_fy} hold for arbitrary $(x, a, y) \in \rho_{\mathsf{A}}$, we conclude that $(f, d \otimes e, h)$ is an element of $\rho_{\mathsf{A} \Rightarrow \mathsf{B}}$.
  \item
    (Weak symmetry)
    We suppose that $(f, d \otimes \sigma_{k}, g) \in \rho_{\mathsf{A} \Rightarrow \mathsf{B}}$ and $(f, e, h) \in \rho_{\mathsf{A} \Rightarrow \mathsf{B}}$ and $d \otimes e \sqsubseteq \sigma_{g}$. Then, for every $(x, a, y) \in \rho_{\mathsf{A}}$, we have
    \begin{align}
      &
        (fx,
        d(x,\sigma_{x}) \otimes \sigma_{kx}
        , gx)
        =
        (fx, d(x,\sigma_{x}) \otimes \sigma_{k}(x, \sigma_{x}), gx)
        \in \rho_{\mathsf{B}},
        \label{eq:rdg}
      \\
      &
        (fx,e(x,a),hy) \in \rho_{\mathsf{B}}.
        \label{eq:feh}
      % \\
      % &
      %   \sigma_{fx} \otimes d(x,\sigma_{x})
      %   = \sigma_{f}(x, \sigma_{x}) \otimes d(x,\sigma_{x})
      %   \sqsubseteq \sigma_{g}(x, \sigma_{x}) = \sigma_{gx}.
    \end{align}
    By quasi-reflexivity and $(x, a, y) \in \rho_{\mathsf{A}}$, we have $a \sqsubseteq \sigma_{x}$. Hence, by closedness of $\rho_{\mathsf{B}}$, \eqref{eq:rdg} implies
    \begin{equation}\label{eq:rdga}
      (fx,
      d(x,a) \otimes \sigma_{kx}
      , gx)
      \in \rho_{\mathsf{B}}.
    \end{equation}
    We also have
    \begin{equation*}
      d(x, a) \otimes e(x, a) \sqsubseteq \sigma_{g}(x, a)
      \sqsubseteq \sigma_{g}(x, \sigma_{x}) = \sigma_{gx}.
    \end{equation*}
    Therefore, by \eqref{eq:rdga} and \eqref{eq:feh} and the weak symmetry of $\rho_{\mathsf{B}}$,
    \begin{equation}
      \label{eq:gdeh}
      (gx, d(x, a) \otimes e(x, a), hy) \in \rho_{\mathsf{B}}.
    \end{equation}
    Furthermore, since $d \otimes e \sqsubseteq \sigma_{g}$, by closedness of $\rho_{\mathsf{B}}$ and $(gx, \sigma_{g}(x,a), gy) \in \rho_{\mathsf{B}}$, we have
    \begin{equation}
      \label{eq:gdeg}
      (gx, d(x,a) \otimes e(x,a), gy) \in \rho_{\mathsf{B}}.
    \end{equation}
    Since \eqref{eq:gdeh} and \eqref{eq:gdeg} hold for arbitrary $(x, a, y) \in \rho_{\mathsf{A}}$, we see that $(g, d \otimes e, h)$ is an element of $\rho_{\mathsf{A} \Rightarrow \mathsf{B}}$.
  \end{itemize}
\end{proof}

Another indistancy property follows from weak symmetry.
\begin{cor}\label{cor:another}
  For every type $\mathsf{A}$, if $(x, \sigma_{y}, y) \in \rho_{\mathsf{A}}$, then $x = y$.
\end{cor}
\begin{proof}
  Suppose that $(x, \sigma_{y}, y) \in \rho_{\mathsf{A}}$. Then, by applying weak symmetry to $(x, \mathbf{1} \otimes \sigma_{y}, y) \in \rho_{\mathsf{A}}$ and $(x, \sigma_{y}, x) \in \rho_{\mathsf{A}}$, we obtain $(y, \sigma_{y}, x) \in \rho_{\mathsf{A}}$. It follows from self-indistancy that $y$ is equal to $x$.
  % \begin{equation*}
  %   \sigma_{x} \otimes ((\sigma_{x} \sqcap \sigma_{y}) \multimap \sigma_{y}) = \sigma_{x},
  % \end{equation*}
  % it follows from Corollary~\ref{cor:sym} that $(y, \sigma_{y} \sqcap \sigma_{x}, x) \in \rho_{\mathsf{A}}$. By quasi-reflexivity of $\rho_{A}$ and $(x, \sigma_{y}, y) \in \rho_{\mathsf{A}}$, we obtain $\sigma_{y} \sqsubseteq \sigma_{x}$. Therefore, $(y, \sigma_{y}, x) \in \rho_{\mathsf{A}}$. By self-indistancy, we obtain $x = y$.
\end{proof}

\begin{rem}
  Let $\mathbf{s}\qqm$ be the full subcategory of $\qqm$ consisting of quasi-quasi-metric spaces satisfying self-indistancy, left strong transitivity, and weak symmetry. Following the argument in the proof of Theorem~\ref{thm:self}, we can establish that $\mathbf{s}\qqm$ is cartesian closed. Our preference for quasi-quasi-metric spaces over those with the additional properties is motivated by their seamless connection to the quantitative equational theory given in Section~\ref{sec:concr-descr-least} as well as quantale-valued predicates. In fact, while it is straightforward to incorporate the axioms of quasi-quasi-metric into the inference rules of the quantitative equational theory, it is less clear how to accommodate self-indistancy, left strong transitivity, and weak symmetry due to their dependence on self-distances. This difficulty is essentially equivalent to that of characterizing left strong transitivity in terms of quantale-valued predicates without using self-distances (see Proposition~\ref{prop:lstrong}).
\end{rem}

%%% Local Variables:
%%% mode: LaTeX
%%% TeX-master: "lmcs_main.tex"
%%% End:

\section{A Poset of Differential Prelogical Relations}
\label{sec:towards-lattice-mmms}

The category $\qqm$ is not the only way to measure distances between terms. In this section, we introduce a notion of differential prelogical relations as ``quasi-quasi-metrics'' on $\LL$. The set of differential prelogical relations has a partial order $R \subseteq S$, which intuitively means that the ``quasi-quasi-metric'' $R$ is finer than the ``quasi-quasi-metric'' $S$. Our purpose in this section is to examine the order theoretic structure of this poset, motivated by the structure of the lattices of program equivalences.

\subsection{Differential Prelogical Relations}
\label{sec:diff-logic-relat}

We introduce differential prelogical relations as type-indexed families of quasi-quasi-metrics over the sets of \emph{definable elements}, following the idea of prelogical relations \cite{HD02}. Here, for a type $\mathsf{A}$, we say that an element $x \in \uset{\mathsf{A}}$ is \emph{$\LL$-definable} when there is a closed term $\mathsf{t}$ of type $\mathsf{A}$ such that $x = \sem{\mathsf{t}}$. We denote the subset of $\uset{\mathsf{A}}$ consisting of $\LL$-definable elements by $\dfn{\mathsf{A}}$.

\begin{defi}
  A \emph{differential prelogical relation} $R$ is a type-indexed family of closed quasi-quasi-metrics
  \begin{equation*}
    R_{\mathsf{A}}
    \subseteq
    \dfn{\mathsf{A}} \times
    \mathcal{Q}_{\mathsf{A}} \times
    \dfn{\mathsf{A}}
  \end{equation*}
  subject to the following conditions:
  \begin{itemize}
  \item
    For every $(x, a, y) \in \dfn{\real} \times
    \mathcal{Q}_{\real} \times \dfn{\real}$,
    \begin{equation*}
      (x, a, y) \in R_{\real} \iff |x - y| \lawsqsupseteq a .
    \end{equation*}
  \item
    For all term $\mathsf{\Gamma} \vdash \mathsf{t} : \mathsf{A}$ and $(\gamma, \xi, \varepsilon) \in R_{\mathsf{\Gamma}}$,
    \begin{equation*}
      (\semp{\mathsf{t}}{\gamma},
      \fsemp{\mathsf{t}}{\gamma}{\xi},
      \semp{\mathsf{t}}{\varepsilon})
      \in R_{\mathsf{A}}.
    \end{equation*}
  \end{itemize}
  We refer to the second condition as the \emph{fundamental property}.
\end{defi}
The fundamental property requires that the interpretation of any term $\mathsf{t}$ preserves differential prelogical relations. It follows from the fundamental property that differential prelogical relations are closed under function applications, projections, and tuplings.
\begin{prop}\label{prop:dpr}
  Let $R$ be a differential prelogical relation.
  \begin{itemize}
  \item
    If $(f, d, g) \in R_{\mathsf{A} \Rightarrow \mathsf{B}}$ and $(x, a, y) \in R_{\mathsf{A}}$, then $(fx, d(x, a), gy) \in R_{\mathsf{B}}$.
  \item
    If $((x, y), (a, b), (z, w)) \in R_{\mathsf{A} \times \mathsf{B}}$, then $(x, a, z) \in R_{\mathsf{A}}$ and $(y, b, w) \in R_{\mathsf{B}}$.
  \item
    If $(x, a, z) \in R_{\mathsf{A}}$ and $(y, b, w) \in R_{\mathsf{B}}$, then $((x, y), (a, b), (z, w)) \in R_{\mathsf{A} \times \mathsf{B}}$.
  \end{itemize}
\end{prop}
\begin{proof}
  We can deduce the first clause by applying the fundamental property to
  \begin{equation*}
    \mathsf{x} : \mathsf{A} \Rightarrow \mathsf{B}, \mathsf{y} : \mathsf{A} \vdash \mathsf{x} \, \mathsf{y} : \mathsf{B}
  \end{equation*}
  and $((f,x),(d,a),(g,y)) \in R_{(\mathsf{x} : \mathsf{A} \Rightarrow \mathsf{B}, \mathsf{y} : \mathsf{A} )}$. Indeed, since $\fsemp{\mathsf{x} \, \mathsf{y}}{(f,x)}{(d, a)} = d(x, a)$, we obtain $(fx, d(x, a), gy) \in R_{\mathsf{B}}$. We can deduce the second clause by applying the fundamental property to
  \begin{equation*}
    \mathsf{x} : \mathsf{A} \times \mathsf{B} \vdash \fst(\mathsf{x}) : \mathsf{A}
    \text{ and }
    \mathsf{x} : \mathsf{A} \times \mathsf{B} \vdash \snd(\mathsf{x}) : \mathsf{B},
  \end{equation*}
  whose derivatives are given by $\fsemp{\fst(\mathsf{x})}{(x,y)}{(a,b)} = a$ and $\fsemp{\snd(\mathsf{x})}{(x,y)}{(a,b)} = b$.
  We can deduce the last clause by applying the fundamental property to
  \begin{equation*}
    \mathsf{x} : \mathsf{A}, \mathsf{y} : \mathsf{B} \vdash \pair{\mathsf{x}}{\mathsf{y}} : \mathsf{A} \times \mathsf{B},
  \end{equation*}
  whose derivative is given by $\fsemp{\pair{\mathsf{x}}{\mathsf{y}}}{(x, y)}{(a, b)} = (a, b)$.
\end{proof}
While a prelogical relation $R$ introduced in \cite{HD02} consists of a type-indexed family of relations $\{R_{\mathsf{A}} \subseteq U_{\mathsf{A}} \times V_{\mathsf{A}}\}_{\mathsf{A} \in \type}$ over type-indexed families of sets $\{U_{\mathsf{A}}\}_{\mathsf{A} \in \type}$ and $\{V_{\mathsf{A}}\}_{\mathsf{A} \in \type}$, a differential prelogical relation consists of a type-indexed family of quasi-quasi-metrics over the \emph{fixed family of sets $\{\dfn{\mathsf{A}}\}_{\mathsf{A} \in \type}$}. We restrict the underlying sets of differential prelogical relations in order to simplify the definition and the comparison of differential prelogical relations. It is, however, straightforward to generalize the definition of differential prelogical relations following the original framework of prelogical relations.

Below, we give a few examples of differential prelogical relations. We defer proving the fundamental property of each example to Section~\ref{sec:fundamental-lemma}.

\begin{exa}\label{eg:QDLR}
  We define a differential prelogical relation
  \begin{equation*}
    \tau = \{\tau_{\mathsf{A}} \subseteq \dfn{\mathsf{A}} \times \mathcal{Q}_{\mathsf{A}} \times \dfn{\mathsf{A}}\}_{\mathsf{A} \in \type}
  \end{equation*}
  by restricting the family of closed quasi-quasi-metrics $\{\rho_{\mathsf{A}}\}_{\mathsf{A} \in \type}$ to $\LL$-definable elements:
  \begin{equation*}
    \tau_{\mathsf{A}}
    = \rho_{\mathsf{A}} \cap
    (\dfn{\mathsf{A}} \times
    \qt{\mathsf{A}} \times \dfn{\mathsf{A}}).
  \end{equation*}
  The $\mathcal{Q}_{\mathsf{A}}$-valued predicate $\tau_{\mathsf{A}}$ is a closed quasi-quasi-metric on $\dfn{\mathsf{A}}$ since $\rho_{\mathsf{A}}$ is a closed quasi-quasi-metric.
\end{exa}
Taking intersections with definable elements is not the only way to construct a differential prelogical relation since the definition of differential prelogical relation does not prescribe how $R_{\mathsf{A} \Rightarrow \mathsf{B}}$ is determined by $R_{\mathsf{A}}$ and $R_{\mathsf{B}}$. In the next example, we give a differential prelogical relation $\delta$ where $(f,d,g) \in \delta_{\mathsf{A} \Rightarrow \mathsf{B}}$ holds if and only if $(f,d,g)$ is required to be compatible only with $(x, a, y) \in \delta_{\mathsf{A}}$. By contrast, in Example~\ref{eg:QDLR}, $(f,d,g) \in \tau_{\mathsf{A} \Rightarrow \mathsf{B}}$ is required to be compatible with \emph{all} $(x, a, y) \in \rho_{\mathsf{A}}$, where $x$ and $y$ may be undefinable.
\begin{exa}\label{eg:LDLR}
  We introduce a differential prelogical relation $\delta$ given by observing distances between functions at \emph{$\LL$-definable} elements: we define a differential prelogical relation $\delta = \{\delta_{\mathsf{A}} \subseteq \dfn{\mathsf{A}} \times \qt{\mathsf{A}} \times \dfn{\mathsf{A}}\}_{\mathsf{A} \in \type}$ by
  \begin{align*}
    (x, a, y)
    \in \delta_{\real}
    &\iff
      |x - y| \lawsqsupseteq a,
    \\
    ((x, y), (a, b), (z, w))
    \in
    \delta_{\mathsf{A} \times \mathsf{B}}
    &\iff
      (x, a, z)
      \in \delta_{\mathsf{A}}
      \text{ and }
      (y, b, w)
      \in \delta_{\mathsf{B}},
    \\
    (f, d, g) \in
    \delta_{\mathsf{A} \Rightarrow \mathsf{B}}
    &\iff
      \text{for every }
      (x, a, y)
      \in \delta_{\mathsf{A}},
    \\
    &\mathrel{\phantom{\iff}}
      (fx, d(x, a), gy)
      \in \delta_{\mathsf{B}}
      \text{ and }
      (fx, d(x, a), fy)
      \in \delta_{\mathsf{B}}.
  \end{align*}
  We note that in the definition of $\delta_{\mathsf{A} \Rightarrow \mathsf{B}}$, we require $f$ and $g$ to be $\LL$-definable. By induction on $\mathsf{A}$, we can show that $\delta_{\mathsf{A}}$ are closed quasi-quasi-metrics following the lines of the proof of Lemma~\ref{lem:product} and Proposition~\ref{prop:ccc}. A similar construction of a differential logical relation can be found in \cite{DGY19}.
\end{exa}

\begin{exa}\label{eg:aDLR}
  We give a construction of differential prelogical relations that subsumes Example~\ref{eg:QDLR} and Example~\ref{eg:LDLR}. The construction is parameterized by a $\LL$-closed predicate: we call a type-indexed family of subsets
  \begin{equation*}
    P_{\mathsf{A}} \subseteq \uset{\mathsf{A}}
  \end{equation*}
  a \emph{$\LL$-closed predicate} whenever for all $\mathsf{\Gamma} \vdash \mathsf{t} : \mathsf{A}$ and $\gamma \in P_{\mathsf{\Gamma}}$, we have $\sem{\mathsf{t}}_{\gamma} \in P_{\mathsf{A}}$. Given a $\LL$-closed predicate $\{P_{\mathsf{A}}\}_{\mathsf{A} \in \type}$, we define a $\mathcal{Q}_{\mathsf{A}}$-predicate $S_{\mathsf{A}} \subseteq P_{\mathsf{A}} \times \qt{\mathsf{A}} \times P_{\mathsf{A}}$ by induction on $\mathsf{A}$ as follows:
  \begin{align*}
    (x, a, y)
    \in S_{\real}
    &\iff
      |x - y| \lawsqsupseteq a,
    \\
    ((x, y), (a, b), (z, w))
    \in
    S_{\mathsf{A} \times \mathsf{B}}
    &\iff
      (x, a, z)
      \in S_{\mathsf{A}}
      \text{ and }
      (y, b, w)
      \in S_{\mathsf{B}},
    \\
    (f, d, g) \in
    S_{\mathsf{A} \Rightarrow \mathsf{B}}
    &\iff
      \text{for every }
      (x, a, y)
      \in S_{\mathsf{A}},
    \\
    &\mathrel{\phantom{\iff}}
      (fx,
      d(x, a),
      gy)
      \in S_{\mathsf{B}}
      \text{ and }
      (fx,
      d(x, a),
      fy)
      \in S_{\mathsf{B}}.
  \end{align*}
  We can inductively show that for every type $\mathsf{A}$, the $\mathcal{Q}_{\mathsf{A}}$-predicate $S_{\mathsf{A}}$ is a closed quasi-quasi-metric on $P_{\mathsf{A}}$. We call the type-indexed family $\{S_{\mathsf{A}}\}_{\mathsf{A} \in \type}$ the \emph{type-indexed family associated to $P$}. We then define the differential prelogical relation $\{R_{\mathsf{A}}\}_{\mathsf{A} \in \type}$ to be the restriction of $\{S_{\mathsf{A}}\}_{\mathsf{A} \in \type}$ to $\LL$-definable elements, i.e., we define $R_{\mathsf{A}}$ by
  \begin{equation*}
    R_{\mathsf{A}} = S_{\mathsf{A}} \cap
    (\dfn{\mathsf{A}} \times \qt{\mathsf{A}}
    \times \dfn{\mathsf{A}}).
  \end{equation*}
  We call $R$ the \emph{differential prelogical relation associated to $P$}. For example, $\tau$ is the differential prelogical relation associated to $\{\uset{\mathsf{A}}\}_{\mathsf{A} \in \type}$, and $\delta$ is the differential prelogical relation associated to $\{\dfn{\mathsf{A}}\}_{\mathsf{A} \in \type}$. We give another example of a $\LL$-closed predicate obtained by extending $\LL$ with some functions. For example, we can extend $\LL$ with a term constructor
  \begin{equation*}
    \prftree{
      \mathsf{\Gamma} \vdash \mathsf{t} : \real \Rightarrow \real
    }{
      \mathsf{\Gamma} \vdash F(\mathsf{t}) : \real
    }
  \end{equation*}
  for some fixed function $F \colon (\mathbb{R} \Rightarrow \mathbb{R}) \to \mathbb{R}$ such as $Ff = \inf_{x \in \mathbb{R}} |fx|$. Let us denote this extended calculus by $\LL^{F}$. Then, we can naturally extend the interpretation function $\sem{-}$ to $\LL^{F}$. We now define $\{\dfn{\mathsf{A}}_{F} \subseteq \uset{\mathsf{A}}\}_{\mathsf{A} \in \type}$ by
  \begin{equation*}
    x \in \dfn{\mathsf{A}}_{F} \iff
    x = \sem{\mathsf{t}} \text{ for some closed term $\vdash \mathsf{t} : \mathsf{A}$ in $\LL^{F}$}.
  \end{equation*}
  It is straightforward to see that $\{\dfn{\mathsf{A}}_{F} \subseteq \uset{\mathsf{A}}\}_{\mathsf{A} \in \type}$ is a $\LL$-closed predicate.
\end{exa}

\subsection{The Fundamental Property}
\label{sec:fundamental-lemma}

We prove the fundamental property for all examples of differential prelogical relations given in Section~\ref{sec:diff-logic-relat}. Since these examples are instances of Example~\ref{eg:aDLR}, we only prove the fundamental property for the differential prelogical relation associated to a $\LL$-closed predicate. We first prove a stronger statement, from which we deduce the fundamental property.

\begin{thm}\label{thm:fl-strong}
  Let $S$ be the type-indexed family associated to a $\LL$-closed predicate $P$. For all $\mathsf{\Gamma} \vdash \mathsf{t} : \mathsf{A}$ and $(\gamma, \xi, \varepsilon) \in S_{\mathsf{\Gamma}}$, we have
  \begin{equation*}
    (\semp{\mathsf{t}}{\gamma},
    \fsemp{\mathsf{t}}{\gamma}{\xi},
    \semp{\mathsf{t}}{\varepsilon})
    \in S_{{\mathsf{A}}}.
  \end{equation*}
\end{thm}
\begin{proof}
  We prove the statement by induction on the derivation of $\mathsf{\Gamma} \vdash \mathsf{t} : \mathsf{A}$. We only check the cases $\mathsf{\Gamma} \vdash \alpha(\mathsf{t}_{1},\ldots,\mathsf{t}_{n}) : \real$ and  $\mathsf{\Gamma} \vdash \lam{\mathsf{x}}{\mathsf{A}}{\mathsf{s}} : \mathsf{A} \Rightarrow \mathsf{B}$.
  The remaining cases can be checked similarly. We first check the case $\mathsf{\Gamma} \vdash \alpha(\mathsf{t}_{1},\ldots,\mathsf{t}_{n}) : \real$. For every $(\gamma, \xi, \varepsilon) \in S_{\mathsf{\Gamma}}$, it follows from the induction hypothesis that for every $i \in \{1,\ldots,n\}$, we have
  \begin{equation*}
    |\semp{\mathsf{t}_{i}}{\gamma} -
    \semp{\mathsf{t}_{i}}{\varepsilon}|
    \lawsqsupseteq
    \fsemp{\mathsf{t}_{i}}{\gamma}{\xi}.
  \end{equation*}
  Therefore, by the definition of $\prt{\alpha}$, we obtain
  \begin{equation*}
    |\semp{\alpha(\mathsf{t}_{1}, \ldots,
      \mathsf{t}_{n})
    }{\gamma}
    - \semp{\alpha(\mathsf{t}_{1},
      \ldots,\mathsf{t}_{n})
    }{\varepsilon}|
    \lawsqsupseteq
    \fsemp{\alpha(\mathsf{t}_{1}, \ldots,
      \mathsf{t}_{n})}{\gamma}{\xi}.
  \end{equation*}
  Hence, by the definition of $S_{\real}$,
  \begin{equation*}
    (\semp{\alpha(\mathsf{t}_{1}, \ldots, \mathsf{t}_{n})}{\gamma}, \fsemp{\alpha(\mathsf{t}_{1}, \ldots, \mathsf{t}_{n})}{\gamma}{\xi}, \semp{\alpha(\mathsf{t}_{1}, \ldots, \mathsf{t}_{n})}{\varepsilon}) \in S_{\real}.
  \end{equation*}
  We next check the case $\mathsf{\Gamma} \vdash \lam{\mathsf{x}}{\mathsf{A}}{\mathsf{s}} : \mathsf{A} \Rightarrow \mathsf{B}$. Let $(\gamma, \xi, \varepsilon)$ be an element of $S_{\mathsf{\Gamma}}$. It follows from the induction hypothesis that for every $(x, a, y) \in S_{\mathsf{A}}$, we have
  \begin{align*}
    (\semp{
    \lam{\mathsf{x}}{\mathsf{A}}{\mathsf{s}}}{\gamma}
    (x),
    \fsemp{
    \lam{\mathsf{x}}{\mathsf{A}}{\mathsf{s}}}{\gamma}{\xi}
    (x, a),
    \semp{
    \lam{\mathsf{x}}{\mathsf{A}}{\mathsf{s}}}{\varepsilon}
    (y))
    &=
    (\semp{\mathsf{s}}{\gamma \add x},
    \fsemp{\mathsf{s}}{\gamma \add x}{\xi \add a},
      \semp{\mathsf{s}}{\varepsilon \add y}) \in S_{\mathsf{B}},
    \\
    (\semp{
    \lam{\mathsf{x}}{\mathsf{A}}{\mathsf{s}}}{\gamma}
    (x),
    \fsemp{
    \lam{\mathsf{x}}{\mathsf{A}}{\mathsf{s}}}{\gamma}{\xi}
    (x, a),
    \semp{
    \lam{\mathsf{x}}{\mathsf{A}}{\mathsf{s}}}{\gamma}
    (y))
    &=
    (\semp{\mathsf{s}}{\gamma \add x},
    \fsemp{\mathsf{s}}{\gamma \add x}{\xi \add a},
    \semp{\mathsf{s}}{\gamma \add y}) \in S_{\mathsf{B}}.
  \end{align*}
  The latter follows from the fact that $(x,a,y) \in S_{\mathsf{A}}$ and $(\gamma, \xi, \gamma) \in S_{\mathsf{\Gamma}}$, which is a consequence of quasi-reflexivity of $S_{\mathsf{\Gamma}}$. Since $P$ is $\LL$-closed, and since $\gamma$ and $\varepsilon$ are elements of $P_{\mathsf{\Gamma}}$, we see that $\semp{\lam{\mathsf{x}}{\mathsf{A}}{\mathsf{s}}}{\gamma}$ and $\semp{\lam{\mathsf{x}}{\mathsf{A}}{\mathsf{s}}}{\varepsilon}$ are elements of $P_{\mathsf{A} \Rightarrow \mathsf{B}}$. Hence,
  \begin{equation*}
    (\semp{\lam{\mathsf{x}}{\mathsf{A}}{\mathsf{s}}}{\gamma}, \fsemp{\lam{\mathsf{x}}{\mathsf{A}}{\mathsf{s}}}{\gamma}{\xi}, \semp{\lam{\mathsf{x}}{\mathsf{A}}{\mathsf{s}}}{\varepsilon})
  \end{equation*}
  is an element of $S_{\mathsf{A} \Rightarrow \mathsf{B}}$.
\end{proof}

We then derive the fundamental property of the differential prelogical relation associated to a $\LL$-closed predicate.
\begin{cor}[Fundamental Property]\label{cor:fl}
  Let $R$ be the differential prelogical relation associated to a $\LL$-closed predicate $P$. For all $\mathsf{\Gamma} \vdash \mathsf{t} : \mathsf{A}$ and $(\gamma, \xi, \varepsilon) \in R_{\mathsf{\Gamma}}$, we have
  \begin{equation*}
    (\semp{\mathsf{t}}{\gamma},
    \fsemp{\mathsf{t}}{\gamma}{\xi},
    \semp{\mathsf{t}}{\varepsilon})
    \in R_{{\mathsf{A}}}.
  \end{equation*}
\end{cor}
\begin{proof}
  Since $\gamma$ and $\varepsilon$ consist of $\LL$-definable elements, $\semp{\mathsf{t}}{\gamma}$ and $\semp{\mathsf{t}}{\varepsilon}$ are $\LL$-definable. Hence, it follows from Theorem~\ref{thm:fl-strong} that $(\semp{\mathsf{t}}{\gamma}, \fsemp{\mathsf{t}}{\gamma}{\xi}, \semp{\mathsf{t}}{\varepsilon})$ is an element of $R_{{\mathsf{A}}}$.
\end{proof}

\subsection{Poset of Differential Prelogical Relations}
\label{sec:partial-order-non}

Let $\mathfrak{R}$ be the set of differential prelogical relations. For differential prelogical relations $R$ and $S$, we write $R \subseteq S$ whenever for every type $\mathsf{A}$, we have $R_{\mathsf{A}} \subseteq S_{\mathsf{A}}$. This defines a partial order on $\mathfrak{R}$. By the definition of the partial order, $R \subseteq S$ holds if and only if $R$ is finer than $S$, i.e., for all $x, y \in \dfn{\mathsf{A}}$, we have $\hat{R}_{\mathsf{A}}(x, y) \sqsubseteq \hat{S}_{\mathsf{A}}(x, y)$. Our question is whether the poset $\mathfrak{R}$ has an order theoretic structure analogous to that of the lattices of program equivalences; specifically, whether $\mathfrak{R}$ is a complete lattice, and whether $\mathfrak{R}$ has a greatest element and a least element.

It is straightforward to check that the poset $\mathfrak{R}$ has arbitrary non-empty meets.
\begin{prop}\label{prop:meet}
  For every non-empty set of differential prelogical relations $\mathfrak{S} \subseteq \mathfrak{R}$, its intersection
  \begin{equation*}
    S_{\mathsf{A}} = \bigcap_{R \in \mathfrak{S}} R_{\mathsf{A}}
  \end{equation*}
  is a differential prelogical relation and is the greatest lower bound of $\mathfrak{S}$ in $\mathfrak{R}$.
\end{prop}
In particular, since $\mathfrak{R}$ is not empty, the poset $\mathfrak{R}$ has the least differential prelogical relation $\nu$ given by
\begin{equation*}
  \nu_{\mathsf{A}} =
  \bigcap_{R \in \mathfrak{R}} R_{\mathsf{A}}.
\end{equation*}
As anticipated, we can present the least differential prelogical relation $\nu$ by a quantitative equational theory. This is what we prove in Section~\ref{sec:concr-descr-least}. Contrary to our expectations, however, we observe in Section~\ref{sec:absence-great-elem} that $\mathfrak{R}$ does not necessarily have a greatest differential prelogical relation.

\subsection{The Least Differential Prelogical Relation as a Quantitative Equational Theory}
\label{sec:concr-descr-least}

We give a concrete presentation of the least differential prelogical relation $\nu \in \mathfrak{R}$. To this end, in Figure~\ref{fig:eqt}, we give a quantitative equational theory $\Theta_{\real}$. In the derivation rules, each judgement $\mathsf{\Gamma} \vdash (\mathsf{t}, a, \mathsf{s}) : \mathsf{A}$ consists of a function $a \in (\uset{\mathsf{\Gamma}}, \mathcal{Q}_{\mathsf{\Gamma}}) \rightarrowtriangle \mathcal{Q}_{\mathsf{A}}$ and a pair of terms $\mathsf{\Gamma} \vdash \mathsf{t} : \mathsf{A}$ and $\mathsf{\Gamma} \vdash \mathsf{s} : \mathsf{A}$. The rules \textsf{VAR} and \textsf{FUN} come from the definitions of the derivatives $\fsem{\mathsf{x}}$ and $\fsem{\alpha(\mathsf{t}_{1},\ldots,\mathsf{t}_{n})}$. The rules \textsf{ABS}, \textsf{APP}, \textsf{PAIR}, \textsf{FST} and \textsf{SND} are given adopting the cartesian closed structure of $\qqm$. For example, for terms $\mathsf{\Gamma}, \mathsf{x} : \mathsf{A} \vdash \mathsf{t}, \mathsf{s} : \mathsf{B}$ and for a function $a \in (\uset{\mathsf{\Gamma},\mathsf{x} : \mathsf{A}} , \mathcal{Q}_{\mathsf{\Gamma},\mathsf{x} : \mathsf{A}} )\rightarrowtriangle \mathcal{Q}_{\mathsf{B}}$, if $(\sem{\mathsf{t}},a,\sem{\mathsf{s}})$ is a morphism from $\psem{\mathsf{\Gamma},\mathsf{x}:\mathsf{A}}$ to $\psem{\mathsf{B}}$, then its currying is given by $(\sem{\lam{\mathsf{x}}{\mathsf{A}}{\mathsf{t}}},b,\sem{\lam{\mathsf{x}}{\mathsf{A}}{\mathsf{s}}}) \colon \psem{\mathsf{\Gamma}} \to \psem{\mathsf{A} \Rightarrow \mathsf{B}}$ where $b$ is the function given in the \textsf{ABS} rule. The \textsf{SEM} rule means that judgements are compatible with set theoretic semantics, and the last four rules correspond to the definition of quasi-quasi-metric spaces. While our idea is inspired by the quantitative equational theories given in \cite{Plotk} and \cite{Dahlqvist2023}, it differs in two respects: first, distances $a$ need not be real numbers, but are presented as functions between quantales; second, we replace non-expansiveness with local similarity.

\begin{figure}
  \fbox{
    \begin{math}
      \begin{array}{c}
        % real
        \prftree[r]{\textsf{REAL}}{
        a \in (\uset{\mathsf{\Gamma}},\mathcal{Q}_{\mathsf{\Gamma}})
        \rightarrowtriangle \mathcal{Q}_{\real}
        }{
        a(\gamma, \xi) = |x - y|
        \text{ for all $\gamma \in \uset{\mathsf{\Gamma}}$ and $\xi \in \mathcal{Q}_{\mathsf{\Gamma}}$}
        }{
        \mathsf{\Gamma} \vdash
        (\underline{x},
        a,
        \underline{y}) : \real
        }
        \\[9pt]
        % var
        \prftree[r]{\textsf{VAR}}{
        \begin{array}{c}
          \mathsf{\Gamma} =
          (\mathsf{x}_{1} : \mathsf{A}_{1},
          \ldots,
          \mathsf{x}_{n} : \mathsf{A}_{n})
          \qquad
          a \in (\uset{\mathsf{\Gamma}},\mathcal{Q}_{\mathsf{\Gamma}})
          \rightarrowtriangle \mathcal{Q}_{\mathsf{A}_{i}}
          \\[3pt]
          a(\gamma,\xi) = \xi|_{i}
          \text{ for all $\gamma \in \uset{\mathsf{\Gamma}}$ and $\xi \in \mathcal{Q}_{\mathsf{\Gamma}}$}
        \end{array}
        }{
        \mathsf{\Gamma} \vdash
        (\mathsf{x}_{i},
        a,
        \mathsf{x}_{i})
        : \mathsf{A}_{i}
        }
        \\[9pt]
        % function
        \prftree[r]{\textsf{FUN}}{
        \begin{array}{c}
          \mathsf{\Gamma} \vdash
          (\mathsf{t}_{1},
          a_{1},
          \mathsf{s}_{1})
          : \real
          \ \
          \ldots
          \ \
          \mathsf{\Gamma} \vdash
          (\mathsf{t}_{n},
          a_{n},
          \mathsf{s}_{n}) : \real
          \qquad
          b \in (\uset{\mathsf{\Gamma}},\mathcal{Q}_{\mathsf{\Gamma}})
          \rightarrowtriangle \mathcal{Q}_{\real}
          \\[3pt]
          b(\gamma,\xi) =
          \prt{\alpha}(
          (\sem{\mathsf{t}_{1}}_{\gamma},
          \ldots,
          \sem{\mathsf{t}_{n}}_{\gamma}),
          (a_{1}(\gamma,\xi),\ldots,a_{n}(\gamma,\xi)))
          \\[1pt]
          \hspace{16em}
          \text{for all $\gamma \in \uset{\mathsf{\Gamma}}$ and $\xi \in \mathcal{Q}_{\mathsf{\Gamma}}$}
        \end{array}
        }{
        \mathsf{\Gamma} \vdash
        (
        \alpha
        (\mathsf{t}_{1}, \ldots, \mathsf{t}_{n}),
        b,
        \alpha
        (\mathsf{s}_{1},
        \ldots,
        \mathsf{s}_{n})
        ) : \real
        }
        \\[9pt]
        \prftree[r]{\textsf{ABS}}{
        \begin{array}{c}
          \mathsf{\Gamma} , \mathsf{x} : \mathsf{A}
          \vdash
          (\mathsf{t}, a, \mathsf{s})
          : \mathsf{B}
          \qquad
          b \in (\uset{\mathsf{\Gamma}},\mathcal{Q}_{\mathsf{\Gamma}})
          \rightarrowtriangle \mathcal{Q}_{\mathsf{A} \Rightarrow \mathsf{B}}
          \\[3pt]
          b(\gamma, \xi)(x, c)
          = a(\gamma \add x, \xi \add c)
          \text{ for all $\gamma \in \uset{\mathsf{\Gamma}}$ and $\xi \in \mathcal{Q}_{\mathsf{\Gamma}}$}
        \end{array}
        }{
        \mathsf{\Gamma} \vdash
        (
        \lam{\mathsf{x}}{
        \mathsf{A}}{
        \mathsf{t}},
        b,
        \lam{\mathsf{x}}{
        \mathsf{A}}{
        \mathsf{s}}
        ) :
        \mathsf{A} \Rightarrow \mathsf{B}
        }
        \\[9pt]
        \prftree[r]{\textsf{APP}}{
        \begin{array}{c}
          \mathsf{\Gamma} \vdash
          (\mathsf{t}, a, \mathsf{s})
          : \mathsf{A} \Rightarrow \mathsf{B}
          \qquad
          \mathsf{\Gamma} \vdash
          (\mathsf{p}, b, \mathsf{q})
          : \mathsf{A}
          \qquad
          c \in (\uset{\mathsf{\Gamma}},\mathcal{Q}_{\mathsf{\Gamma}})
          \rightarrowtriangle \mathcal{Q}_{\mathsf{B}}
          \\[3pt]
          c(\gamma,\xi)
          =
          a(\gamma,\xi)
          (\sem{\mathsf{p}}_{\gamma}, b(\gamma,\xi))
          \text{ for all $\gamma \in \uset{\mathsf{\Gamma}}$ and $\xi \in \mathcal{Q}_{\mathsf{\Gamma}}$}
        \end{array}
        }{
        }{
        }{
        \mathsf{\Gamma} \vdash
        (\mathsf{t} \, \mathsf{p},
        c,
        \mathsf{s} \, \mathsf{q})
        : \mathsf{B}
        }
        \\[9pt]
        \prftree[r]{\textsf{PAIR}}{
        \mathsf{\Gamma} \vdash
        (\mathsf{t}, a, \mathsf{s})
        : \mathsf{A}
        }{
        \mathsf{\Gamma} \vdash
        (\mathsf{p}, b, \mathsf{q})
        : \mathsf{B}
        }{
        \mathsf{\Gamma} \vdash
        (\pair{\mathsf{t}}{\mathsf{p}},
        \langle a, b \rangle,
        \pair{\mathsf{s}}{\mathsf{q}})
        : \mathsf{A} \times \mathsf{B}
        }
        \\[9pt]
        \prftree[r]{\textsf{FST}}{
        \mathsf{\Gamma} \vdash
        (\mathsf{t},
        \langle a, b \rangle,
        \mathsf{s})
        : \mathsf{A} \times \mathsf{B}
        }{
        \mathsf{\Gamma} \vdash
        (\fst(\mathsf{t}),
        a,
        \fst(\mathsf{s}))
        : \mathsf{A}
        }
        \quad
        \prftree[r]{\textsf{SND}}{
        \mathsf{\Gamma} \vdash
        (\mathsf{t},
        \langle a, b \rangle,
        \mathsf{s})
        : \mathsf{A} \times \mathsf{B}
        }{
        \mathsf{\Gamma} \vdash
        (\snd(\mathsf{t}),
        b,
        \snd(\mathsf{s}))
        : \mathsf{B}
        }
        \\[9pt]
        \prftree[r]{\textsf{SEM}}{
        \mathsf{\Gamma} \vdash
        (\mathsf{t},a,\mathsf{s})
        : \mathsf{A}
        }{
        \mathsf{\Gamma} \vdash
        \mathsf{p} : \mathsf{A}
        }{
        \mathsf{\Gamma} \vdash
        \mathsf{q} : \mathsf{A}
        }{
        \sem{\mathsf{t}} = \sem{\mathsf{p}}
        }{
        \sem{\mathsf{s}} = \sem{\mathsf{q}}
        }{
        \mathsf{\Gamma} \vdash
        (\mathsf{p}, a, \mathsf{q})
        : \mathsf{A}
        }
        \\[9pt]
        \prftree[r]{\textsf{DOWN}}{
        \mathsf{\Gamma} \vdash
        (\mathsf{t}, a, \mathsf{s}) : \mathsf{A}
        }{
        b \sqsubseteq a
        }{
        \mathsf{\Gamma} \vdash
        (\mathsf{t}, b, \mathsf{s}) : \mathsf{A}
        }
        \quad
        \prftree[r]{\textsf{LUB}}{
        \mathsf{\Gamma} \vdash
        (\mathsf{t}, a_{i}, \mathsf{s})
        : \mathsf{A}
        }{
        \textstyle
        a = \bigsqcup_{i \in I} a_{i}
        }{
        \mathsf{\Gamma} \vdash
        (\mathsf{t}, a, \mathsf{s})
        : \mathsf{A}
        }
        \\[9pt]
        \prftree[r]{\textsf{TRAN}}{
        \mathsf{\Gamma} \vdash
        (\mathsf{t}, a, \mathsf{s}) : \mathsf{A}
        }{
        \mathsf{\Gamma} \vdash
        (\mathsf{s}, b, \mathsf{u}) : \mathsf{A}
        }{
        \mathsf{\Gamma} \vdash
        (\mathsf{t}, a \otimes b, \mathsf{u}) : \mathsf{A}
        }
        \qquad
        \prftree[r]{\textsf{REFL}}{
        \mathsf{\Gamma} \vdash
        (\mathsf{t}, a, \mathsf{s}) : \mathsf{A}
        }{
        \mathsf{\Gamma} \vdash
        (\mathsf{t}, a, \mathsf{t}) : \mathsf{A}
        }
      \end{array}
    \end{math}
}
\caption{Derivation Rules of $\Theta_{\real}$}
\label{fig:eqt}
\end{figure}

The quantitative equational theory $\Theta_{\real}$ induces a differential prelogical relation: for every type $\mathsf{A}$, we define a $\mathcal{Q}_{\mathsf{A}}$-predicate $\theta_{\mathsf{A}} \subseteq \dfn{\mathsf{A}} \times \mathcal{Q}_{\mathsf{A}} \times \dfn{\mathsf{A}}$ by
\begin{equation*}
  (x, a, y)
  \in \theta_{\mathsf{A}}
  \iff
  x = \sem{\mathsf{t}}
  \text{ and }
  y = \sem{\mathsf{s}}
  \text{ for some derivable judgement}
  \vdash
  (\mathsf{t},
  a,
  \mathsf{s}) : \mathsf{A}
\end{equation*}
where we identify $a \in \mathcal{Q}_{\mathsf{A}}$ with the constant function $a \colon \{\ast\} \times \{\ast\} \to \mathcal{Q}_{\mathsf{A}}$ taking the value $a$. We first show that $\theta$ is bounded by an arbitrary differential prelogical relation.
\begin{prop}\label{prop:theta_least}
  For every differential prelogical relation $R$ and every type $\mathsf{A}$, we have $\theta_{\mathsf{A}} \subseteq R_{\mathsf{A}}$.
\end{prop}
\begin{proof}
  We introduce some notations. Let $\mathsf{\Gamma} = (\mathsf{x}_{1} : \mathsf{A}_{1}, \ldots, \mathsf{x}_{n}: \mathsf{A}_{n})$ be a typing environment. For a type $\mathsf{A}$, we define $\mathsf{\Gamma} \Rightarrow \mathsf{A}$ to be
  \begin{equation*}
    % \mathsf{\Gamma} \Rightarrow \mathsf{A} =
    \mathsf{A}_{1} \Rightarrow \cdots \Rightarrow \mathsf{A}_{n} \Rightarrow \mathsf{A},
    % \begin{cases}
    %   \mathsf{A}_{1} \Rightarrow \cdots \Rightarrow \mathsf{A}_{n} \Rightarrow \mathsf{A},
    %   & \text{if $n \geq 1$};
    %   \\
    %   \mathsf{A},
    %   & \text{if $n = 0$};
    % \end{cases}
  \end{equation*}
  and for a term $\mathsf{\Gamma} \vdash \mathsf{t} : \mathsf{A}$, we define a closed term $\vdash \lambda \mathsf{\Gamma}.\,\mathsf{t} : \mathsf{\Gamma} \Rightarrow \mathsf{A}$ to be
  \begin{equation*}
    % \lambda \mathsf{\Gamma}.\, \mathsf{t} =
    \lam{\mathsf{x}_{1}}{\mathsf{A}_{1}}{\cdots \lam{\mathsf{x}_{n}}{\mathsf{A}_{n}}{\mathsf{t}}}.
    % \begin{cases}
    %   \lam{\mathsf{x}_{1}}{\mathsf{A}_{1}}{\cdots \lam{\mathsf{x}_{n}}{\mathsf{A}_{n}}{\mathsf{t}}},
    %   & \text{if $n \geq 1$};
    %   \\
    %   \mathsf{t},
    %   & \text{if $n = 0$}.
    % \end{cases}
  \end{equation*}
  Similarly, for $a \in (\uset{\mathsf{\Gamma}}, \mathcal{Q}_{\mathsf{\Gamma}}) \rightarrowtriangle \mathcal{Q}_{\mathsf{A}}$, we define $\lambda \mathsf{\Gamma}.\, a \in \mathcal{Q}_{\mathsf{\Gamma} \Rightarrow \mathsf{A}}$ to be
  \begin{align*}
    (\lambda \mathsf{\Gamma}.\, a)(x_{1}, b_{1}) \cdots (x_{n}, b_{n}) =
    a((x_{1}, \ldots, x_{n}), (b_{1}, \ldots, b_{n})).
    % &(\lambda \mathsf{\Gamma}.\, a)(x_{1}, a_{1}) \cdots (x_{n}, a_{n}) =
    % a((x_{1}, \ldots, x_{n}), (a_{1}, \ldots, a_{n})),
    % & \text{if $n \geq 1$};
    % \\
    % &\lambda \mathsf{\Gamma}.\, a = a(\ast, \ast),
    % & \text{if $n = 0$}.
  \end{align*}
  To deduce the claim, given a differential prelogical relation $R$, we show that for every typing environment $\mathsf{\Gamma} = (\mathsf{x}_{1}:\mathsf{A}_{1},\ldots,\mathsf{x}_{n}:\mathsf{A}_{n})$ and every derivable judgement $\mathsf{\Gamma} \vdash (\mathsf{t}, a, \mathsf{s}) : \mathsf{A}$, we have
  \begin{equation*}
    (\sem{\lambda \mathsf{\Gamma}.\,\mathsf{t}},
    \lambda \mathsf{\Gamma}.\, a,
    \sem{\lambda \mathsf{\Gamma}.\,\mathsf{s}})
    \in R_{\mathsf{\Gamma} \Rightarrow \mathsf{A}}.
  \end{equation*}
  The claim $\theta_{\mathsf{A}} \subseteq R_{\mathsf{A}}$ follows from the case $\mathsf{\Gamma} = \varnothing$. We prove the above claim by induction on the derivation of $\mathsf{\Gamma} \vdash (\mathsf{t}, a, \mathsf{s}) : \mathsf{A}$, using the fundamental property of $R$. Here, we only prove three cases:
  \begin{itemize}
  \item
    Suppose that $\mathsf{\Gamma} \vdash (\alpha(\mathsf{t}_{1},\ldots,\mathsf{t}_{n}), b, \alpha(\mathsf{s}_{1},\ldots,\mathsf{s}_{n})) : \real$ is derived by applying the \textsf{FUN} rule to
    \begin{equation*}
      \mathsf{\Gamma} \vdash (\mathsf{t}_{1},a_{1},\mathsf{s}_{1}) : \real,
      \quad
      \ldots
      \quad
      \mathsf{\Gamma} \vdash (\mathsf{t}_{n},a_{n},\mathsf{s}_{n}) : \real,
    \end{equation*}
    and
    \begin{equation*}
      b(\gamma,\xi) = \prt{\alpha}((\semp{\mathsf{t}_{1}}{\gamma},\ldots,
      \semp{\mathsf{t}_{n}}{\gamma}),
      (a_{1}(\gamma,\xi),\ldots,a_{n}(\gamma,\xi))).
    \end{equation*}
    For simplicity, we assume $n = 1$ and $\mathsf{\Gamma} = (\mathsf{x}: \mathsf{A})$. By the induction hypothesis, we have
    \begin{equation*}
      (\sem{\lam{\mathsf{x}}{\mathsf{A}}{\mathsf{t}_{1}}},
      a_{1}, \sem{\lam{\mathsf{x}}{\mathsf{A}}{\mathsf{s}_{1}}})
      \in R_{\mathsf{A} \Rightarrow \real}.
    \end{equation*}
    We note that when the length of $\mathsf{\Gamma}$ is $1$, we have $a = \lambda \mathsf{\Gamma}.\,a$. Hence, it follows from the fundamental property that
    \begin{align*}
      (
      \semp{\mathsf{h}}{\sem{\lam{\mathsf{x}}{\mathsf{A}}{\mathsf{t}_{1}}}},
      \fsemp{\mathsf{h}}{\sem{\lam{\mathsf{x}}{\mathsf{A}}{\mathsf{t}_{1}}}}{a_{1}},
      \semp{\mathsf{h}}{\sem{\lam{\mathsf{x}}{\mathsf{A}}{\mathsf{s}_{1}}}}
      ) \in R_{\mathsf{A} \Rightarrow \real},
    \end{align*}
    where the term $\mathsf{z} : \mathsf{A} \Rightarrow \real \vdash \mathsf{h} : \mathsf{A} \Rightarrow \real$ is given by
    \begin{equation*}
      \mathsf{h} = \lam{\mathsf{y}}{\mathsf{A}}{\alpha(\mathsf{z} \, \mathsf{y})}.
    \end{equation*}
    To be concrete, $\fsem{\mathsf{h}}$ is given by
    \begin{equation*}
      \fsemp{\mathsf{h}}{z}{d}(y, c)
      = \fsemp{\alpha(\mathsf{z} \, \mathsf{y})}{(z,y)}{(d,c)}
      = \prt{\alpha}(\semp{\mathsf{z} \, \mathsf{y}}{(z,y)}, \fsemp{\mathsf{z} \, \mathsf{y}}{(z,y)}{(d,c)})
      = \prt{\alpha}(zy, d(y, c)).
    \end{equation*}
    Thus, we have
    \begin{align*}
      \semp{\mathsf{h}}{\sem{\lam{\mathsf{x}}{\mathsf{A}}{\mathsf{t}_{1}}}}
      &= \sem{\lam{\mathsf{x}}{\mathsf{A}}{\alpha(\mathsf{t}_{1})}},
      \\
      \semp{\mathsf{h}}{\sem{\lam{\mathsf{x}}{\mathsf{A}}{\mathsf{s}_{1}}}}
      &= \sem{\lam{\mathsf{x}}{\mathsf{A}}{\alpha(\mathsf{s}_{1})}},
      \\
      \fsemp{\mathsf{h}}{\sem{\lam{\mathsf{x}}{\mathsf{A}}{\mathsf{t}_{1}}}}{a_{1}}
      &= b.
    \end{align*}
    Hence, $(\sem{\lam{\mathsf{x}}{\mathsf{A}}{\alpha(\mathsf{t}_{1})}},
    b, \sem{\lam{\mathsf{x}}{\mathsf{A}}{\alpha(\mathsf{s}_{1})}}) \in R_{\mathsf{A} \Rightarrow \real}$.
  \item
    Suppose that $\mathsf{\Gamma} \vdash (\lam{\mathsf{x}}{\mathsf{A}}{\mathsf{t}},b,\lam{\mathsf{x}}{\mathsf{A}}{\mathsf{s}}) : \mathsf{A} \Rightarrow \mathsf{B}$ is derived by applying the \textsf{ABS} rule to
    \begin{equation*}
      \mathsf{\Gamma}, \mathsf{x}:\mathsf{A} \vdash (\mathsf{t}, a, \mathsf{s}) : \mathsf{B}
    \end{equation*}
    and
    \begin{equation*}
      b(\gamma,\xi)(x,c) = a(\gamma \add x, \xi \add c).
    \end{equation*}
    By the induction hypothesis, we obtain
    \begin{multline*}
      (\sem{\lambda \mathsf{\Gamma}.\, \lam{\mathsf{x}}{\mathsf{A}}{\mathsf{t}}},
      \lambda \mathsf{\Gamma}.\,b,
      \sem{\lambda \mathsf{\Gamma}.\, \lam{\mathsf{x}}{\mathsf{A}}{\mathsf{s}}})
      = \\
      (\sem{\lambda (\mathsf{\Gamma},\mathsf{x}:\mathsf{A}).\, \mathsf{t}},
      \lambda (\mathsf{\Gamma},\mathsf{x}:\mathsf{A}).\, a,
      \sem{\lambda (\mathsf{\Gamma},\mathsf{x}:\mathsf{A}).\, \mathsf{s}})
      \in R_{\mathsf{\Gamma} \Rightarrow \mathsf{A} \Rightarrow \mathsf{B}}.
    \end{multline*}
  \item
    Suppose that $\mathsf{\Gamma} \vdash (\mathsf{t}\,\mathsf{p}, c, \mathsf{s}\,\mathsf{q}) : \mathsf{B}$ is derived by applying the \textsf{APP} rule to
    \begin{equation*}
      \mathsf{\Gamma} \vdash (\mathsf{t}, a, \mathsf{s}) : \mathsf{A} \Rightarrow \mathsf{B},
      \qquad
      \mathsf{\Gamma} \vdash (\mathsf{p}, b, \mathsf{q}) : \mathsf{A}
    \end{equation*}
    and
    \begin{equation*}
      c(\gamma,\xi) =
      a(\gamma,\xi)(\semp{\mathsf{p}}{\gamma}, b(\gamma,\xi)).
    \end{equation*}
    For simplicity, we assume that $\mathsf{\Gamma} = (\mathsf{x} : \mathsf{C})$. By the induction hypothesis, we have
    \begin{equation*}
      (\sem{\lam{\mathsf{x}}{\mathsf{C}}{\mathsf{t}}},
      a,
      \sem{\lam{\mathsf{x}}{\mathsf{C}}{\mathsf{s}}})
      \in R_{\mathsf{C} \Rightarrow \mathsf{A} \Rightarrow \mathsf{B}},
      \qquad
      (\sem{\lam{\mathsf{x}}{\mathsf{C}}{\mathsf{p}}},
      b,
      \sem{\lam{\mathsf{x}}{\mathsf{C}}{\mathsf{q}}})
      \in R_{\mathsf{C} \Rightarrow \mathsf{A}}.
    \end{equation*}
    Let $\mathsf{y} : \mathsf{C} \Rightarrow \mathsf{A} \Rightarrow \mathsf{B},
    \mathsf{z} : \mathsf{C} \Rightarrow \mathsf{A} \vdash \mathsf{h} : \mathsf{C} \Rightarrow \mathsf{B}$ be
    \begin{equation*}
      \lam{\mathsf{x}}{\mathsf{C}}{\mathsf{y} \, \mathsf{x} \, (\mathsf{z} \, \mathsf{x})}.
    \end{equation*}
    Its derivative $\fsem{\mathsf{h}}$ is given by
    \begin{multline*}
      \fsemp{\mathsf{h}}{(y,z)}{(d,e)}(x, f)=
      \fsemp{\mathsf{y} \, \mathsf{x} \, (\mathsf{z} \, \mathsf{x})}{(y,z,x)}{(d,e,f)}
      = \\
      \fsemp{\mathsf{y} \, \mathsf{x}}{(y,z,x)}{(d,e,f)}(
      \semp{\mathsf{z} \, \mathsf{x}}{(y,z,x)},
      \fsemp{\mathsf{z} \, \mathsf{x}}{(y,z,x)}{(d,e,f)}
      )
      =
      d(x, f)(zx, e(x, f)).
    \end{multline*}
    It follows from the fundamental property that
    \begin{equation*}
      (
      \semp{\mathsf{h}}{
        (\sem{\lam{\mathsf{x}}{\mathsf{C}}{\mathsf{t}}},\sem{\lam{\mathsf{x}}{\mathsf{C}}{\mathsf{p}}})},
      \fsemp{\mathsf{h}}{
        (\sem{\lam{\mathsf{x}}{\mathsf{C}}{\mathsf{t}}},\sem{\lam{\mathsf{x}}{\mathsf{C}}{\mathsf{p}}})}{
        (a,b)
      },
      \semp{\mathsf{h}}{
        (\sem{\lam{\mathsf{x}}{\mathsf{C}}{\mathsf{s}}},\sem{\lam{\mathsf{x}}{\mathsf{C}}{\mathsf{q}}})}
      ) \in R_{\mathsf{C} \Rightarrow \mathsf{B}}.
    \end{equation*}
    Since
    \begin{align*}
      \semp{\mathsf{h}}{
      (\sem{\lam{\mathsf{x}}{\mathsf{C}}{\mathsf{t}}},\sem{\lam{\mathsf{x}}{\mathsf{C}}{\mathsf{p}}})}
      &= \sem{\lam{\mathsf{x}}{\mathsf{C}}{\mathsf{t} \, \mathsf{p}}},
      \\
      \semp{\mathsf{h}}{
      (\sem{\lam{\mathsf{x}}{\mathsf{C}}{\mathsf{s}}},\sem{\lam{\mathsf{x}}{\mathsf{C}}{\mathsf{q}}})}
      &= \sem{\lam{\mathsf{x}}{\mathsf{C}}{\mathsf{s} \, \mathsf{q}}},
      \\
      \fsemp{\mathsf{h}}{
        (\sem{\lam{\mathsf{x}}{\mathsf{C}}{\mathsf{t}}},\sem{\lam{\mathsf{x}}{\mathsf{C}}{\mathsf{p}}})}{
        (a,b)}
      &= c,
    \end{align*}
    we conclude that $(\sem{\lam{\mathsf{x}}{\mathsf{C}}{\mathsf{t}\,\mathsf{p}}}, c, \sem{\lam{\mathsf{x}}{\mathsf{C}}{\mathsf{s}\,\mathsf{q}}}) \in R_{\mathsf{C} \Rightarrow \mathsf{B}}$.
  \end{itemize}
\end{proof}
We next show that $\theta$ is a differential prelogical relation. The fundamental property of $\theta$ follows from the following two lemmas.
\begin{lem}\label{lem:qrefl-theta}
  For every $\mathsf{\Gamma} \vdash \mathsf{t} : \mathsf{A}$,
  \begin{equation*}
    \mathsf{\Gamma} \vdash
    (\mathsf{t}, \fsem{\mathsf{t}}, \mathsf{t})
    : \mathsf{A}.
  \end{equation*}
\end{lem}
\begin{proof}
  We prove this claim by induction on the derivation of $\mathsf{\Gamma} \vdash \mathsf{t} : \mathsf{A}$. Here, we only prove the cases $\mathsf{\Gamma} \vdash \lam{\mathsf{x}}{\mathsf{A}}{\mathsf{t}} : \mathsf{A} \Rightarrow \mathsf{B}$ and $\mathsf{\Gamma} \vdash \mathsf{t} \, \mathsf{s} : \mathsf{B}$. For the first case, by the induction hypothesis, we have
  \begin{equation*}
    \mathsf{\Gamma}, \mathsf{x} : \mathsf{A} \vdash
    (\mathsf{t}, \fsem{\mathsf{t}}, \mathsf{t}) : \mathsf{B}.
  \end{equation*}
  It follows from the definition of derivatives that $\fsem{\lam{\mathsf{x}}{\mathsf{A}}{\mathsf{t}}}$ satisfies
  \begin{equation*}
    \fsemp{\lam{\mathsf{x}}{\mathsf{A}}{\mathsf{t}}}{\gamma}{\xi}(x, a) = \fsemp{\mathsf{t}}{\gamma \add x}{\xi \add a}.
  \end{equation*}
  Hence, by the \textsf{ABS} rule, we obtain
  \begin{equation*}
    \mathsf{\Gamma} \vdash (\lam{\mathsf{x}}{\mathsf{A}}{\mathsf{t}},
    \fsem{\lam{\mathsf{x}}{\mathsf{A}}{\mathsf{t}}}, \lam{\mathsf{x}}{\mathsf{A}}{\mathsf{t}}) : \mathsf{A} \Rightarrow \mathsf{B}.
  \end{equation*}
  For the second case, by the induction hypothesis, we have
  \begin{equation*}
    \mathsf{\Gamma} \vdash (\mathsf{t}, \fsem{\mathsf{t}}, \mathsf{t}) : \mathsf{A} \Rightarrow \mathsf{B},
    \qquad
    \mathsf{\Gamma} \vdash (\mathsf{s}, \fsem{\mathsf{s}}, \mathsf{s}) : \mathsf{A}.
  \end{equation*}
  Since
  \begin{equation*}
    \fsemp{\mathsf{t}\,\mathsf{s}}{\gamma}{\xi}
    = \fsemp{\mathsf{t}}{\gamma}{\xi}(\semp{\mathsf{s}}{\gamma}, \fsemp{\mathsf{s}}{\gamma}{\xi}),
  \end{equation*}
  by the \textsf{APP} rule, we obtain
  \begin{equation*}
    \mathsf{\Gamma} \vdash (\mathsf{t} \, \mathsf{s}, \fsem{\mathsf{t} \, \mathsf{s}}, \mathsf{t} \, \mathsf{s})
    : \mathsf{B}.
  \end{equation*}
  The remaining cases can be checked similarly using the corresponding derivation rules.
\end{proof}

\begin{lem}\label{lem:substitute}
  For every typing environment $\mathsf{\Gamma} = (\mathsf{x}_{1} : \mathsf{A}_{1},\ldots,\mathsf{x}_{n} : \mathsf{A}_{n})$, every $i \in \{1,\ldots,n\}$, and for all judgements $\mathsf{\Gamma} \vdash (\mathsf{t}, a, \mathsf{s}) : \mathsf{B}$ and $\vdash (\mathsf{p}, b, \mathsf{q}) : \mathsf{A}_{i}$, we have
  \begin{equation*}
    \mathsf{\Gamma} \setminus i \vdash
    (\mathsf{t}[\mathsf{p}/\mathsf{x}_{i}],
    a',
    \mathsf{s}[\mathsf{q}/\mathsf{x}_{i}])
    : \mathsf{B},
    \qquad
    \mathsf{\Gamma} \setminus i \vdash
    (\mathsf{t}[\mathsf{p}/\mathsf{x}_{i}],
    a',
    \mathsf{t}[\mathsf{q}/\mathsf{x}_{i}])
    : \mathsf{B},
  \end{equation*}
  where $\mathsf{\Gamma} \setminus i$ is the typing environment obtained by removing the $i$-th element of $\mathsf{\Gamma}$, and
  \begin{equation*}
    a' \colon (|\mathsf{\Gamma} \setminus i|, \mathcal{Q}_{\mathsf{\Gamma} \setminus i}) \rightarrowtriangle \mathcal{Q}_{\mathsf{B}}
  \end{equation*}
  is the function given as follows: for $\gamma \in \uset{\mathsf{\Gamma} \setminus i}$ and $\xi \in \mathcal{Q}_{\mathsf{\Gamma} \setminus i}$,
  \begin{equation*}
    a'(\gamma, \xi) =
    a((\gamma|_{1},\ldots,\gamma|_{i-1},\sem{\mathsf{p}},\gamma|_{i},\ldots,\gamma|_{n - 1}),
    (\xi|_{1},\ldots,\xi|_{i-1},b,\xi|_{i},\ldots,\xi|_{n - 1})).
  \end{equation*}
\end{lem}
\begin{proof}
  We can prove the claim by induction on the derivation of $\mathsf{\Gamma} \vdash (\mathsf{t}, a, \mathsf{s}) : \mathsf{B}$. As an illustration, we prove the case where $\mathsf{\Gamma} \vdash (\lam{\mathsf{x}}{\mathsf{A}}{\mathsf{t}}, b, \lam{\mathsf{x}}{\mathsf{A}}{\mathsf{s}}) : \mathsf{A} \Rightarrow \mathsf{B}$ is derived by applying the \textsf{ABS} rule to
  \begin{equation*}
    \mathsf{\Gamma} , \mathsf{x} : \mathsf{A} \vdash  (\mathsf{t}, a, \mathsf{s}) : \mathsf{B}
  \end{equation*}
  and
  \begin{equation*}
    b(\gamma,\xi)(x, c) = a(\gamma \add x,\xi \add c).
  \end{equation*}
  Let $\mathsf{\Gamma}$ be $(\mathsf{x}_{1}:\mathsf{A}_{1},\ldots,\mathsf{x}_{n}:\mathsf{A}_{n})$, and let $i$ be an element of $\{1,\ldots, n\}$. Our goal is to prove that for every $\vdash (\mathsf{p}, d, \mathsf{q}) : \mathsf{A}_{i}$,
  \begin{align*}
    & \mathsf{\Gamma} \setminus i \vdash
      (\lam{\mathsf{x}}{\mathsf{A}}{\mathsf{t}}[\mathsf{p}/\mathsf{x}_{i}], b', \lam{\mathsf{x}}{\mathsf{A}}{\mathsf{s}}[\mathsf{q}/\mathsf{x}_{i}]) : \mathsf{A} \Rightarrow \mathsf{B},
    \\
    & \mathsf{\Gamma} \setminus i \vdash
      (\lam{\mathsf{x}}{\mathsf{A}}{\mathsf{t}}[\mathsf{p}/\mathsf{x}_{i}], b', \lam{\mathsf{x}}{\mathsf{A}}{\mathsf{t}}[\mathsf{q}/\mathsf{x}_{i}]) : \mathsf{A} \Rightarrow \mathsf{B},
  \end{align*}
  where $b' \colon (|\mathsf{\Gamma} \setminus i|, \mathcal{Q}_{\mathsf{\Gamma} \setminus i}) \rightarrowtriangle \mathcal{Q}_{\mathsf{A} \Rightarrow \mathsf{B}}$ is given by
  \begin{equation*}
    b'(\gamma, \xi) =
    b((\gamma|_{1},\ldots,\gamma|_{i-1},\sem{\mathsf{p}},\gamma|_{i},\ldots,\gamma|_{n - 1}),
    (\xi|_{1},\ldots,\xi|_{i-1},d,\xi|_{i},\ldots,\xi|_{n - 1})).
  \end{equation*}
  It follows from the induction hypothesis that we have
  \begin{align*}
    & \mathsf{\Gamma} \setminus i, \mathsf{x} : \mathsf{A} \vdash
      (\mathsf{t}[\mathsf{p}/\mathsf{x}_{i}],
      a',
      \mathsf{s}[\mathsf{q}/\mathsf{x}_{i}]) : \mathsf{B},
    \\
    & \mathsf{\Gamma} \setminus i, \mathsf{x} : \mathsf{A} \vdash
      (\mathsf{t}[\mathsf{p}/\mathsf{x}_{i}],
      a',
      \mathsf{t}[\mathsf{q}/\mathsf{x}_{i}]) : \mathsf{B},
  \end{align*}
  where $a' \colon (|\mathsf{\Gamma} \setminus i, \mathsf{x} : \mathsf{A}|, \mathcal{Q}_{\mathsf{\Gamma} \setminus i, \mathsf{x} : \mathsf{A}}) \rightarrowtriangle \mathcal{Q}_{\mathsf{B}}$ is given by
  \begin{equation*}
    a'(\gamma \add x, \xi \add c)
    =
    a((\gamma|_{1},\ldots,\gamma|_{i-1},\sem{\mathsf{p}},\gamma|_{i},\ldots,\gamma|_{n-1},x),
    (\xi|_{1},\ldots,\xi|_{i-1},d,\xi|_{i},\ldots,\xi|_{n-1},c)).
  \end{equation*}
  Then, by the \textsf{ABS} rule, we obtain
  \begin{align*}
    & \mathsf{\Gamma} \setminus i \vdash
      (\lam{\mathsf{x}}{\mathsf{A}}{\mathsf{t}[\mathsf{p}/\mathsf{x}_{i}]},
      b',
      \lam{\mathsf{x}}{\mathsf{A}}{\mathsf{s}[\mathsf{q}/\mathsf{x}_{i}]}) : \mathsf{A} \Rightarrow \mathsf{B},
    \\
    & \mathsf{\Gamma} \setminus i \vdash
      (\lam{\mathsf{x}}{\mathsf{A}}{\mathsf{t}[\mathsf{p}/\mathsf{x}_{i}]},
      b',
      \lam{\mathsf{x}}{\mathsf{A}}{\mathsf{t}[\mathsf{q}/\mathsf{x}_{i}]}) : \mathsf{A} \Rightarrow \mathsf{B}.
  \end{align*}
  We can similarly check the other cases.
\end{proof}

From the two lemmas, we deduce that $\theta$ satisfies the fundamental property and is a differential prelogical relation.
\begin{prop}\label{prop:thetafund}
  The fundamental property of $\theta$ holds.
\end{prop}
\begin{proof}
  Let $\mathsf{\Gamma} = (\mathsf{x}_{1}:\mathsf{A}_{1},\ldots,\mathsf{x}_{n}:\mathsf{A}_{n})$ be a typing environment. For every $\mathsf{\Gamma} \vdash \mathsf{t} : \mathsf{A}$ and every $(\gamma, \xi, \varepsilon) \in \theta_{\mathsf{\Gamma}}$, we show that $(\semp{\mathsf{t}}{\gamma}, \fsemp{\mathsf{t}}{\gamma}{\xi}, \semp{\mathsf{t}}{\varepsilon}) \in \theta_{\mathsf{A}}$. By Lemma~\ref{lem:qrefl-theta}, we have $\mathsf{\Gamma} \vdash (\mathsf{t}, \fsem{\mathsf{t}}, \mathsf{t}) : \mathsf{A}$.
  By the definition of $\theta$, for every $i \in \{1,\ldots,n\}$, we have $\vdash (\mathsf{s}_{i}, \xi|_{i}, \mathsf{u}_{i}) : \mathsf{A}_{i}$ for some closed terms $\mathsf{s}_{i}$ and $\mathsf{u}_{i}$ such that $\sem{\mathsf{s}_{i}} = \gamma|_{i}$ and $\sem{\mathsf{u}_{i}} = \varepsilon|_{i}$. Then, by repeatedly applying Lemma~\ref{lem:substitute}, we obtain
  \begin{equation*}
    \vdash
    (\mathsf{t}[\mathsf{s}_{1}/\mathsf{x}_{1},\ldots,\mathsf{s}_{n}/\mathsf{x}_{n}],
    \fsemp{\mathsf{t}}{(\sem{\mathsf{s}_{1}},\ldots,\sem{\mathsf{s}_{n}})}{\xi},
    \mathsf{t}[\mathsf{u}_{1}/\mathsf{x}_{1},\ldots,\mathsf{u}_{n}/\mathsf{x}_{n}])
    : \mathsf{A}.
  \end{equation*}
 Since $\sem{\mathsf{t}[\mathsf{s}_{1}/\mathsf{x}_{1},\ldots,\mathsf{s}_{n}/\mathsf{x}_{n}]} = \semp{\mathsf{t}}{\gamma}$, and $\sem{\mathsf{t}[\mathsf{u}_{1}/\mathsf{x}_{1},\ldots,\mathsf{u}_{n}/\mathsf{x}_{n}]} = \semp{\mathsf{t}}{\varepsilon}$, we obtain the claim.
\end{proof}
\begin{prop}\label{prop:theta_dpr}
  The family of ternary relations $\theta$ is a differential prelogical relation.
\end{prop}
\begin{proof}
  We have already checked that the fundamental property of $\theta$ holds in Proposition~\ref{prop:thetafund}. For each type $\mathsf{A}$, it follows from the last five rules in Figure~\ref{fig:eqt} that $\theta_{\mathsf{A}}$ is a closed quasi-quasi-metric on $\dfn{\mathsf{A}}$. Here, we use the \textsf{SEM} rule to derive that $\theta_{\mathsf{A}}$ is transitive and closed under least upper bounds. Below, we only give the proof of transitivity so as to illustrate how we use the \textsf{SEM} rule. Suppose that
  \begin{equation*}
    (x, a, y) \in \theta_{\mathsf{A}},
    \qquad
    (y, b, z) \in \theta_{\mathsf{A}}.
  \end{equation*}
  Then, by the definition of $\theta$, there are judgements
  \begin{equation*}
    \vdash (\mathsf{t}, a, \mathsf{s}) : \mathsf{A},
    \qquad
    \vdash (\mathsf{p}, b, \mathsf{q}) : \mathsf{A}
  \end{equation*}
  such that $x = \sem{\mathsf{t}}$, $y = \sem{\mathsf{s}} = \sem{\mathsf{p}}$ and $z = \sem{\mathsf{q}}$. Then, by the \textsf{SEM} rule, we obtain
  \begin{equation*}
    \vdash (\mathsf{s}, b, \mathsf{q}) : \mathsf{A}.
  \end{equation*}
  Now, by applying the \textsf{TRAN} rule, we obtain $\vdash (\mathsf{t}, a \otimes b, \mathsf{q}) : \mathsf{A}$. Hence, $(x, a \otimes b, z) \in \theta_{\mathsf{A}}$. Finally, we check that $(x, a, y) \in \theta_{\real}$ if and only if $|x - y| \lawsqsupseteq a$. It follows from the \textsf{REAL} rule that $(x, |x - y|, y) \in \theta_{\real}$. Then, by the \textsf{DOWN} rule, if $|x - y| \lawsqsupseteq a$, then $(x, a, y) \in \theta_{\real}$. The other implication follows from Proposition~\ref{prop:theta_least}. Indeed, since $\theta_{\real} \subseteq \delta_{\real}$, if $(x, a, y) \in \theta_{\real}$, then $(x, a, y) \in \delta_{\real}$. Hence, $|x - y| \lawsqsupseteq a$.
\end{proof}
By Proposition~\ref{prop:theta_least} and Proposition~\ref{prop:theta_dpr}, we see that $\theta$ is the least differential prelogical relation.
\begin{thm}
  For every type $\mathsf{A}$, we have $\theta_{\mathsf{A}} = \nu_{\mathsf{A}}$.
\end{thm}

\subsection{Absence of a Greatest Differential Prelogical Relation}
\label{sec:absence-great-elem}

The goal of this section is to show that when $\mathbb{F}$ is empty, the poset $\mathfrak{R}$ does not necessarily have a greatest differential prelogical relation. Recall that $\mathbb{F}$ is the set of multi-arity functions on $\mathbb{R}$ in our target calculus $\LL$ (see Section~\ref{sec:target-language-its}). In the rest of this section, we assume that $\mathbb{F}$ is empty and show that $\mathfrak{R}$ lacks a greatest differential prelogical relation.
\begin{lem}\label{lem:cid}
  If $\mathbb{F}$ is empty, then for every term $\vdash \mathsf{t} : \real \Rightarrow \real$, the set theoretic interpretation $\sem{\mathsf{t}}$ is equal to either a constant function or the identity function on $\mathbb{R}$.
\end{lem}
\begin{proof}
  For a binary relation $u \subseteq \mathbb{R} \times \mathbb{R}$, we define a type-indexed family of binary relations $u^{\#}_{\mathsf{A}} \subseteq \uset{\mathsf{A}} \times \uset{\mathsf{A}}$ by
  \begin{align*}
    u^{\#}_{\real}
    &= u,
    \\
    u^{\#}_{\mathsf{A} \times \mathsf{B}}
    &= \{((x, y), (z, w)) \in
      \uset{\mathsf{A} \times \mathsf{B}}
      \times \uset{\mathsf{A} \times \mathsf{B}}
      \mid
      (x, z) \in u^{\#}_{\mathsf{A}} \text{ and }
      (y, w) \in u^{\#}_{\mathsf{B}}\},
    \\
    u^{\#}_{\mathsf{A} \Rightarrow \mathsf{B}}
    &= \{(f, g) \in \uset{\mathsf{A} \Rightarrow \mathsf{B}}
      \times \uset{\mathsf{A} \Rightarrow \mathsf{B}}
      \mid
      \forall (x, y) \in u^{\#}_{\mathsf{A}},\,
      (fx, gy) \in u^{\#}_{\mathsf{B}}\}.
  \end{align*}
  Then, for a typing environment $\mathsf{\Gamma} = (\mathsf{x}_{1} : \mathsf{A}_{1}, \ldots, \mathsf{x}_{n} : \mathsf{A}_{n})$, we define a binary relation $u^{\#}_{\mathsf{\Gamma}} \subseteq \uset{\mathsf{\Gamma}} \times \uset{\mathsf{\Gamma}}$ by $(\gamma,\varepsilon) \in u^{\#}_{\mathsf{\Gamma}}$ if and only if for all $i \in \{1,\ldots,n\}$, we have $(\gamma|_{i}, \varepsilon|_{i}) \in u^{\#}_{\mathsf{A}_{i}}$. We show that the fundamental lemma for $u^{\#}$ holds whenever $u$ is reflexive: for every reflexive $u \subseteq \mathbb{R} \times \mathbb{R}$, every $\mathsf{\Gamma} \vdash \mathsf{t} : \mathsf{A}$, and every $(\gamma,\varepsilon) \in u^{\#}_{\mathsf{\Gamma}}$, we have $(\semp{\mathsf{t}}{\gamma}, \semp{\mathsf{t}}{\varepsilon}) \in u^{\#}_{\mathsf{A}}$. We deduce the fundamental lemma for $u^{\#}$ by induction on $\mathsf{\Gamma} \vdash \mathsf{t} : \mathsf{A}$. Here, we only check the following three cases:
  \begin{itemize}
  \item
    Consider $\mathsf{\Gamma} \vdash \underline{r} : \real$. Since $u$ is reflexive, we have $(r, r) \in u^{\#}_{\real}$.
  \item
    Consider $\mathsf{\Gamma} \vdash \lam{\mathsf{x}}{\mathsf{A}}{\mathsf{t}} : \mathsf{A} \Rightarrow \mathsf{B}$. By the induction hypothesis, for every $(\gamma,\varepsilon) \in u^{\#}_{\mathsf{\Gamma}}$ and every $(x, y) \in u^{\#}_{\mathsf{A}}$, we have
    \begin{equation*}
      (\semp{\mathsf{t}}{\gamma \add x}, \semp{\mathsf{t}}{\varepsilon \add y}) \in u^{\#}_{\mathsf{B}}.
    \end{equation*}
    Since $(\semp{\mathsf{t}}{\gamma \add x}, \semp{\mathsf{t}}{\varepsilon \add y}) = (\semp{\lam{\mathsf{x}}{\mathsf{A}}{\mathsf{t}}}{\gamma}(x),\semp{\lam{\mathsf{x}}{\mathsf{A}}{\mathsf{t}}}{\varepsilon}(y))$, we obtain
    \begin{equation*}
      (\semp{\lam{\mathsf{x}}{\mathsf{A}}{\mathsf{t}}}{\gamma}(x),\semp{\lam{\mathsf{x}}{\mathsf{A}}{\mathsf{t}}}{\varepsilon}(y)) \in u^{\#}_{\mathsf{B}},
      \qquad \text{for every $(x, y) \in u^{\#}_{\mathsf{A}}$}.
    \end{equation*}
    Therefore,
    \begin{equation*}
      (\semp{\lam{\mathsf{x}}{\mathsf{A}}{\mathsf{t}}}{\gamma}, \semp{\lam{\mathsf{x}}{\mathsf{A}}{\mathsf{t}}}{\varepsilon}) \in u^{\#}_{\mathsf{A} \Rightarrow \mathsf{B}}.
    \end{equation*}
  \item
    Consider $\mathsf{\Gamma} \vdash \mathsf{t} \, \mathsf{s} : \mathsf{B}$. By the induction hypothesis, for every $(\gamma,\varepsilon) \in u^{\#}_{\mathsf{\Gamma}}$, we have
    \begin{equation*}
      (\semp{\mathsf{t}}{\gamma},\semp{\mathsf{t}}{\varepsilon}) \in u^{\#}_{\mathsf{A} \Rightarrow \mathsf{B}},
      \qquad
      (\semp{\mathsf{s}}{\gamma},\semp{\mathsf{s}}{\varepsilon}) \in u^{\#}_{\mathsf{A}}.
    \end{equation*}
    Therefore, by the definition of $u^{\#}_{\mathsf{A} \Rightarrow \mathsf{B}}$, we obtain
    \begin{equation*}
      (\semp{\mathsf{t} \, \mathsf{s}}{\gamma}, \semp{\mathsf{t} \, \mathsf{s}}{\varepsilon})
      = (\semp{\mathsf{t}}{\gamma}(\semp{\mathsf{s}}{\gamma}),
      \semp{\mathsf{t}}{\varepsilon}(\semp{\mathsf{s}}{\varepsilon})) \in u^{\#}_{\mathsf{B}}.
    \end{equation*}
  \end{itemize}
  Now, we observe that for every closed term $\vdash \mathsf{t} : \real \Rightarrow \real$, if $\sem{\mathsf{t}}$ is a non-constant function, then $\sem{\mathsf{t}}$ is the identity function. Suppose that $\sem{\mathsf{t}}$ is a non-constant function, and fix arbitrary $r \in \mathbb{R}$. We show that $\sem{\mathsf{t}}(r) = r$. Since $\sem{\mathsf{t}}$ is non-constant, there is $s \in \mathbb{R}$ such that $\sem{\mathsf{t}}(r) \neq \sem{\mathsf{t}}(s)$. We note that $r \neq s$. Let $u \subseteq \mathbb{R} \times \mathbb{R}$ be a reflexive binary relation given by
  \begin{equation*}
    (x, y) \in u \iff x = y \text{ or } (x, y) = (r, s).
  \end{equation*}
  Since $(r, s) \in u$, the fundamental lemma for $u^{\#}$ implies that $(\sem{\mathsf{t}}(r), \sem{\mathsf{t}}(s)) \in u$. Because $(r, s)$ is the unique pair in $u$ with distinct components, we see that $(\sem{\mathsf{t}}(r), \sem{\mathsf{t}}(s)) = (r, s)$. Thus, $\sem{\mathsf{t}}(r) = r$. Since $r$ is an arbitrary real, $\sem{\mathsf{t}}$ is the identity function.
\end{proof}

We proceed to prove that $\mathfrak{R}$ lacks a greatest differential prelogical relation. We first show that the differential prelogical relation $\delta$ given in Example~\ref{eg:LDLR} is maximal in $\mathfrak{R}$.
\begin{prop}\label{prop:maximal}
  For every differential prelogical relation $R$, if $R \supseteq \delta$, then $R = \delta$.
\end{prop}
\begin{proof}
  We prove that $R_{\mathsf{A}}$ is equal to $\delta_{\mathsf{A}}$ by induction on $\mathsf{A}$. The base case $\mathsf{A} = \real$ is trivial. For the case $\mathsf{A} = (\mathsf{B} \times \mathsf{C})$, if $((x, y), (a, b), (z, w)) \in R_{\mathsf{B} \times \mathsf{C}}$, then, by Proposition~\ref{prop:dpr}, $(x, a, z) \in R_{\mathsf{B}} = \delta_{\mathsf{B}}$ and $(y, b, w) \in R_{\mathsf{C}} = \delta_{\mathsf{C}}$. Hence, $((x, y), (a, b), (z, w))$ is an element of $\delta_{\mathsf{B} \times \mathsf{C}}$. For the case $\mathsf{A} = (\mathsf{B} \Rightarrow \mathsf{C})$, if $(f, d, g) \in R_{\mathsf{B} \Rightarrow \mathsf{C}}$, then for every $(x, a, y) \in R_{\mathsf{B}} = \delta_{\mathsf{B}}$, by Proposition~\ref{prop:dpr} and the quasi-reflexivity of $R_{\mathsf{B} \Rightarrow \mathsf{C}}$, we have $(fx, d(x, a), gy) \in R_{\mathsf{C}} = \delta_{\mathsf{C}}$ and $(fx, d(x, a), fy) \in R_{\mathsf{C}} = \delta_{\mathsf{C}}$. Since both $f$ and $g$ are $\LL$-definable, we obtain $(f, d, g) \in \delta_{\mathsf{B} \Rightarrow \mathsf{C}}$.
\end{proof}

We next show that $\delta$ is not a greatest differential prelogical relation by proving $\tau \not\subseteq \delta$ where $\tau$ is the differential prelogical relation obtained by restricting $\{\rho_{\mathsf{A}}\}_{\mathsf{A} \in \type}$ to $\LL$-definable elements (see Example~\ref{eg:QDLR}). A witness of $\tau \not\subseteq \delta$ resides in $((\real \Rightarrow \real) \Rightarrow \real) \Rightarrow \real$. As a preparation, we examine $\rho$ on the type $(\real \Rightarrow \real) \Rightarrow \real$. We define functions $O,Z \colon (\mathbb{R} \Rightarrow \mathbb{R}) \to \mathbb{R}$ by
\begin{equation*}
  O(f) = 0,
  \qquad
  Z(f) = f(0).
\end{equation*}
Both $O$ and $Z$ are $\LL$-definable. In the following lemma, for brevity, we write $D$ for $\hat{\delta}_{(\real \Rightarrow \real) \Rightarrow \real}(O, Z)$.

\begin{lem}\label{lem:ord2}
  For every $F \colon (\mathbb{R} \Rightarrow \mathbb{R}) \to \mathbb{R}$, if
  \begin{equation*}
    (O, D, F) \in \rho_{(\real \Rightarrow \real) \Rightarrow \real},
  \end{equation*}
  then $F = O$.
\end{lem}
\begin{proof}
  Similar to the proof of Proposition~\ref{lem:arrow}, we can show the following equality:
  \begin{equation}
    \label{eq:dfd}
    D(f, d) =
    \lawbigsqcap_{
      (f, d, g) \in
      \delta_{\real \Rightarrow \real}
    }
    |O(f) - Z(g)| =
    \lawbigsqcap_{
      (f, d, g) \in
      \delta_{\real \Rightarrow \real}
    }
    |g(0)|.
  \end{equation}
  Assuming that $(O, D, F) \in \rho_{(\real \Rightarrow \real) \Rightarrow \real}$, we deduce that for every function $f \colon \mathbb{R} \to \mathbb{R}$, we have $F(f) = 0$. We define $h \colon \mathbb{R} \to \mathbb{R}$ by
  \begin{equation*}
    h(x) = 2x.
  \end{equation*}
  Since $(h, \hat{\rho}_{\real \Rightarrow \real}(h, f), f) \in \rho_{\real \Rightarrow \real}$ and $(O,D,F) \in \rho_{(\real \Rightarrow \real) \Rightarrow \real}$, we obtain
  \begin{equation*}
    (O(h), D(h, \hat{\rho}_{\real \Rightarrow \real}(h, f)), F(f))
    \in \rho_{\real}.
  \end{equation*}
  By Lemma~\ref{lem:cid}, $h$ is not $\LL$-definable. Hence, there is no $g \colon \mathbb{R} \to \mathbb{R}$ such that
  \begin{equation*}
    (h, \hat{\rho}_{\real \Rightarrow \real}(h, f), g) \in \delta_{\real \Rightarrow \real}.
  \end{equation*}
  Therefore, $D(h, \hat{\rho}_{\real \Rightarrow \real}(h, f)) = 0$ by \eqref{eq:dfd}. Consequently,
  \begin{equation*}
    (0, 0, F(f)) =
    (O(h), D(h, \hat{\rho}_{\real \Rightarrow \real}(h, f)), F(f))
    \in \rho_{\real}.
  \end{equation*}
  By the definition of $\rho_{\real}$, we conclude that $F(f) = 0$ for every $f \colon \mathbb{R} \to \mathbb{R}$.
\end{proof}

\begin{prop}
  \label{prop:ord3}
  Let $k \colon \mathbb{R} \to \mathbb{R}$ be the constant function given by $k(x) = 1$, and we define $J \colon ((\mathbb{R} \Rightarrow \mathbb{R}) \Rightarrow \mathbb{R}) \to \mathbb{R}$ by $J(F) = F(k)$. Then, we have
  \begin{equation*}
    (J, \sigma_{J}, J) \notin \delta_{((\real \Rightarrow \real) \Rightarrow \real) \Rightarrow \real}.
  \end{equation*}
  Here, $\sigma_{J}$ is the self-distance of $J$ with respect to $\rho_{((\real \Rightarrow \real) \Rightarrow \real) \Rightarrow \real}$ (see Section~\ref{sec:selfdistance} for the definition of self-distance.)
\end{prop}
\begin{proof}
  We suppose $(J, \sigma_{J}, J) \in \delta_{((\real \Rightarrow \real) \Rightarrow \real) \Rightarrow \real}$ and derive a contradiction. It follows from Proposition~\ref{lem:arrow} that we have
  \begin{equation*}
    \sigma_{J}(O, D) =
    \lawbigsqcap_{(O, D, F)
      \in \rho_{(\real \Rightarrow \real) \Rightarrow \real}}
    |JO - JF|
    =
    |JO - JO| = 0
  \end{equation*}
  where the second equality follows from Lemma~\ref{lem:ord2}. Since $(O, D, Z) \in \delta_{(\real \Rightarrow \real) \Rightarrow \real}$, we obtain
  \begin{equation*}
    (0, 0, 1) =
    (JO, \sigma_{J}(O, D), JZ) \in \delta_{\real},
  \end{equation*}
  a contradiction.
\end{proof}
By the definition of the self-distance $\sigma_{J}$ and Proposition~\ref{prop:ord3},
\begin{equation*}
  (J, \sigma_{J}, J) \in \tau_{((\real \Rightarrow \real) \Rightarrow \real) \Rightarrow \real},
  \qquad
  (J, \sigma_{J}, J) \notin \delta_{((\real \Rightarrow \real) \Rightarrow \real) \Rightarrow \real}.
\end{equation*}
Thus, $\delta$ is a maximal differential prelogical relation (Proposition~\ref{prop:maximal}) but is not a greatest differential prelogical relation of $\mathfrak{R}$. Therefore, there is no greatest element in $\mathfrak{R}$.

\subsection{Absence of a Greatest Differential Prelogical Relation, Informally}
\label{sec:absence-great-diff}

We demonstrate how an attempt to prove that $\delta$ is the greatest differential prelogical relation fails, identifying the form of the fundamental property as the obstacle. Assuming $\tau_{\mathsf{A}} \subseteq \delta_{\mathsf{A}}$, we attempt to prove that for every $(\sem{\mathsf{t}}, d, \sem{\mathsf{t}}) \in \tau_{\mathsf{A} \Rightarrow \real}$, we have $(\sem{\mathsf{t}}, d, \sem{\mathsf{t}}) \in \delta_{\mathsf{A} \Rightarrow \real}$. To this end, we must show that for every $(\sem{\mathsf{p}}, a, \sem{\mathsf{q}}) \in \delta_{\mathsf{A}}$, we have
\begin{equation}\label{eq:dpr}
  (\sem{\mathsf{t} \, \mathsf{p}}, d(\sem{\mathsf{p}}, a), \sem{\mathsf{t} \, \mathsf{q}}) \in \delta_{\real}.
\end{equation}
However, the assumptions $\tau_{\mathsf{A}} \subseteq \delta_{\mathsf{A}}$ and $(\sem{\mathsf{t}}, d, \sem{\mathsf{t}}) \in \tau_{\mathsf{A} \Rightarrow \real}$ are insufficient to deduce the claim \eqref{eq:dpr} since $(\sem{\mathsf{p}}, a, \sem{\mathsf{q}}) \in \delta_{\mathsf{A}}$ is not necessarily an element of $\tau_{\mathsf{A}}$. Therefore, we must appeal to the fundamental property. By the fundamental property, we have
\begin{equation*}
  (\sem{\mathsf{t}}, \fsem{\mathsf{t}}, \sem{\mathsf{t}})
  \in \delta_{\mathsf{A} \Rightarrow \real},
\end{equation*}
from which we obtain
\begin{equation*}
  (\sem{\mathsf{t} \, \mathsf{p}}, \fsem{\mathsf{t}}(\sem{\mathsf{p}}, a),
  \sem{\mathsf{t} \, \mathsf{q}}) \in \delta_{\real}.
\end{equation*}
This result is close to our goal: if $\fsem{\mathsf{t}} = \sigma_{\sem{\mathsf{t}}}$, then since $d \sqsubseteq \sigma_{\sem{\mathsf{t}}} = \fsem{\mathsf{t}}$, we can deduce the claim \eqref{eq:dpr}. However, this equality is not necessarily true. This constitutes the obstacle to this proof: the fundamental property only guarantees a construction of a lower bound $\fsem{\mathsf{t}} \sqsubseteq \sigma_{\sem{\mathsf{t}}}$. In fact, it follows from Proposition~\ref{prop:ord3} that we have $\fsem{\mathsf{t}} \sqsubsetneq \sigma_{\sem{\mathsf{t}}}$ where
\begin{equation*}
  \mathsf{t} = \lam{\mathsf{f}}{(\real \Rightarrow \real) \Rightarrow \real}{\mathsf{f} \, (\lam{\mathsf{x}}{\real}{\underline{1}})}.
\end{equation*}

\subsection{Contextual Metric?}
\label{sec:observ-metr}

It is well-known that different notions of program equivalence for a given language can be compared, with one equivalence being \emph{coarser} than another one when it identifies \emph{more} programs than the other. Under this ordering, program equivalences do indeed form a complete lattice where we have the coarsest one given by  the contextual equivalence, and the finest one arising from a syntactic equational theory. The former equips program equivalences with computational meaning, and the latter is a finitary reasoning principle for program equivalences. Furthermore, we can also derive program equivalences from denotational semantics, which lie between the contextual equivalence and the syntactic equational theory. These program equivalences enable us to prove contextual equivalences and disprove syntactic equalities.

Here, we examined the poset $\mathfrak{R}$ to explore whether the situation is similar to the one of program equivalence. As we observed in Section~\ref{sec:concr-descr-least}, we have the least differential prelogical relation as a quantitative counterpart of the syntactic equational theories in the lattices of program equivalences. However, there is no greatest differential prelogical relation in general as shown in Section~\ref{sec:absence-great-elem}. This suggests the absence of a natural ``contextual program metric'', as such a metric should be a greatest differential prelogical relation analogously to how the contextual equivalence is the greatest program equivalence.

How can we bridge the gap between program equivalences and differential prelogical relations? The main technical obstacle towards defining ``contextual program metric'' lies in how we measure program distances. Since distances between programs describe dependence of outputs on inputs, when we measure distances between programs, we should observe relationships between distances between outputs and distances between contexts. Therefore, we should define a notion of distances between contexts \emph{before} we define a contextual metric. However, defining distances between contexts requires metrics on terms that contexts contain, which leads to a circularity.

\section{Related Work}

Differential logical relations for a simply typed language were introduced in \cite{DGY19}, and later extended to languages with monads \cite{DG22}, and related to incremental computing \cite{DLG21}. Moreover, a unified framework for operationally-based logical relations, subsuming differential logical relations, was introduced in \cite{lmcs:11041}. The connections with metric spaces and partial metric spaces have been explored already in \cite{Geoffroy2020, PistoneLICS}, on the one hand providing a series of negative results that motivate the present work, and on the other hand producing a class of metric and partial metric models based on a different relational construction.

The literature on the interpretation of linear or graded lambda-calculi in the category of metric spaces and non-expansive functions is ample \cite{Reed2010, Gaboardi2013, Gaboardi2017, Gavazzo2018, dallagoFSCD23, SB26}. A related approach is that of quantitative algebraic theories \cite{Plotk}, which aims at capturing metrics over algebras via an equational presentation. These have been extended both to quantale-valued metrics \cite{Dahlqvist2023} and to the simply typed (i.e., non-graded) languages \cite{Honsell2022}, although in the last case the non-expansivity condition makes the construction of interesting algebras rather challenging.

The literature on partial metric spaces is vast, as well. Introduced by Matthews \cite{matthews}, they have been largely explored for the metrization of domain theory \cite{Bukatin1997, Schellekens2003, Smyth2006} and, more recently, of $\lambda$-theories \cite{maestracci2025}. An elegant categorical description of partial metrics via the quantaloid of \emph{diagonals} is introduced in \cite{Stubbe2018}. As this construction is obviously related to the notion of quasi-reflexivity and left strong transitivity here considered, it would be interesting to look for analogous categorical descriptions of the quasi-quasi-metrics here introduced.

%%% Local Variables:
%%% mode: LaTeX
%%% TeX-master: "tmp.tex"
%%% End:

% !TEX root = main.tex

\section{Conclusion}
\label{sec:conclusion}

In this paper we have explored the connections between the notions of program distance arising from differential logical relations and those defined via quasi-quasi-metric spaces. Our results in Section~\ref{sec:quasi-quasi-metric} and Section~\ref{sec:category-quasi-quasi} provide a conceptual bridge between differential logical relations in the literature and various formulations of quantale-valued metrics. This bridge could be used to exploit methods and results from the vast area of research on quantale-valued relations \cite{Hofmann2014, Stubbe2014} for the study of program distances in higher-order programming languages. For instance, natural directions are the characterization of limits and, more generally, of topological properties via logical relations, as suggested by recent work \cite{BCDG22}, although in a qualitative setting. At the same time, as discussed in Section~\ref{sec:towards-lattice-mmms}, our approach based on quasi-quasi-metric spaces suggests natural questions concerning the definition of a mathematically reasonable class of ``metrics'' on the lambda calculus. Namely, is there a class of ``metrics'' on the lambda calculus where we have a greatest and a least ``metrics'' equipped with computational meanings? Our argument in Section~\ref{sec:absence-great-diff} points out that the current form of the fundamental lemma is too weak. In fact, the fundamental lemma (Theorem~\ref{thm:fundamental-lemma}) only provides the construction of derivatives as approximations of self-distances, and we have no mathematical characterization of derivatives.

% The reason that we define $\qqm$ as the category of quasi-quasi-metric spaces rather than those satisfying these three additional properties is their dependecy on self-distances. Indeed, we use self-distances in the definition of these three properties. The dependency on self-distances is an obstcle to describe these three properties purely in terms of $Q$-predicates and adopt them into a quantitative equational theory $\Theta_{\real}$ given in Section~\ref{sec:concr-descr-least}.

There are several further directions for this work. While we here focused on non-symmetric differential (pre)logical relations, understanding the metric structure of the strictly symmetric case, as in \cite{DGY19}, would be interesting as well. Notice that this would require to abandon quasi-reflexivity, cf.~Remark~\ref{rem:symmetry}. In Section~\ref{sec:relating-quasi-quasi-metric}, we observed that quasi-quasi-metrics arising from the interpretation of $\LL$ satisfy several additional properties, namely, self-indistancy, left strong transitivity and weak symmetry. It would be worth examining whether these three properties induce topological or order theoretic structures on quasi-quasi-metric spaces adapting the constructions for partial metric spaces \cite{matthews,Schellekens2003,maestracci2025}. This direction could provide insight into establishing a mathematically reasonable notion of ``metrics'' for the lambda calculus. Finally, while in this paper we only considered simple types, it would be interesting to extend our work based on quasi-quasi-metric spaces to account for other constructions like e.g.~monadic types as in \cite{DG22}. It is thus natural to explore the application of methods arising from quasi-metrics, quasi-quasi-metrics, left strong quasi-quasi-metrics and partial quasi-metrics for the study of languages with effects like e.g.~probabilistic choice.

\bibliographystyle{alphaurl}
\bibliography{main}

\appendix

\end{document}